\documentclass[10pt]{article}

\def\final{0}

\usepackage{csc}

\usepackage{graphicx}
\usepackage{forest}
\usepackage{amsmath}
\usepackage{amssymb}
\usepackage{fullpage}
\usepackage{comment}
\usepackage{commath}
\usepackage{amsthm}
\usepackage{setspace}
\usepackage[margin=1in]{geometry}
\usepackage{float}
\usepackage{framed}
\usepackage{algorithm, algorithmic}
\usepackage{tikz}
\usepackage[backref]{hyperref}
    \definecolor{darkcolor1}{RGB}{15, 43, 109}
    \definecolor{darkcolor2}{RGB}{90, 15, 28}
    \hypersetup{
    colorlinks=true,
    urlcolor=blue,
    linkcolor=darkcolor1,
    citecolor=darkcolor1,
    linktocpage=true
    }
\usepackage{thm-restate}
\usepackage{microtype}

\excludecomment{question}

\def \E {{\mathbb{E}}}

\ifnum\final=0
\newcommand{\mynote}[2]{\marginpar{\color{#1}\tiny\sf #2}}
\newcommand{\mybignote}[2]{{\color{#1} $\langle \langle$ #2$\rangle \rangle$}}
\else
\newcommand{\mynote}[2]{}
\newcommand{\mybignote}[2]{}
\fi

\renewcommand{\vec}[1]{#1}

\newcommand{\uni}{U}

\newcommand{\dsize}{n}
\newcommand{\qmat}{Q}

\newcommand{\queries}{\mathcal{Q}}
\newcommand{\ds}{x}
\newcommand{\elem}{u}
\newcommand{\eqdef}{\coloneqq}
\newcommand{\eps}{\varepsilon}
\DeclareMathOperator{\polylog}{polylog}

\newcommand{\INDSTATE}[1][1]{\STATE\hspace{#1\algorithmicindent}}

\newcommand{\ex}[2]{{\ifx&#1& \mathbb{E} \else \underset{#1}{\mathbb{E}} \fi \left[#2\right]}}
\newcommand{\pr}[2]{{\ifx&#1& \mathbb{P} \else \underset{#1}{\mathbb{P}} \fi \left[#2\right]}}
\newcommand{\var}[2]{{\ifx&#1& \mathsf{Var} \else \underset{#1}{\mathsf{Var}} \fi \left[#2\right]}}

\newcommand{\nope}[1]{}

\newcommand{\zo}{\{0,1\}}

\newcommand{\card}[1]{\left| #1 \right|}

\newcommand{\from}{:}

\renewcommand{\epsilon}{\varepsilon}

\newcommand{\prot}{\rightleftharpoons}

\newcommand{\mat}[1]{\ensuremath{#1}}

\newcommand{\OPT}{\mathit{OPT}}

\newcommand{\cQ}{\mathcal{Q}}
\newcommand{\cR}{\mathcal{R}}

\newcommand{\cY}{\mathcal{Y}}

\newcommand{\Lap}{\mathrm{Lap}}

\newcommand{\QS}{\mat{Q}_{S}}
\newcommand{\hdisc}{\mathrm{hdisc}}
\newcommand{\disc}{\mathrm{disc}}

\newtheorem{thm}{Theorem}[section]
\newtheorem{defn}[thm]{Definition}
\newtheorem{cor}[thm]{Corollary}

\newtheorem{lem}[thm]{Lemma}

\newtheorem{rem}[thm]{Remark}

\newtheorem{theorem}[thm]{Theorem}
\newtheorem{corollary}[thm]{Corollary}
\newtheorem{lemma}[thm]{Lemma}

\newtheorem{definition}[thm]{Definition}

\newtheorem{fact}[thm]{Fact}

\theoremstyle{definition}

\let\originalleft\left
\let\originalright\right
\renewcommand{\left}{\mathopen{}\mathclose\bgroup\originalleft}
\renewcommand{\right}{\aftergroup\egroup\originalright}

\title{Online Matrix Factorization, Online Private Query Release, and Online Discrepancy Minimization}
\author{Aleksandar Nikolov\thanks{Department of Computer Science, University of Toronto. Supported in part by an NSERC Discovery Grant (RGPIN-2021-03206), and the Canada Research Chairs program (CRC-2020-00004).} \and Haohua Tang\thanks{Department of Computer Science, University of Toronto. Supported in part by an NSERC Discovery Grant (RGPIN-2021-03206), and the Canada Research Chairs program (CRC-2020-00004).} \and Jonathan Ullman\thanks{Khoury College of Computer Sciences, Northeastern University.  Part of this work were supported by NSF awards CNS-2232692 and CNS-2247484.}}
\date{}

\begin{document}
\maketitle

\begin{abstract}
    In this paper we consider several related online computation problems. First, we study answering sequences of statistical queries arriving online, and being answered immediately when they arrive with differential privacy. Known matrix factorization mechanisms can answer a set of statistical queries with error bounded by the $\gamma_2$ norm of their query matrix, but require that all queries are known in advance. We show that nearly the same error bounds can be achieved in the online setting for non-adaptively chosen queries. To do so, we give an online factorization algorithm that competitively matches the best offline factorization up to logarithmic factors. In the online matrix factorization problem, a new row $q_t$ of a matrix arrives at each time step $t$, and the algorithm needs to maintain a factorization $L_tR_t=Q_t$ such that at each time it appends some rows to $R_t$, and outputs a new row $\ell_t$ s.t. $\ell_tR_t=q_t$. Our algorithm maintains the competitiveness over this online process, even if the number of rows to arrive is unknown. As another application, we give an online discrepancy minimization algorithm that achieves discrepancy competitive against the $\gamma_2$ norm (and also against hereditary discrepancy) up to logarithmic factors.
\end{abstract}
\vfill\newpage

\tableofcontents
\vfill\newpage

\section{Introduction}

Suppose that a data analyst $A$ is analyzing a dataset $\ds$ that contains potentially sensitive information about $n$ individuals, and we want to enable the analyst to extract useful information about $\ds$ without compromising the privacy of the individuals. Differential privacy~\cite{DworkMNS06} has become the \emph{de facto} standard way to approach this problem. On a high level, differential privacy stipulates that the analyst can only learn information about $\ds$ through an algorithm $M$, and $M$ must satisfy the property that changing any single data point in $\ds$ has only a small effect on the distribution of outputs of $M$. 
One option for the analyst is to specify \emph{all} the queries they want to ask about $\ds$ ahead of time, and send them to $M$, which responds with its answers. We call this the offline, or batch, or non-interactive setting. The offline setting certainly has some advantages. In particular, it allows $M$ to adapt how it computes its private answers to the structure of the queries of $A$, giving more accurate answers to easier sets of queries. On the other hand, specifying all queries ahead of time can be a challenging task for the analyst, who might instead prefer to explore $x$ gradually over time, perhaps by engaging with $M$ through an interactive process. Can the private algorithm $M$ still adapt to the structure of $A$'s queries in this interactive setting, and provide more accurate answers to easier queries? This question is the starting point of this paper. 

\subsection{Private Query Release}

Let us introduce some notation and definitions so that we can discuss the question above more precisely. We model the dataset $\ds$ as a collection $(\ds_1, \ldots, \ds_n)$ of $n$ elements from some data universe $\uni$. For our results, we will need the data universe to be finite, and we will denote its size by $N$. We think of each data point $x_i$ as holding information about one individual. We say two datasets $x$ and $x'$ are neighboring, written $x\sim x'$, if they differ in exactly one data point. Then differential privacy is defined as follows.
\begin{defn}[Differential Privacy \cite{DworkMNS06}] A randomized algorithm $M \from U^n \to \cR$ is \emph{$(\eps, \delta)$-differentially private} if for every pair of neighboring datasets $x \sim x'$ and every $R \subseteq \cR$,
$$
\pr{}{M(x) \in R} \leq e^{\eps} \pr{}{M(x') \in R} + \delta.
$$
\end{defn}
We focus on differentially private algorithms that release answers to statistical queries. A statistical query on $x$ is defined by a function $q:\uni\to \R$, and, overloading notation, equals $q(x) = \frac1n \sum_{i=1}^n q(x_i)$. An important special case are counting queries, i.e., the queries defined by a Boolean function $q:\uni \to \zo$. Counting queries ask what fraction of the dataset satisfies some property, e.g., live in a geographic region and have voted for a candidate in an election. Given a set of counting queries $\mathcal{Q} = (q_1, \ldots, q_m)$, we use $\mathcal{Q}(x) = (q_1(x), \ldots, q_m(x))$ for the vector of answers to the queries. 

Developing private algorithms for answering statistical queries (also known as private query release) has been a cornerstone of research in differential privacy. Many data releases, such as those published by official statistics agencies like the US Census Bureau~\cite{CensusDP17}, take the form of releasing answers to sets of counting queries. Moreover, answering such queries is a useful primitive for other private algorithms~\cite{BlumDMN05}. As a baseline, one of the simplest algorithm that releases answers to a set of queries $\queries$ is the Gaussian mechanism: we release $\queries(x) + Z$, where $Z$ is a vector of independent mean $0$ Gaussian random variables, each with standard deviation proportional to the sensitivity of $\mathcal{Q}$, defined as $\Delta\mathcal{Q} = \max_{x\sim x'} \norm{\mathcal{Q}(x)-\mathcal{Q}(x')}_2$. Note that $\Delta\mathcal{Q}$ is just $\frac{1}{n} \max_{\elem,\elem' \in \uni} \|\queries(\elem) - \queries(\elem')\|_2$, where $\queries(\elem) = (q_1(\elem), \ldots, q_m(\elem))$. The Gaussian mechanism works both in the offline and in the online setting, even when the queries are chosen adaptively by the analyst.

\subsubsection{Offline Private Query Release via Factorization}

As mentioned above, in the offline setting, when the analyst asks a set of statistical queries $\mathcal{Q}$ that is fully specified in advance, the private algorithm can tailor its answers to $\mathcal Q$. As an important example, let us consider data-oblivious, i.e., noise-adding algorithms, which release $\mathcal{Q}(x) + \frac{1}{\dsize} Z$ for some $k$-dimensional random noise vector $Z$ that depends only on $\mathcal Q$ but not on $\ds$. The Gaussian noise mechanism is a special case, where $Z$ has independent Gaussian coordinates. More generally, factorization mechanisms~\cite{EdmondsNU20}, defined more precisely below, choose $Z$ to be a mean $0$ multivariate Gaussian random vector, and optimize the covariance matrix of $Z$ so that some error measure is minimized and the output $\mathcal{Q}(x) + \frac{1}{\dsize} Z$ is private. Adding correlated noise can often lead to much smaller error than adding independent noise to each query. The benefits of correlating noise have also been observed for threshold queries~\cite{ChanSS11,DworkNPR10}, higher-dimensional range queries~\cite{XiaoWG11}, for halfspace queries and other queries with low VC dimension~\cite{MuthukrishnanN12}, for queries arising in streaming algorithms~\cite{LebedaT24}, and even for arbitrary workloads of $k$ counting queries~\cite{Lebeda25}. Moreover, in some settings factorization mechanisms achieve nearly optimal error among all differentially private algorithms (once we fix the privacy parameters and dataset size)~\cite{EdmondsNU20}. Existing factorization mechanisms~\cite{LiHRMM10,EdmondsNU20,HDMM,NikolovT24} for general sets of counting queries, however, require full knowledge of $\mathcal{Q}$ so that they can compute an optimal choice of covariance for the noise so that privacy is preserved and the noise is minimized. It is therefore, not clear if the error bounds achieved by factorization mechanisms can be matched in an interactive setting.

There is an alternative characterization of factorization mechanisms, for which it helps to introduce some linear algebraic notation. For an ordered set of queries $\mathcal{Q} = (q_1, \ldots, q_m)$ on a universe $U$ of size $N$, let us define the corresponding $m\times N$ query matrix $\qmat$, in which the $i$-th row is the truth table of $q_i$. Overloading notation, we use $q_i$ as well for the $i$-th row of $Q$. We define the histogram $h$ of the dataset $\ds = (\ds_1, \ldots, \ds_\dsize)$ as a vector in $\R^\uni$ with entries $h_{\elem} \eqdef |\{i: \ds_i = \elem\}|$. We can then write $\mathcal{Q}(\ds)$ as a linear transformation: $\mathcal{Q}(\ds) = \frac1n \qmat h$. The central insight of factorization mechanisms is that, instead of releasing $\mathcal{Q}(\ds) = \frac1n \qmat h$ with independent and identically distributed (i.i.d.)~Gaussian noise, as in the Gaussian mechanism, it may be beneficial to instead release the answers to some \emph{strategy queries} with i.i.d.~noise, and then reconstruct answers to $\mathcal{Q}$ by taking linear combinations of the noisy answers to the strategy queries. The strategy queries are encoded by a matrix $R$, and the linear combinations required for reconstruction are encoded by a matrix $L$. Then, the exact answers to the strategy queries are given by $\frac1n Rh$, and being able to reconstruct answers to $\mathcal{Q}$ translates to requiring that $\qmat = LR$. A factorization mechanism takes matrices $L$ and $R$ satisfying $\qmat = LR$, and outputs $\frac1n L(Rh+Z) = Qh + \frac1n LZ$, where $Z \sim N(0,\sigma^2 I)$ is Gaussian, and $\sigma$ is proportional to the sensitivity of the strategy queries. In particular, we need $\sigma = O\left(\frac1n\|R\|_{1\to 2}\right)$, where $\|R\|_{1\to 2}$ is the matrix norm equal to the maximum $\ell_2$ norm of a column of $R$. (Here we ignore the privacy parameters for simplicity.) Then a standard analysis shows that the maximum expected squared error over queries is
\(
\max_{i=1}^m \E[(LZ)_i^2] = \sigma^2 \|L\|_{2 \to \infty}^2,
\)
where $\|L\|_{2 \to \infty}$ is the maximum $\ell_2$ norm of a row of $L$. To minimize the maximum standard deviation of the error  over choices of $L$ and $R$, we need to solve the optimization problem
\[
\min \left\{\frac{\|R\|_{1\to 2}\|L\|_{2 \to \infty}}{n}: \qmat = LR\right\}= 
\frac{1}{ n}\cdot\min \left\{\|R\|_{1\to 2}\|L\|_{2 \to \infty}: \qmat = LR\right\}.
\]
The quantity $\min\left\{\|R\|_{1\to 2}\|L\|_{2 \to \infty}: \qmat = LR\right\}$ is known as the $\gamma_2$ factorization norm of $\qmat$ and is denoted by $\gamma_2(\qmat)$. Thus, an optimal choice of factorization mechanism allows us to release $\mathcal{Q}(\ds)$ under differential privacy constraints so that the standard deviation of the error for each query is bounded by $O\left(\frac1n\gamma_2(\qmat)\right)$. Standard concentration arguments also allow us to bound the maximum error over queries (i.e., $\ell_\infty$ error) by $O\left(\frac1n\gamma_2(\qmat)\sqrt{\log m}\right)$ with high probability.

\subsubsection{Online Private Query Release}

A factorization of $Q$ that achieves $\gamma_2(Q)$ is efficiently computable by semidefinite programming~\cite{LinialMSS07}, but, naturally, this requires knowing all of $Q$. One might guess that attaining the error bound $O\left(\frac{\gamma_2(\qmat)\sqrt{\log m}}{n}\right)$ might be impossible when queries arrive online, but in this work we show that, surprisingly, this is not the case, at least up to logarithmic factors. We prove the following theorem. 

\begin{restatable}{theorem}{thmmainpriv}\label{thm:main-priv}
    For any $\varepsilon > 0$, $\delta \le e^{-\eps}$, $\beta, c > 0$, there is an interactive algorithm that is $(\varepsilon,\delta)$-differentially private for non-adaptive data analysts, and, when the analyst chooses counting queries $\mathcal{Q} = (q_1, \ldots, q_m)$ on a universe $U$, the algorithm computes answers $a_1, \ldots, a_m$ such that, with probability at least $1-\beta$, $\max_{t=1}^m |q_t(x) - a_t|$ is bounded by
    \[
      O\left(\frac{\gamma_2(Q)\log(\gamma_2(Q))^{\frac{1}{2}+c}\log(|U|)^3 \sqrt{\log(m + 1/\beta)\log(1/\delta)}}{\varepsilon n}\right),
    \]
    where $Q$ is the query matrix of $\mathcal{Q}$. Moreover, when the queries are specified by the corresponding row vectors in the query matrix $Q$, the algorithm runs in polynomial time in $m,|U|,n$.
\end{restatable}

Note that Theorem~\ref{thm:main-priv} holds even when the interactive algorithm has no prior information about the sequence of queries $\mathcal{Q}$, and, in particular, the algorithm does not know the number of queries $m$.


The proof of Theorem~\ref{thm:main-priv} is given in Section~\ref{sec:gen-fact}, and is based on an online algorithm for factoring matrices (see Section~\ref{sec:intro-fact} for an overview). We, in fact, prove a more general statement that applies to arbitrary statistical queries rather than just counting queries. We also show, in Theorem~\ref{thm:priv-lb}, that the $\log(\gamma_2(Q))^{\frac{1}{2}+c}$ cannot be improved significantly in this more general setting.

Note that Theorem~\ref{thm:main-priv} assumes that the data analyst is \emph{non-adaptive}, i.e., that the analyst determines the queries in $\mathcal{Q}$ independently of the answers $a_1, \ldots, a_m$ computed by the algorithm. The privacy of our algorithm in fact holds against adaptive adversaries, but the accuracy guarantee does not (see Section~\ref{sec:prelims} for details about how privacy and accuracy are defined in this setting). It is an interesting question whether a similar result holds in the fully adaptive setting.  We remark that for worst-case queries there is no separation between oblivious and adaptive adversaries~\cite{BunSU17}.

It is interesting to compare Theorem~\ref{thm:main-priv} to results by Edmonds, Nikolov, and Ullman~\cite{EdmondsNU20}, who showed that these offline factorization mechanisms are optimal in the \emph{large dataset regime}. That is, for small enough constant $\varepsilon, \delta, \beta$, and for any (offline)  $(\varepsilon,\delta)$-differentially private algorithm computing answers $a_1, \ldots, a_m$ to some set $\mathcal Q$ of statistical queries, we have that, \emph{for sufficiently large values of $n$}, there is some dataset $x \in \uni^n$,

\[
\max_{i = 1}^m |a_i - q_i(x)| = \Omega\left(\frac{\gamma_2(Q)}{\varepsilon n}\right)
\]
with probability at least $1-\beta$.  Thus, for large enough datasets, the error bound in Theorem~\ref{thm:main-priv} is \emph{competitive} against an optimal private query release algorithm up to a competitive ratio polynomial in $\log(mU/(\beta\delta))$.  As far as we are aware, this is the first competitiveness result of this type for online private query release.  In Section~\ref{sec:intro-smalln} we discuss partial progress towards giving online competitive private query release in the small dataset regime.   More generally, our algorithms are competitive with the optimal offline factorization mechanism in all regimes, whether or not that mechanism is optimal among all differentially private algorithms.

\subsubsection{Online Private Query Release for Small Datasets}\label{sec:intro-smalln}
When the dataset size $n$ is small a different set of private mechanisms give optimal error guarantees with better error bounds than Theorem~\ref{thm:main-priv}, and the optimal error rate is governed by different quantities than $\gamma_2(Q)$.  Specifically, Nikolov, Talwar, and Zhang showed that if we fix a set of queries with query matrix $Q$, and a dataset of size $n$, then the optimal error is governed by the hereditary discrepancy of $Q$.  Given a matrix $Q$ with columns indexed by $U$, and a subset $S \subseteq U$, we use $\QS$ to denote the $m \times |S|$ submatrix of $\mat{Q}$ consisting of all columns of $\mat{Q}$ indexed by $S$.  With this notation, we define the \emph{hereditary discrepancy}, and the $\ell_2$ hereditary discrepancy as
$$
\hdisc(\mat{Q}, n) := \max_{S \subset [N] : |S| \leq n} \left(\min_{\vec{x} \in \{\pm 1\}^{S}} \| \QS \vec{x} \|_{\infty}\right)
\quad \text{and} \quad
\hdisc_2(\mat{Q}, n) := \max_{S \subset [N] : |S| \leq n} \left(\min_{\vec{x} \in \{\pm 1\}^{S}} \| \QS \vec{x} \|_2\right).
$$
To characterize the optimal error, we actually need to slightly modify this standard definition of hereditary discrepancy. For this reason we introduce the \emph{modified hereditary discrepancy} and modified hereditary $\ell_2$ discrepancy $\hdisc^*(Q,n)$ and $\hdisc_2^*(Q,n)$ defined, respectively, by replacing $$\| \QS \vec{x} \|_{\infty} \quad \text{with} \quad \| \QS \vec{x} \|_{\infty} + \bigg|\sum_{i \in S} x_i\bigg|$$ and replacing $$\| \QS \vec{x} \|_2 \quad \text{with} \quad \| \QS \vec{x} \|_2 + \bigg|\sum_{i \in S} x_i\bigg|.$$ These quantities are usually equal to $\hdisc(Q,n)$ and $\hdisc(Q,n)$ up to constants.

In the small dataset regime, the optimal error for a \emph{worst-case} set of queries $\qmat$ is given by the private multiplicative weights algorithm of Hardt an Rothblum~\cite{HardtR10} and its subsequent analysis~\cite{GuptaRU12,HardtLM12}.  Their algorithm $M_{\queries}^\textit{PMW}$ (which also works for online queries) has error
$$
\ex{}{\| \queries(x) - M_{\queries}(x) \|_{\infty}} = \alpha_{\textit{PMW}} = \tilde{O}\left(\frac{\sqrt{\log(|U|)} \log m}{n} \right)^{1/2}
$$
and this is known to be optimal for some workload of queries $\qmat$ by the lower bounds of Bun, Ullman, and Vadhan~\cite{BunUV14}.  Note that the dependence on $n$ is  proportional to $n^{-1/2}$ and the error grows polylogarithmically in the number of queries and the universe size for any set of queries.  In the offline setting, Nikolov, Talwar, and Zhang~\cite{NikolovTZ13} showed that we can achieve better bounds for the weaker \emph{$\ell_2$ error guarantee} for nice workloads of queries that have small hereditary discrepancy.  Specifically their mechanism has error
$$
\ex{}{\frac{ \| \queries(x) - M_{\queries}^{\mathit{NTZ}}(x) \|_{2}}{\sqrt{m}}} = \alpha_{\mathit{NTZ}} = O\left(\frac{\hdisc_2^*(Q,n)}{n}\cdot \mathrm{polylog}(|U|,m)\right).
$$

Moreover, a slight modification of the results of Muthukrishnan and Nikolov~\cite{MuthukrishnanN12} shows that for every differentially private algorithm $M_{\queries}$ there exists a dataset $x \in U^n$ such that 
$$
\ex{}{\frac{\| \queries(x) - M_{\queries}(x) \|_{2}}{\sqrt{m}}} = \alpha_{\mathit{DISC}} = \Omega\left(\frac{\hdisc_2^*(Q,n)}{n\log n} \right).
$$

We give the first online differentially private algorithm that competes with this refined discrepancy bound.  Moreover, our algorithm is the first that gives any algorithm that competes with the discrepancy lower bound for the more stringent $\ell_\infty$ error guarantee.
\begin{thm} [Informal, Some Parameters Omitted] \label{thm:intro-smalln}
There exists a differentially private interactive algorithm that takes a dataset $x \in [U]^n$ such that if a data analyst asks a set of adaptively chosen Boolean statistical queries $\mathcal{Q} = (q_1,\dots,q_k)$, and the algorithm answers each query with $a = (a_1,\dots,a_k)$, then with probability at least $1 - o_n(1)$,
$$
\| Q(x) - a \|_{\infty} \leq \alpha = \tilde{O}\left( \frac{\sqrt{\hdisc^*(Q,n) \log(|U|)}\log m}{n} \right)^{2/3} = \alpha_{\mathit{DISC}} \cdot \tilde{O}\left(\frac{\hdisc^*(Q,n)^2}{n \log(|U|) \log(m)^2}\right)^{1/3},
$$
\end{thm}

To interpret the guarantee of our algorithm is it useful to focus only on the dependence on $n$ in the extreme cases for $\hdisc(Q,n)$.  In the worst case where $\hdisc(Q,n) \approx \sqrt{n}$ is approximately maximal, the optimal error is $\tilde{O}(n^{-1/2})$ and our algorithm achieves the same asymptotic guarantee.  However, in the case of easy workloads where $\hdisc(Q,n) \approx 1$ is approximately minimal, then the optimal error is $\tilde{O}(n^{-1})$, and our algorithm's error improves to $\tilde{O}(n^{-2/3})$.  One way to state the error $\alpha$ for our algorithm is that $\alpha \approx \alpha_{\mathit{PMW}}^{2/3} \alpha_{\mathit{NTZ}}^{1/3}$.  Since $\alpha_{\mathit{NTZ}}$ nearly matches the lower bound $\alpha_{\mathit{DISC}}$, which holds even for offline algorithms that know $Q$, we can say that on a logarithmic scale we get at least one third of the way from a worst-case optimal algorithm to one that can adapt as well as an offline algorithm.

\paragraph{Techniques.} The algorithm we use to prove Theorem~\ref{thm:intro-smalln} is based on the Median Mechanism of Roth and Roughgarden~\cite{RothR10}, which shows how to answer any set of online queries with error that is determined by the size of the smallest cover of the set of possible answers.  To obtain worst-case bounds, they use the fact that for any dataset and any workload of $m$ queries over domain $U$, a random subsample of size $s= O(\log m / \alpha^2)$ will give accurate answers with probability $\frac12$, which implies a cover of size $U^s$.  We give bounds that adapt to the difficulty of the workload by observing that low hereditary discrepancy implies even smaller covers.  Specifically, for any dataset and any workload $Q$, there exists a small dataset of size $s \approx O(\frac{1}{\alpha} \hdisc(Q,\frac{1}{\alpha^2}))$ that gives accurate answers up to error $\alpha$.  This immediately yields smaller covers for sets with small hereditary discrepancy, which we then plug into the median mechanism machinery.  Proofs for this section are contained in Section~\ref{sect:jon}.

\paragraph{A Remark on Discrepancy vs.\ Factorization Norms.} Let us note here that, also in the large dataset setting, the first results that showed factorization mechanisms achieve optimal error for offline queries also used the connection between hereditary discrepancy and differential privacy~\cite{NikolovTZ13}. In the large dataset setting, the relevant discrepancy quantity of interest is $\hdisc(Q)$, which is the maximum of $\hdisc(Q,n)$ over all $n$. It turns out that the $\gamma_2$ norm and hereditary discrepancy are closely related, and, in fact, $\gamma_2(Q)$ and $\hdisc(Q)$ are equal up to logarithmic (in $m$) factors for all $m\times N$ matrices $Q$~\cite{disc-gamma2}. Similarly to $\hdisc(Q,n)$, one can define a matrix factorization quantity $\gamma_2(\qmat, n)$ as the maximum $\gamma_2$ norm of any submatrix of $Q$ with $n$ columns. Thus, we could equally well state our results in terms of $\gamma_2(\qmat, n)$, which would be superficially more similar to the large dataset case.  However, the techniques we use in the small dataset case are more closely related to hereditary discrepancy than factorization so we state these results in terms of hereditary discrepancy.

\paragraph{A Remark on $\ell_\infty$ vs.\ $\ell_2$ Error.} Our results exhibit a gap in competitive ratio, where the error of our mechanism does not match the best offline lower bound.  We remark that, our techniques naturally give a guarantee on the $\ell_\infty$ error, as opposed to the weaker $\ell_2$ error guarantee.  In fact, if we insist on $\ell_\infty$ error, then we don't know any algorithm that approaches the discrepancy lower bound, even in the easier offline setting.  Thus, if we want a better competitive guarantee in the online setting, we would either need to improve the best offline mechanisms for $\ell_\infty$ error or would need to find a different set of techniques that allow us to benefit from the weaker $\ell_2$ guarantee.


\subsection{Online Discrepancy Minimization}

Given the close connections between private query release, factorization norms, and hereditary discrepancy, one can reasonably ask if our techniques have implications for discrepancy theory as well. The most basic notion in discrepancy theory is the discrepancy of an $m\times N$ matrix $A$, defined by $\disc(A) = \min_{x \in \{-1,+1\}^N} \|Ax\|_{\infty}$. The hereditary discrepancy $\hdisc(A)$, introduced in the previous subsection, and equal to the maximum discrepancy of any submatrix of $A$, is a more robust notion that is often better behaved. As alluded to above, the $\gamma_2$ norm and hereditary discrepancy of any matrix are equal up to logarithmic factors: Matou\v{s}ek, Nikolov, and Talwar~\cite{disc-gamma2} showed that for any $m\times N$ matrix $A$,
\begin{align*}
    \gamma_2(A) = O(\log m)\cdot \hdisc(A); \hspace{2em}
    \hdisc(A) = O(\sqrt{\log m}) \cdot \gamma_2(A).
\end{align*}
Moreover, Bansal, Dadush, Garg, and Lovett~\cite{BDGL} showed that there is a polynomial time algorithm that, on input the matrix $A$, computes an $x \in \{-1,+1\}^N$ which achieves $\|Ax\|_\infty = O(\sqrt{\log m}) \cdot \gamma_2(A)$. A natural question, analogous to the questions we study in online private query release, is whether such an $x$ can also be computed online, when the matrix $A$ is revealed gradually. It turns out that our techniques allow us to do so, as we explain in more detail below. 

An online variant of matrix discrepancy was first formulated by Spencer~\cite{Spencer77}, and studied further by B\'ar\'any~\cite{Barany79}. In this online discrepancy minimization problem, at each time step a \emph{column} $a_t \in \R^m$ of the matrix $A$ arrives, and the algorithm must immediately decide the value of $x_t \in \{-1,+1\}$ so as to minimize $\max_{t=1}^N \|x_1 a_1 + \ldots + x_t a_t\|_\infty$ (where the value of $N$ is not known to the algorithm). Spencer and B\'ar\'any considered the case in which the columns of $A$ are determined by an adaptive adversary, i.e., $a_t$ can depend on the choice of $x_1, \ldots, x_{t-1}$. In this adaptive setting the discrepancy achievable by any online algorithm may be much larger than the offline discrepancy of $A$, as already observed by Spencer~\cite{Spencer77}. To see this, consider an adversary that always chooses $a_t$ to be a vector of unit $\ell_2$ norm, orthogonal to $x_1 a_1 + \ldots + x_{t-1}a_{t-1}$. Then, at time $t$, regardless of the choice of $x_1, \ldots, x_t \in \{-1,+1\}$, we have
\[
\|x_1 a_1 + \ldots + x_t a_t\|_\infty \ge \frac{1}{\sqrt{m}}\|x_1 a_1 + \ldots + x_t a_t\|_2 = \sqrt{\frac{t}{m}}.
\]
By contrast, Banaszczyk showed that, for any matrix $A$ whose columns have $\ell_2$ norm at most $1$, $\hdisc(A) \lesssim \sqrt{\log m}$, regardless of the number of columns of $A$~\cite{Bana98} (this bound was recently improved by Bansal and Jiang~\cite{BansalJiang25}). Because of examples like this, recent work on the online discrepancy problem has focused on stochastic adversaries, i.e., each $a_t$ is drawn independently from a fixed distribution~\cite{BansalSpencer20}, and the more general setting of oblivious adversaries, i.e., the columns $a_t$ of $A$ are chosen ahead of time, before $x_1, \ldots, x_N$~\cite{alweiss2020discrepancy}. In this paper we also only consider oblivious adversaries.  

As an example of a result in this online discrepancy model, Kulkarni, Reis, and Rothvoss~\cite{KulkarniRR24} showed that, if all columns of $A$ have $\ell_2$ norm at most $1$, then there is an online (but inefficient) algorithm achieving, with high probability,
\[
\max_{t = 1}^N\|x_1 a_1 + \ldots + x_t a_t\|_\infty = O(\sqrt{\log(mN)}).
\]
Their result improved on a bound of $O(\log(mN))$ achieved by Alweiss, Liu, and Sawhney with an efficient algorithm~\cite{alweiss2020discrepancy}. There is, however, no known algorithm that achieves a bound on $\max_{t=1}^N\|x_1 a_1 + \ldots + x_t a_t\|_\infty$ that is competitive against $\hdisc(A)$ or against $\gamma_2(A)$. We show that a variant of our online matrix factorization algorithm, combined with the results of Kulkarni, Reis, and Rothvoss, does give such an algorithm. 
\begin{restatable}{theorem}{thmmaindisc}
    \label{thm:disc-main}
    There is an online algorithm such that, for any matrix $A$ whose columns $a_1, \ldots, a_N\in \R^m$ arrive online in the oblivious model, the algorithm computes online a sequence of signs $x_1, \ldots, x_N \in \{-1,+1\}$ such that, with probability at least $1-\frac{1}{N}$,
    \[
    \max_{t=1}^N \|x_1 a_1 + \ldots + x_t a_t\|_\infty = O\left(  \gamma_2(A)\log(m)^3\sqrt{\log(mN)}\right).
    \]
    In particular, the signs $x_1, \ldots, x_N$ computed by the algorithm also satisfy
    \[
    \max_{t=1}^N \|x_1 a_1 + \ldots + x_t a_t\|_\infty = O\left( \hdisc(A) \log(m)^4\sqrt{\log(mN)}\right)
    \]
    with probability at least $1-\frac1N$.
\end{restatable}

Like Theorem~\ref{thm:main-priv}, Theorem~\ref{thm:disc-main} is also based on an online algorithm for factoring matrices, but in the column arrival model. The proof can be found in Section~\ref{sec:gen-fact}, and we give an overview of our online matrix factorization techniques and how they are applied in Section~\ref{sec:intro-fact} below.

The algorithm guaranteed by Theorem~\ref{thm:disc-main} does not run in polynomial time because it relies on the inefficient algorithm of Kulkarni, Reis, and Rothvoss. Nevertheless, by using the results of Alweiss, Liu, and Sawhney instead, we can get a polynomial time algorithm achieving the guarantee 
\begin{equation}\label{eq:disc-main-efficient}
    \max_{t=1}^N \|x_1 a_1 + \ldots + x_t a_t\|_\infty = O\left(  \gamma_2(A) \log(m)^3\log(mN)\right)
\end{equation}
with probability $1-\frac1N$. 

Theorem~\ref{thm:disc-main} is an online version of a ``pseudoapproximation'' algorithm of Bansal, which, on input an $m\times N$ matrix $A$, finds a coloring $x \in \{-1,+1\}^N$ in polynomial time that satisfies $\|Ax\|_\infty = O(\hdisc(A)\log(mN))$ with high probability. This type of guarantee cannot be improved to a true polynomial time approximation, i.e., to $\|Ax\|_\infty \le C\disc(A)$ for any finite $C$, unless $\mathsf{P} = \mathsf{NP}$, because deciding if $\disc(A) = 0$ is an $\mathsf{NP}$-complete problem~\cite{dischard}.  Similarly, in the online setting we cannot hope for an algorithm (even an inefficient one) that is competitive against $\disc(A)$. For example, consider the single row matrix $A = (1, 1, 2)$. We have $\disc(A_2) = 0$, and this discrepancy is achieved only if $x_1 \neq x_2$; we also have $\disc(A_3) = \disc(A) = 0$, and this discrepancy is achieved only if $x_1 = x_2\neq x_3$. Clearly no online algorithm can achieve discrepancy $0$ both at times $t = 2$ and $t=3$. It is not clear, however, if an online algorithm could compete against the prefix discrepancy of $A$, i.e., against $\min_{x \in \{-1,+1\}^N} \max_{t=1}^N \|x_1 a_1 + \ldots + x_t a_t\|_\infty$.

\subsection{Online Factorization} \label{sec:intro-fact}

To prove Theorems~\ref{thm:main-priv}~and~\ref{thm:disc-main}, we formulate an online variant of the matrix factorization optimization problem. We formulate both a \emph{row arrival} model, suitable for applications to online query release, and a \emph{column arrival} model, suitable for applications to online discrepancy minimization. Let us introduce the row arrival model first.
At each time step a new row $q_t$ of an $m\times N$ matrix $Q$ arrives. We treat $q_t$ as a row vector. An \emph{online matrix factorization algorithm} receives this online sequence of rows, and maintains two growing matrices $L_t$ and $R_t$ such that $Q_t = L_t R_t$, where $Q_t$ is the submatrix of $Q$ consisting of the first $t$ rows. Initially, we set $R_0$ to be some arbitrary, possibly empty, matrix. At time step $t$, after receiving $q_t$, the algorithm computes $R_t$, with the requirement that the rows of $R_{t-1}$ are a subset of the rows of $R_t$. Together with $R_t$, the algorithm also computes a row vector $\ell_t$ so that $\ell_t R_t= q_t$. Note that the dimensions of $\ell_t$ equals the number of rows of $R_t$, so the vectors $\ell_1, \ldots, \ell_t$ grow in dimension. The matrix $L_t$ then consists of the rows $\ell_1, \ldots, \ell_t$, suitably padded with $0$ entries so that $L_t R_t = Q_t$. We require that, for every $t \in [m]$, $\|R_t\|_{1\to 2} \le 1$, i.e., that each column of $R_t$ has $\ell_2$ norm at most $1$. We say that the online factorization algorithm achieves \emph{factorization value} $\gamma$ for the matrix $Q$ if $\max_{t = 1}^m \|\ell_t\|_2 = \|L_t\|_{2\to \infty} \le \gamma$. Note that the value of $m$ is not known to the algorithm.

The column arrival online factorization model can be seen as ``transpose'' of the row arrival model. Here, at each time step a new column $a_t$ of the matrix $A$ (seen as column vector in $\R^m$) arrives, and we maintain growing matrices $L_t$ and $R_t$ such that $L_t R_t = A_t$ at all times $t$, where $A_t$ is the matrix with columns $a_1, \ldots, a_t$. The roles of $L_t$ and $R_t$ are swapped in the column arrival model. In particular, $L_0$ is an arbitrary, potentially empty, initial left matrix, and at time step $t$ we add columns to $L_{t-1}$ to form $L_t$, and compute a single column $r_t$ such that $L_t r_t = a_t$; then $R_t$ is the matrix with columns $r_1, \ldots, r_t$, suitably padded with $0$'s so that they are in the same dimension. Like the row arrival model, we still require that $\|R_t\|_{1\to 2} \le 1$ at all times $t$, and we define the value of the factorization as $\gamma = \|L_N\|_{2\to\infty}$, which is also an upper bound on $\|L_t\|_{2\to\infty}$ for all $t$.

\subsubsection{Main Result and Applications}

Our main result for online factorization is an algorithm that factors any matrix arriving online with factorization value that matches the best possible offline factorization, i.e., the $\gamma_2$ norm, up to logarithmic factors. Our techniques work both in the column and row arrival model with only minor modifications: in fact, our results in both models can be reduced to the same decision version of the online factorization problem, as explained in Section~\ref{sec:gen-fact}.

\begin{restatable}{theorem}{factmain} \label{thm:fact-main}
For any $c>0$, there is an online factorization algorithm in the row arrival model such that, for any $m\times N$ matrix $Q$ the algorithm achieves factorization value 
\[
\gamma = O\left(  \gamma_2(Q)\log(\gamma_2(Q)/\|q_1\|_\infty)^{\frac12 + c}\log(N)^3\right).
\]
Moreover, there is an online factorization algorithm  in the column arrival model such that, for any $m\times N$ matrix $A$, the algorithm achieves factorization value
\[
\gamma = O\left(  \gamma_2(A)\log(m)^3\right).
\]
\end{restatable}

Note that the guarantee in the row arrival model is slightly worse as the competitive ratio bound includes the factor $\log(\gamma_2(Q)/\|q_1\|_\infty)^{\frac12 + c}$. In Theorem~\ref{thm:fact-lb} we show that this factor cannot be improved significantly. 

Theorem~\ref{thm:fact-main} is proved in Section~\ref{sec:gen-fact}, where we also explain how it implies Theorems~\ref{thm:main-priv}~and~\ref{thm:disc-main}. Both applications are relatively straightforward and follow along the same lines as offline applications of matrix factorization norms to private query release, and to discrepancy minimization, respectively. Before giving an overview of the the techniques behind the proof of Theorem~\ref{thm:fact-main}, we briefly describe how the applications follow from it. 

To prove Theorem~\ref{thm:main-priv}, we use the online factorization guaranteed by Theorem~\ref{thm:fact-main} similarly to how the offline factorization is used in standard factorization mechanisms. When query $q_t$ arrives, the online factorization adds some rows to $R_t$, and the private algorithm releases answers to the statistical queries corresponding to these rows, perturbed by independent Gaussian noise. This is just an instance of the Gaussian mechanism applied to the online queries encoded by $R_t$. The requirement that $\|R_t\|_{1\to 2} \le 1$ guarantees that the sensitivity of these online queries is at most $\frac2n$. The private algorithm can then reconstruct an answer to $q_t$ by taking the dot product of $\ell_t$ with the vector of noisy values computed so far. The privacy and accuracy analysis is similar to the offline setting.

The derivation of Theorem~\ref{thm:disc-main} from Theorem~\ref{thm:fact-main} relies on a result of Kulkarni, Reis, and Rothvoss~\cite{KulkarniRR24}, showing that, for an oblivious online sequence of vectors $v_1, \ldots, v_N$, each of $\ell_2$ norm at most $1$, one can compute random signs $x_1, \ldots, x_N \in \{-1,+1\}$ online so that, for each $t \in [N]$, the random vector $x_1 v_1 + \ldots + x_tv_t$ is $O(1)$-subgaussian in every direction. We apply this theorem to the sequence of vectors $r_1, \ldots, r_N$ computed by the online factorization algorithm. By the subgaussian guarantee, each coordinate of $x_1 a_1 + \ldots + x_t a_t = L_t (x_1 r_1 + \ldots + x_t r_t)$ is $O(\gamma)$-subgaussian, where $\gamma$ is the value of the factorization. A union bound now implies that $\|x_1 a_1 + \ldots + x_t a_t\|_\infty$ is $O(\gamma \sqrt{\log(mN)})$ with probability at least $1-\frac{1}{N^2}$. 

\subsubsection{Techniques}

To prove Theorem~\ref{thm:fact-main}, we first give an algorithm that solves the online \emph{average} factorization problem, where we ask the Frobenius norm of the $R_t$ matrix to be bounded by $\sqrt{N}$ for evert $t \in [m]$. That is, rather than requiring that each column of $R_t$ has $\ell_2$ norm at most $1$, we only require the average squared $\ell_2$ norm of $R_t$ is at most $1$. In the offline setting, this corresponds to minimizing $\|L\|_{2\to\infty}\frac{\|R\|_{F}}{\sqrt{N}}$ over all factorizations $Q = LR$. The quantity $\min\{\|L\|_{2\to\infty}\|R\|_{F}: Q = LR\}$ is equivalent to the factorization norm $\gamma_F(Q^T)$ defined by Edmonds et al.~\cite{EdmondsNU20}. The first step in designing our online average factorization algorithm is to represent $\gamma_F(Q^T)$ as a semidefinite program (SDP):
\begin{equation}\label{eq:sdp-gammaF}
\gamma_F(Q^T) = \min\{\tr(X): X\succeq q_t^T q_t \ \forall t \in [m], X\succeq 0\},
\end{equation}
where $q_t$ is the row vector equal to the $t$-th row of $Q$, $X$ ranges over $N\times N$ positive semidefinite matrices, and $X\succeq q_t^T q_t$ means that $X-q_t^T q_t$ is positive semidefinite (PSD). We prove \eqref{eq:sdp-gammaF} after Lemma \ref{lm:fact-psd}, but let us sketch here why the equality holds. The PSD matrix $X$ represents $R^T R$, so $\tr(X) = \|R\|_F^2$, and the constraints $\forall t \in [m]: X\succeq q_t^T q_t$ ensure that there exists a matrix $L$ with $Q= LR$ such that $\|L\|_{2 \to \infty} \le 1$ (see Lemma~\ref{lm:fact-psd}). The right hand side of \eqref{eq:sdp-gammaF} then corresponds to minimizing $\|R\|_{F}^2$ over factorizations $Q = LR$ with $\|L\|_{2\to\infty} \le 1$. It is easy to see that this problem is equivalent to computing a factorization achieving $\gamma_F(Q^T)$. 

Suppose we have guessed that $\gamma = \gamma_2(Q)$: we can remove this assumption using a standard doubling search for $\gamma_2(Q)$. The online factorization $Q_t = L_t R_t$ computed by our online average factorization algorithm then needs to satisfy $\|R_t\|_F \le \sqrt{N}$ and $\|L_t\|_{2\to\infty}\le C \gamma$ for some competitive ratio $C$. In particular, this would certify that $\gamma_F(Q_t^T) \le C\gamma \sqrt{N}$, so it is natural to approach this problem by solving the SDP in \eqref{eq:sdp-gammaF} online. Notice that the requirement that the rows of the matrix $R_t$ contain the rows of $R_{t-1}$ is equivalent to requiring that $R_t^T R_t \succeq R_{t-1}^T R_{t-1}$. Motivated by this observation, we pose the following online version of the SDP in \eqref{eq:sdp-gammaF}: on online input the rows $q_1, \ldots, q_m$ of $Q$ (as row vectors), compute PSD matrices $X_0 \preceq X_1 \preceq\ldots \preceq X_m$, where $X_0$ is an initial matrix, and $X_t$ is computed upon receiving $q_t$; each matrix $X_t$ should satisfy $X_t \succeq \frac{1}{C^2\gamma^2} q_t^T q_t$, as well as $\tr(X_t) \le N$. Analogously to the offline case, $X_t$ represents $R_t^T R_t$; the constraint $X_t \succeq \frac{1}{C^2\gamma^2}q_t^T q_t$ guarantees that there is a row vector $\ell_t$ such that $\ell_t R_t = q_t$ and $\|\ell_t\|_2 \le C\gamma$; finally, the constraint $\tr(X_t) \le N$ ensures that $\|R_t\|_F \le \sqrt{N}.$ We would like our online algorithm to either output $X_t$ at time step $t$, or provide a certificate that $\gamma_F(Q_t^T) > \gamma\sqrt{N}$ (at which point we can double our guess $\gamma$ for $\gamma_2(Q)$, and restart the online average factorization algorithm). Let us call this problem the online average factorization SDP. 

This online semidefinite programming problem is reminiscent of the well-studied online covering linear programming (LP) problem~\cite{onlineLP}. In an online covering LP, the goal is to compute a sequence of vectors $0 \le x_1 \le x_2 \le \ldots \le x_m$ in $\R^N$ (where the inequality holds coordinate-wise), minimizing $\sum_i (x_m)_i$, and subject to the constraints $a_t x_t \ge b_t$, where the non-negative row vector $a_t$ and the non-negative real number $b_t$ arrive online at time step $t$. The standard way to generalize covering LPs to semidefinite programs is to have the constraints be of the form $\sum_i (x_t)_i A_{t,i} \succeq B_t$, where $A_{t,i}$ and $B_t$ are PSD matrices. In the online variant of covering SDPs, studied by Elad, Kale, and Naor~\cite{onlineSDP}, the matrices $A_{t,1}, \ldots, A_{t,N}, B_t$ defining the $t$-th constraint arrive online at time step $t$. There appears, however, to be no way to formulate the SDP in \eqref{eq:sdp-gammaF} as a covering SDP, since in a standard covering SDP the variables are scalars, whereas in the average factorization SDP the variable is a PSD matrix.

One may still hope that, even though the online average factorization SDP is not an instance of the online covering SDP problem, existing techniques extend to it.
Unfortunately, attempting to adapt the online primal-dual algorithms used for online covering LPs and SDPs~\cite{onlineLP,onlineSDP} to SDPs, like the online average factorization SDP, where the variables are matrix valued, runs into significant roadblocks. In particular, the online primal-dual algorithms for covering LPs and SDPs compute the LP solutions $x_t$ as exponentials of positive linear combinations of dual variables, where dual variables increase over time. The analysis then crucially uses the simple fact that the exponential function is non-decreasing in order to prove that $x_1 \le x_2\le \ldots \le x_m$. It is natural to formulate an analogous algorithm for online SDPs with matrix-valued variables using the matrix exponential function, but this strategy fails because the matrix exponential is not monotone, i.e., $Y' \succeq Y$ does not guarantee that $\exp(Y') \succeq \exp(Y)$. 

Nevertheless, we are able to develop a simple online primal-dual algorithm for our SDP formulation of online average factorization achieving constant competitiveness ratio. Our starting point is the following dual formulation of $\gamma_F(Q^T)$: denoting, as before, $q_1, \ldots, q_m$ the rows of $Q$ as row vectors, 
\begin{equation}\label{eq:sdp-gammaF-dual}
    \gamma_F(Q^T) = \max \left\{\tr\left(\sqrt{\sum_{t=1}^m y_t q_t^T q_t}\right): y_1, \ldots, y_m \ge 0, \sum_{t = 1}^m y_t = 1\right\}.
\end{equation}
At time $t$, our algorithm maintains the dual variables $y_1, \ldots, y_t$ for the rows that have arrived so far, and a primal solution $X_t = X(y_1, \ldots, y_t) = \zeta I + \eta \sqrt{\sum_{i = 1}^t y_i q_i^T q_i}$ for some carefully chosen constants $\zeta$ and $\eta$. When a new row $q_t$ arrives, the algorithm increases $y_t$ continuously until either the constraint  $X_t \succeq \frac{1}{C^2\gamma^2} q_t^T q_t$ becomes satisfied, or $\tr(X_t)$  becomes equal to $N$, in which case the algorithm terminates and declares that $\gamma_F(Q^T) > \gamma \sqrt{N}$. The inequality $X_t \succeq X_{t-1}$ follows from the well-known fact that the matrix square root is monotone (in the PSD sense) on SDP matrices. In our analysis (in Section~\ref{sec:online-average}) we show that, for the right choice of $\eta$, $\zeta$, and the competitiveness ratio $C = O(1)$, if the algorithm ever terminates and asserts $\gamma_F(Q^T) > \gamma \sqrt{N}$, then $\sum_{i} y_i < 1$ and $$\tr\left(\sqrt{\sum_{t=1}^m y_iq_t^T q_t}\right) = \gamma\sqrt{N},$$ so the dual solution indeed verifies the assertion $\gamma_F(Q^T) > \gamma\sqrt{N}$. In particular, we can take $C$ to be arbitrarily close to $\sqrt{2}$. This online algorithm bears some similarities to algorithms by Eghbali, Saunderson, and Fazel~\cite{ESF-18} for a class of problems that includes the dual on the right hand side of~\eqref{eq:sdp-gammaF-dual}.

Having developed an $O(1)$-competitive algorithm for online average factorization, the final step in the proof of Theorem~\ref{thm:fact-main} is to build a reduction from online factorization to online average factorization. In the offline setup, \cite{MuthukrishnanN12} proposed a simple reduction based on removing columns of $Q$ for which the $\ell_2$ norms of the corresponding columns in the right matrix of the average factorization are too large, and factoring these columns of $Q$ recursively. In more detail, suppose that for any submatrix $Q_J$ of $Q$ consisting of the columns of $Q$ indexed by the set $J$, we have a factorization $L R = Q_J$ with $\|L\|_{2\to\infty} \|R\|_{F} \le \gamma \sqrt{|J|}$. We claim that this implies that $\gamma_2(Q) \lesssim \gamma \sqrt{\log(2N)}$. To build a factorization satisfying this bound, we start with the factorization $L' R' = Q$ such that $\|L'\|_{2\to\infty} \le \gamma$ and $\|R'\|_F = \sqrt{N}$,\footnote{Such a factorization exists because we can always replace $L'$ by $cL'$ and $R'$ by $\frac1c R'$ where $c = \frac{\sqrt{N}}{\|R'\|_{F}}$} we let $J$ be the set of indices of columns in $R'$ with $\ell_2$ norm at least $\sqrt{{2}}$, and recursively factor $Q_J$. Note that, by Markov's inequality, $|J| \le \frac{N}{2}$, so the recursive calls will eventually terminate. Suppose that the recursively found factorization is $L'' R'' = Q_{J}$, where $\|L''\|_{2\to\infty} \le C\gamma \sqrt{(1 +  \log_2(|J|))}$ for some constant $C$ and $\|R''\|_{1\to 2} \le \sqrt{2}$. Then we can set 
\[
L := 
\begin{pmatrix}
    L' & L''
\end{pmatrix}; \hspace{3em}
R := 
\begin{pmatrix}
    R'_{[N]\setminus J} & 0\\
    0 & R''
\end{pmatrix}.
\]
Here we assume that we have reordered the columns of $Q$ so that the columns indexed by $J$ come last. It is easy to see that $LR = Q$, $\|R\|_{1\to 2} \le \sqrt{2}$, and 
\[
\|L\|_{2\to\infty}^2 \le \|L'\|_{2\to\infty}^2 + \|L''\|_{2\to\infty}^2
\le \gamma^2 + C\gamma^2 (1 +  \log_2(|J|)) 
\le C\gamma^2(1 +  \log_2(N)), 
\]
where the last inequality holds for large enough $C$ because $|J| \le \frac{N}{2}$.

To extend this argument to the online setting, we need to extend the base online average factorization algorithm to also support insertion and one-time deletion of elements in its domain. These extensions allow us to dynamically remove columns of $Q_t$ for which the $\ell_2$ norm of the corresponding column of $R_t$ has become too large, and to insert them into a new instance of online average factorization. The new instance itself may remove columns that have become problematic, and insert them into yet another instance, etc. To carry out this strategy, we give transformations that take an online average factorization algorithm and augment it to handle insertion and one-time deletion, using classical methods for transforming static data structures to dynamic data structures~\cite{BentleySaxe}. Each transformation incurs a $O(\log N)$ factor in the competitiveness ratio, and our online version of the reduction from \cite{MuthukrishnanN12} contributes another $O(\log N)$ factor. 

With more information about the matrix being factored, we can further improve the reduction to online average factorization. As an example, we consider Boolean matrices with bounded VC dimension, and, more generally, matrices with polynomially bounded shatter function. The shatter function $\pi_Q$ of an $m\times N$ Boolean matrix $Q$ (i.e., a matrix with entries in $\{0,1\}$) is defined on any non-negative integer $s \le N$ so that $\pi_Q(s)$ equals the maximum number of \emph{distinct} rows in any submatrix of $Q$ with $s$ columns. The VC dimension of $Q$ is then the largest integer $d$ so that $\pi_Q(d) = 2^d$, i.e., some submatrix of $Q$ with $d$ columns contains all possible $2^d$ rows. We say that a matrix $Q$ has shatter function exponent $d$ if $\pi_Q(s) = O(s^d)$. By the Sauer-Shellah lemma~\cite[Theorem 8.3.16]{Vershynin}, the VC dimension of any Boolean matrix $Q$ bounds its shatter function exponent, but sometimes the shatter function exponent can be smaller. For example, if each column of $Q$ is indexed by a set $U$ of $N$ points in $\R^d$, and each row of $Q$ is the indicator vector of the set $U \cap H$ for some halfspace $H$ in $\R^d$, then the shatter function exponent of $Q$ is $d$ but the VC dimension is $d+1$ as long as $U$ contains $d+1$ affinely independent points. 

We have the following result for matrices with bounded shatter function exponent, proved in the Appendix in Section~\ref{sec:vcdim}.
\begin{restatable}{theorem}{factvcdim} 
\label{thm:fact-vcdim}
There is an online factorization algorithm in the row arrival model such that, for any $m\times N$ binary matrix $Q$ with shatter function exponent $d$, the algorithm achieves factorization value 
\[
\gamma = O\left(\min\left\{m^{\frac{1-1/d}{2}}\log^3{N}, N^{\frac{1-1/d}{2}}\log^{1/d}{N}\right\}\right).
\]
\end{restatable}
The optimal bound on $\gamma_2(Q)$ for $m\times N$ Boolean matrices $Q$ in terms of the shatter function exponent $d$, and the dimensions $m$ and $N$, is $O(\min\{m^{\frac{1-1/d}{2}}, N^{\frac{1-1/d}{2}}\})$~\cite{MuthukrishnanN12}. Thus, Theorem~\ref{thm:fact-vcdim} is only a factor $\log^{1/d}{N}$ away from the optimal offline bound when $m$ is large enough.

There are many natural examples of matrices with bounded VC dimension and bounded shatter function exponent coming from geometry, such as the matrix of indicator vectors of halfspaces restricted to a finite universe $U \subseteq \R^d$ mentioned above, 
As a corollary to Theorem~\ref{thm:fact-vcdim}, we get corresponding bounds for online private query release of halfspace counting queries, or, more generally counting queries with VC dimension $d$ or shatter function exponent $d$. We don't state these results formally, but note they follow in the same way that Theorem~\ref{thm:main-priv} follows from Theorem~\ref{thm:fact-main}.

\section{Preliminaries} \label{sec:prelims}

\subsection{Notational Conventions}
Given $n \in \N$, we use $[n] := \set{1,\dots,n}$ to denote the set of the first $n$ natural numbers.  Given a finite set $S$ and $k \in \N$, we use $\binom{S}{k} := \set{ T \subseteq S : |T| = k}$ and $\binom{S}{\leq k} := \set{ T \subseteq S : |T| \leq k}$ to denote the family of subsets of $S$ of size exactly $k$ and size at most $k$, respectively.  For a vector $v \in \R^d$, we use the standard notation for the $\ell_0$, $\ell_1$, $\ell_2$, and $\ell_\infty$ norms:  $$\| v \|_{0} := \card{\set{i : |v_i| > 0}}~~~~~~\| v \|_1 := |v_1| + \dots + |v_d|~~~~~~\| v \|_2 := \sqrt{|v_1|^2 + \dots + |v_d|^2}~~~~~~~\| v \|_{\infty} := \max_{i \in [d]} |v_{i}|.$$ Given $d \in \N$, we use $\Delta(d)$ to denote the probability simplex over the domain $[d]$. 
We use $X \succeq 0$ to denote that $X$ is a positive semidefinite (PSD) matrix and $X \succeq Y$ to denote that $X - Y \succeq 0$. The notation $X \preceq Y$ is equivalent to $Y \succeq X$.

For two quantities $A$ and $B$ that may depend on other parameters, we use the notation $A\lesssim B$ when there is an absolute constant $C$ such that $A \le CB$; the notation $A \gtrsim B$ is equivalent to $B \lesssim A$.

\subsection{Matrix Factorization Norms}

Factorization norms were introduced in functional analysis to study linear operators that factor through Hilbert space. Here, we focus on a special case that is most common in applications in computer science, and give the definitions in terms of matrices. We define the $\gamma_2$ factorization norm of an $m\times N$ matrix $A$ as
\[
\gamma_2(A) = \min\{\|L\|_{2\to\infty} \|R\|_{1\to 2}: LR = A\}.
\]
Recall that $\|R\|_{1\to 2}$ is the maximum $\ell_2$ norm of a column of $R$, and $\|L\|_{2\to\infty}$ is the maximum $\ell_2$ norm of a row of $L$. The $\gamma_2$ norm satisfies many nice properties, e.g.:
\begin{itemize}
    \item $\gamma_2(A) = \gamma_2(A^T)$;
    \item $\gamma_2(B) \le \gamma_2(A)$ for any submatrix $B$ of $A$;
    \item $\gamma_2(A)$ is a norm on matrices;
    \item $\gamma_2(A)$ can be approximated within any degree of accuracy in polynomial time given $A$ as input.
\end{itemize}
We refer to~\cite{disc-gamma2} for proofs of these facts. 

Motivated by applications, mostly to differential privacy, related factorization norms have been defined. Here we also consider the $\gamma_F$ norm~\cite{EdmondsNU20}, defined as
\[
\gamma_F(A) = \min\{\|L\|_{F} \|R\|_{1\to 2}: LR = A\}.
\]
Note that $\gamma_F(A) \le \sqrt{m} \gamma_2(A)$ for any $m\times N$ matrix $A$. One can show that $\gamma_F(A)$ is also a norm, and can be efficiently approximated in polynomial time, although we do not need these facts.

\subsection{Private Interactive Data Analysis}

We need to define privacy and accuracy for algorithms interacting with a private analyst. We do so in terms of a meta algorithm that outputs the transcript of the interaction between the algorithm $M$ and an \emph{interactive data analyst} $A$.

We first consider a non-adaptive analyst $A$ that determines the sequence of queries given to $M$ independently of the answers $M$ gives. Without loss of generality, we can assume that $A$ decides the sequence of queries before the interaction starts, as shown in Figure~\ref{fig:transcript-oblivious}.

\begin{figure}[ht!]
\begin{framed}
\begin{algorithmic}
\INDSTATE[0]{{\bf Input:} dataset $x \in U^n$}
\INDSTATE[0]{$A$ chooses queries $q_1,\ldots, q_m$}
\INDSTATE[0]{{\bf For} $j = 1,\dots,m$:}
\INDSTATE[1]{Give $q_{j}$ to $M(x)$ and receive an answer $a_j$}
\INDSTATE[0]{{\bf Output:} $(\mathcal{Q} = (q_1,\dots,q_m),a = (a_1,\dots,a_m))$.}
\end{algorithmic}
\end{framed}
\vspace{-6mm}
\caption{The Meta-Algorithm $\mathsf{Trans}_{m}^{M \prot A}(x)$ for non-adaptive analysts\label{fig:transcript-oblivious}}
\end{figure}

We also consider adaptive analysts $A$, that determine the next query to be given to $M$ based on the answers to previous queries, as shown in Figure~\ref{fig:transcript-adaptive}

\begin{figure}[ht!]
\begin{framed}
\begin{algorithmic}
\INDSTATE[0]{{\bf Input:} dataset $x \in U^n$}
\INDSTATE[0]{{\bf For} $j = 1,\dots,m$:}
\INDSTATE[1]{Give $a_{j-1}$ to $A$ and receive a query $q_j$}
\INDSTATE[1]{Give $q_{j}$ to $M(x)$ and receive an answer $a_j$}
\INDSTATE[0]{{\bf Output:} $(\mathcal{Q} = (q_1,\dots,q_m),a = (a_1,\dots,a_m))$.}
\end{algorithmic}
\end{framed}
\vspace{-6mm}
\caption{The Meta-Algorithm $\mathsf{Trans}_{m}^{M \prot A}(x)$ for adaptive analysts \label{fig:transcript-adaptive}}
\end{figure}

\begin{defn}[Privacy of Interactive Algorithms]
An interactive algorithm $M$ is \emph{$(\eps, \delta)$-differentially private} (for non-adaptive/adaptive analysts) if for every dataset $x \in U^n$, and every (resp.~non-adaptive/adaptive) analyst $A$, the algorithm $\mathsf{Trans}_{m}^{M \prot A}$ is $(\eps, \delta)$-differentially private.
\end{defn}

Similarly, we can define accuracy for interactive algorithms via the same meta-algorithm
\begin{defn}[Accuracy of Interactive Algorithms]
An interactive algorithm $M$ is \emph{$(\alpha,\beta)$-accurate for $m$ queries} and non-adaptive/adaptive analysts if for every dataset $x \in U^n$, and every (resp.~non-adaptive/adaptive) analyst $A$,
$$
\pr{(\mathcal{Q},a) \gets \mathsf{Trans}_{k}^{M \prot A}(x)}{ \| \mathcal{Q}(x) - a \|_{\infty} \leq \alpha} \geq 1-\beta.
$$
\end{defn}

We remark that in the definition of accuracy, typically we will think of $\beta, m, n, N$, as fixed and express the accuracy $\alpha$ as a function of those parameters.

\section{Online Matrix Factorization and Applications}

\subsection{The Online Matrix Factorization Model}

The row and column arrival models of online matrix factorization were defined in the introduction. 
We remark here that, by transposing the input matrix, and also switching the roles of $L_t$ and $R_t$ and transposing them, we can see that the column arrival model is equivalent to a modification of the row arrival model where we require that $\|L_t\|_{2\to\infty} \le 1$ at all time steps $t$, and we define the factorization value as $\|R_t\|_{1\to 2}$.

Since our online factorization algorithms are deterministic, we can assume, without loss of generality, that the new rows/columns that arrive are chosen adaptively. Recall, however, that our applications to online private query release, and online discrepancy minimization require oblivious adversaries.

It will be convenient to assume that our online factorization algorithms are given the factorization value $\gamma$ ahead of time. This gives us a type of ``decision version'' of the online factorization problem. More precisely, we define the $(\gamma, C)$-bounded online factorization problem as follows. The algorithm receives the values $\gamma$ and $C$ at time step $0$, and decides on the initial matrix $R_0$. Then, at each time step $t \in [m]$, the $t$-th row $q_t$ of $Q$ arrives, similarly to the row arrival model above. Then, at each time step $t \ge 1$ we require that either
\begin{itemize}
    \item the algorithm extends $R_{t-1}$ to $R_t$ by adding new rows, and extends $L_{t-1}$ to $L_t$ by adding one additional row $\ell_t$ so that $\ell_t R_t = q_t$; moreover, it holds that $\|R_t\|_{1\to 2} \le 1$, and $\|L_t\|_{2\to\infty} \le C\gamma$, or

    \item the algorithm asserts $\gamma_2(Q_t) > \gamma$.
\end{itemize}
Note that we can also equivalently require that, in the first case, $\|L_t\|_{2\to\infty} \le 1$ and $\|R_t\|_{1\to 2} \le C\gamma$, since we can replace $L_t$ by $\frac{1}{C\gamma}L_t$ and $R_t$ by $C\gamma R_t$. This is possible since $\gamma$ and $C$ are known to the algorithm. Because of this observation, the row and column arrival models are equivalent in the decision version.

The standard doubling trick allows us to use algorithms solving the $(\gamma, C)$-bounded online factorization problem to get online factorization with factorization value $\tilde{O}(C\gamma_2(Q))$, where $\tilde{O}$ hides at most a logarithmic factor. This is captured by the following lemma.



\begin{lemma}\label{lm:fact-doubling}
    Suppose that we have an online algorithm that solves the $(\gamma, C)$-bounded online factorization problem for a matrix $Q$ and for some fixed $C$ and every $\gamma > 0$. Then, for any $c > 0$, we have an online factorization algorithm for $Q$ in the row arrival model with factorization value $O(C\gamma_2(Q)\log(\gamma_2(Q)/\|q_1\|_\infty)^{\frac12 + c})$. Moreover, we have an online factorization algorithm for $Q^T$ in the column arrival model with factorization value $O(C)$.
\end{lemma}
\begin{proof}
    We first prove the result for the row arrival model, and then describe how to modify the proof for the column arrival model. 
    
    Let us assume that there are no all-zero rows in $Q$, since we can set $\ell_t = 0$ for any $q_t = 0$.
    Let $f(x) = x^{1+2c}\cdot \sum_{y = 1}^\infty \frac{1}{y^{1+2c}}$. Initially, the algorithm is in phase $\phi = 1$. 
    On receiving the first row $q_1$, we set $\gamma^{(1)} = \|q_1\|_\infty = \gamma_2(Q_1)$. We start executing the algorithm for the $(\gamma^{(1)}, C)$-bounded online factorization problem to get the matrices $L^{(1)}_t$ and $R^{(1)}_t$. We set $L_t = \sqrt{f(1)}L^{(1)}_t$ and $R_t = \frac{1}{\sqrt{f(1)}}R^{(1)}_t$. At the first time step $t$ when the algorithm asserts that $\gamma_2(Q_t) > \gamma^{(1)}$, we start the next phase. We increase $\phi$ to $2$, and set $\gamma^{(2)} = 2\gamma^{(1)}$. We then start executing the algorithm for the $(\gamma^{(2)}, C)$-bounded online factorization problem. Suppose that $L^{(1)}$ is the final value of the matrix $L^{(1)}_t$ computed in phase $1$, and $R^{(1)}$ is the final value of $R^{(1)}_t$ computed in phase $1$, and that $L^{(2)}_t$ and $R^{(2)}_t$ are the matrices computed by time step $t$ in phase $2$. We set 
    \[
    L_t = \begin{pmatrix}
        \sqrt{f(1)} L^{(1)} & 0\\
        0 & \sqrt{f(2)} L^{(2)}_t
    \end{pmatrix};\hspace{1em}
    R_t = \begin{pmatrix}
        \frac{1}{\sqrt{f(1)}}R^{(1)}\\
        \frac{1}{\sqrt{f(2)}}R^{(2)}_t
    \end{pmatrix}.
    \]
    In general, we end phase $\phi-1$ and start phase $\phi$ at the first time step $t$ when the algorithm for the bounded online factorization problem asserts $\gamma_2(Q_t) > \gamma^{(\phi-1)}$. We then set $\gamma^{(\phi)} = 2\gamma^{(\phi-1)}$. We start executing the algorithm for the $(\gamma^{(\phi)}, C)$-bounded online factorization problem to get the matrices $L^{(\phi)}_t$ and $R^{(\phi)}_t$. We then add the rows of $R^{(\phi)}_t$ to $R_t$, scaled by $\frac{1}{\sqrt{f(\phi)}}$, and add the rows $L^{(\phi)}_t$ to $L_t$, scaled by $\sqrt{f(\phi)}$, and padded with $0$'s so that the rows of $L^{(\phi)}_t$ and the old rows in $L_t$ have disjoint supports; this guarantees that $L_t R_t = Q_t$.
    
    The fact that, at every step $t$, $\|R_t\|_{1\to 2} \le 1$ follows from the choice of $f(x)$ which guarantees that $\sum_{\phi=1}^\infty \frac{1}{f(\phi)} = 1$. The factorization value is bounded by $C\sqrt{f(\phi)}\gamma^{(\phi)}$ for the final phase $\phi$ reached by the algorithm. Observe that if at time step $t$ the algorithm is in phase $\phi$, then $\gamma_2(Q_t) > \gamma^{(\phi-1)} = \frac12 \gamma^{(\phi)}$. Moreover, $\gamma^{(\phi)} \ge 2^{\phi-1}\gamma^{(1)} = 2^{\phi-1} \|q_1\|_\infty$, since we double $\gamma^{(\phi)}$ in each phase. Therefore, the final phase $\phi$ satisfies $\phi < 2 + \log_2(\gamma_2(Q)/\|q_1\|_\infty)$. Since $f(\phi) \lesssim \phi^{1 + 2c}$, we have
    \(
    \sqrt{f(\phi)}\gamma^{(\phi)}
    \lesssim 
    \gamma_2(Q) (1+\log_2(\gamma_2(Q)/\|q_1\|_\infty))^{\frac12 + c},
    \)
    and this gives the bound on the factorization value that we need.

    The proof of the lemma in the column arrival model is slightly simpler. As observed above, we can assume by rescaling that we have an online factorization algorithm in the column arrival model solving the $(\gamma, C)$-bounded online factorization problem for $Q^T$. I.e., in step $t$ a new column $q_t^T$ of $Q^T$ arrives, and the algorithm extends $R_{t-1}$ with a single new column $r_t$ such that $\|r_t\|_2 \le 1$, and also extends $L_{t-1}$ to $L_t$ by adding new columns, so that $\|L_t\|_{2\to\infty} \le C\gamma$, and $L_t r_t = q_t^T$. If the algorithm fails to do so, it declares $\gamma_2(Q_t^T) = \gamma_2(Q_t) > \gamma$. We now use this algorithm very similarly to the proof for the row arrival model, but without the rescaling. Initially, we are in phase $\phi=1$, and set $\gamma_{(1)} = \|q_1\|_{\infty}$. We start executing the $(\gamma^{(1)}, C)$-bounded online factorization algorithm above to compute matrices $L_t^{(1)}$ and $R_t^{(1)}$ until the algorithm asserts $\gamma_2(Q_t^T) > \gamma^{(1)}$ in some time step $t$. Then we increase the phase number to $\phi=2$, set $\gamma^{(2)} = 2\gamma^{(1)}$, and start executing the $(\gamma^{(2)}, C)$-bounded online factorization algorithm to compute matrices $L_t^{(1)}$ and $R_t^{(1)}$. With $L^{(1)}$ and $R^{(1)}$ denoting the final matrices computed in phase $1$, we set 
    \[
    L_t = \begin{pmatrix}
        L^{(1)} & L^{(2)}_t
    \end{pmatrix};\hspace{1em}
    R_t = \begin{pmatrix}
        R^{(1)} & 0\\
        0 & R^{(2)}_t
    \end{pmatrix}.
    \]
    We continue in this manner: phase $\phi-1$ ends when the algorithm declares $\gamma_2(Q_t^T) > \gamma^{(\phi-1)}$, and then we set $\gamma^{(\phi)} = 2\gamma^{(\phi-1)}$, start phase $\phi$ by running the $(\gamma^{(\phi)}, C)$-bounded online factorization algorithm  to compute matrices $L_t^{(\phi)}$ and $R_t^{(\phi)}$. We form $L_t$ and $R_t$ by adding the columns of $L^{(\phi)}_t$ to $L_t$, and adding the columns of $R^{(\phi)}_t$ to $R_t$ so that  $R^{(\phi)}_t$ forms a new diagonal block in $R_t$. It is easy to see that we have $L_t R_t = Q^T_t$.

    Now the fact that, at every step $t$, $\|R_t\|_{1\to 2} \le 1$ is trivial. Let us fix a time step $t$, and suppose the algorithm is in phase $\phi$ at this step. Then, as above, $\gamma_2(Q_t^T) > \frac12 \gamma^{(\phi)}$, and, since  we double $\gamma^{(\phi)}$ in each phase, we have
    \[
    \|L_t\|_{2\to\infty}^2 \le C^2 \sum_{\psi = 1}^\phi (\gamma^{(\psi)})^2
    = C^2 (\gamma^{(\phi)})^2 \sum_{\psi = 1}^\phi 4^{\psi- \phi}
    \le \frac{4}{3} C^2 (\gamma^{(\phi)})^2 < \frac83 C^2\gamma_2(Q_t^T)^2.\qedhere
    \]
 \end{proof}

Note that in the proof of Lemma~\ref{lm:fact-doubling} we can replace $f(x)$ with $f(x) = \tilde{f}(x)\sum_{y=1}^\infty \frac{1}{\tilde{f}(y)}$ for any $\tilde{f}(x)$ for which the sum $\sum_{x=1}^\infty \frac{1}{\tilde{f}(x)}$ converges. For example, we can choose \[\tilde{f}(x) = x \log(x) \log \log(x) \ldots \log^{(i)}(x)^{1+c},\] where $i$ is an integer, $c > 0$, and $\log^{(i)}(x)$ is the $i$-th iterated logarithm. This allows a corresponding improvement in the logarithmic term of Lemma~\ref{lm:fact-doubling}.

The proof of our main online factorization result, Theorem~\ref{thm:fact-main}, which appears in Section~\ref{sec:gen-fact}, goes via Lemma~\ref{lm:fact-doubling}. Before we give this proof, we describe our two main applications of Theorem~\ref{thm:fact-main}. For convenience we restate Theorem~\ref{thm:fact-main}.

\factmain*

\subsection{Application: Online Private Query Release}

In this subsection, we discuss the application of our factorization result to online private query release, and prove Theorem~\ref{thm:main-priv}. As a preliminary, let us first define the Gaussian mechanism ~\cite{DinurNissim03,DworkN04,DworkKMMN06OurDataOurselves}. The mechanism returns the true query answers perturbed by unbiased spherical Gaussian noise, scaled proportionally to the $\ell_2$ sensitivity of the queries. Furthermore, the mechanism handles online queries, even adaptively chosen ones, with the same privacy guarantee as with offline queries. 

\begin{lemma}[The Gaussian mechanism]
    \label{lem:gaussian}
    Let $\mathcal{Q} = (q_1, \ldots, q_m)$ be a sequence of statistical queries arriving online, determined by a potentially adaptive data analyst. Suppose that the $\ell_2$ sensitivity of $\mathcal{Q}$ satisfies the upper bound $\Delta \mathcal{Q} \eqdef \max_{x \sim x'} \|\mathcal{Q}(x) - \mathcal{Q}(x')\|_2\leq \bar{\Delta}$, and let 
    \(
    \sigma^2 \eqdef \frac{2\bar{\Delta}^2(\varepsilon + \ln (1/\delta))}{\varepsilon^2}.
    \)
    The mechanism that, upon receiving $q_t$, outputs $q_t(x) + z_t$, where $z_t$ is an independent sample from $\mathcal{N}\left(0, \sigma^2\right)$, is $(\varepsilon,\delta)$-differentially private for adaptive and non-adaptive analysts.
\end{lemma}

In the offline setting, a factorization $LR = Q$ of the query matrix $Q$ of a set of statistical queries $\mathcal Q$ achieving value $\gamma = \|L\|_{2\to\infty}\|R\|_{1\to 2}$ can be used to give a query release algorithm which is $(\alpha,\beta)$-accurate for $\alpha = O\left(\frac{\gamma\sqrt{\log m}}{ n}\right)$ (ignoring the dependence on the privacy parameters and on $\beta$)~\cite{EdmondsNU20}. This is done by releasing the answers $\frac{1}{n}Rh + z$ for $z \sim \mathcal{N}\left(0, \sigma^2 I\right)$ to the ``strategy queries'' defined by $R$, and post-processing them to get the final output $\frac{1}{n}L(Rh + z) = \frac1n Qh + Lz$. We need to set $\sigma^2 = O\left(\frac{\|R\|_{1\to 2}^2}{n^2}\right)$ since the sensitivity of the queries defined by $R$ is $\frac{2\|R\|_{1\to 2}}{n}$, and then we can observe that the maximum variance of any coordinate of $Lz$ is $\|L\|_{2\to\infty}^2 \sigma^2$, from which the bound on the error $\alpha$ follows by a standard concentration argument. 

In the online setting we can use an online factorization algorithm achieving factorization value $\gamma$  in a similar way to get similar error bounds. This is captured by the following lemma.

\begin{lemma}\label{lm:priv-to-fact}
    Let $\mathcal{Q} = (q_1, \ldots, q_m)$ be a sequence of statistical queries arriving online, chosen by a non-adaptive data analyst. If there is an online factorization algorithm for the query matrix $Q$ of of $\mathcal Q$ in the row arrival model achieving factorization value $\gamma$, then, for any $\delta \le e^{-\varepsilon}$, there is an interactive algorithm that is $(\varepsilon,\delta)$-differentially private, and $(\alpha,\beta)$-accurate for any $\beta > 0$, and $\alpha = O\left(\frac{\gamma\sqrt{\log(m + 1/\beta)\log(1/\delta)}}{\varepsilon n}\right)$.
\end{lemma}
\begin{proof}
    Let $Q$ be the query matrix of the queries, and where $q_t$ is also the $t$-th row of this matrix. Suppose that the online factorization algorithm computes the factorization $L_tR_t = Q_t$, and, in particular, computes a row $\ell_t$ at time step $t$ such that $\ell_t R_t = q_t$.
    
    Consider the online sequence $\mathcal{R}$ of statistical queries defined by the rows of $R_m$. I.e., each row $r$ of $R_m$ defines the query $r(x) \eqdef \frac{\langle r,h\rangle}{n}$, where $x$ is the private dataset, $n$ is its size, and $h$ is its histogram; $\mathcal{R}$ consists of all such queries, in the order in which rows are added to $R_t$. Thus the vector of true answers to the queries in $\mathcal{R}$ is $R_m h$. Since any two histograms $h$ and $h'$, corresponding, respectively, to neighboring datasets $x$ and $x'$, satisfy $\|h-h'\|_1 \le \frac2n$, the $\ell_2$ sensitivity of $\mathcal{R}$ is at most $\frac2n$ because $\|R_m (h - h')\|_2 \le \norm{R_m}_{1\to 2}\|h-h'\|_{1} \le \frac{2}{n}$. At every time step $t$, we use the Gaussian mechanism to answer the private queries $r(x)$ for each row $r$ that is in $R_t$ but not in $R_{t-1}$. Thus, by time step $t$, the algorithm has computed $R_t h + \mathcal{N}(0,\sigma^2 I)$ where $\sigma^2 \eqdef \frac{8(\varepsilon + \ln (1/\delta))}{\varepsilon^2n^2}$ is such that the Gaussian mechanism satisfies $(\varepsilon,\delta)$-differential privacy (by Lemma~\ref{lem:gaussian}). Then our mechanism outputs $\ell_t (R_t h + \mathcal{N}(0,\sigma^2 I))$ as the answer to $q_t$, which is a post-processing of $R_t h + \mathcal{N}(0,\sigma^2 I)$.

    The privacy guarantee follows by the privacy of the (online) Gaussian mechanism, and holds even for adaptive data analysts. For the accuracy guarantee, note that for non-adaptive data analysts the queries $(q_1, \ldots, q_m)$, and, therefore, the row vectors $\ell_1, \ldots, \ell_m$, are independent of the Gaussian noise. The accuracy claim then follows from a standard Gaussian concentration bound and the union bound.
\end{proof}

Theorem~\ref{thm:main-priv} (restated below) now follows as a corollary from Theorem~\ref{thm:fact-main} and Lemma~\ref{lm:priv-to-fact}.
\thmmainpriv*


Edmonds, Nikolov, and Ullman~\cite{EdmondsNU20} have shown that any (offline)  $(\varepsilon,\delta)$-differentially private algorithm, and for any large enough $n$ and small enough constant $\varepsilon$, $\delta$, and $\beta$, can only be $(\alpha,\beta)$-accurate for
\[
\alpha = \Omega\left(\frac{\gamma_2(Q)}{\varepsilon n}\right).
\]
Thus, the error bound in Theorem~\ref{thm:main-priv} is competitive against an optimal private query release algorithm up to a competitive ratio polynomial in $\log(mU/(\beta\delta))$. 

\subsection{Application: Online Discrepancy Minimization}



In this subsection, we discuss the application of our online factorization algorithm to online discrepancy minimization, and prove Theorem~\ref{thm:disc-main}. Towards this goal, let us first recall the definition of subgaussian random variables.

\begin{definition}
    A real-valued random variable $y$ is $\sigma$-subgaussian (for $\sigma > 0$) if $\E[\exp({y^2/\sigma^2})] \le 2$. A random variable $z$ taking values in $\R^d$ is $\sigma$-subgaussian if, for each $\theta \in \R^d$ with $\|\theta\|_2 =1$, the random variable $\langle z, \theta \rangle$ is $\sigma$-subgaussian.
\end{definition}

We recall that a real-valued $\sigma$-subgaussian random variable $y$ satisfies, for any $t \ge 0$, the tail bound $\Pr[|y| \ge t] \le 2e^{-t^2/\sigma^2}$~\cite[Proposition 2.6.1]{Vershynin}. As an easy consequence, a $\sigma$-subgaussian random vector $z$ satisfies 
\begin{equation}\label{eq:subgaussian-tail}
    \forall \theta \in \R^d \setminus \{0\}: 
    \Pr[|\langle z, \theta \rangle| \ge t] \le \exp\left({-\frac{t^2}{\sigma^2 \|\theta\|_2^2}}\right).
\end{equation}

A basic result in discrepancy theory, due to Banaszczyk~\cite{Bana98} shows that for any vectors $v_1, \ldots, v_N$, each of $\ell_2$ norm at most $1$, there is a distribution of random signs $x_1, \ldots, x_N \in \{-1, +1\}$ so that the random vector $\sum_{i=1}^N x_i v_i$ is $O(1)$-subgaussian.\footnote{This re-formulation of Banaszczyk's theorem is due to Dadush, Garg, Lovett, and Nikolov~\cite{DGLN}.} This result can be used to show that any $m\times N$ matrix $A$ with factorization $\gamma = \|L\|_{2\to \infty} \|R\|_{1\to 2}$ has discrepancy at most $\disc(A) = O(\gamma\sqrt{\log m})$. Notice first that we can assume that $\|R\|_{1\to 2} = 1$ and $\|L\|_{2\to\infty} = \gamma$ by appropriately rescaling $L$ and $R$. Then, since the columns of $R$ are vectors of $\ell_2$ norm at most $1$, Banaszczyk's theorem implies that we can find random signs $x \in \{-1, +1\}^N$ so that $Rx$ is $O(1)$-subgaussian. Since every row of $L$ is a vector of $\ell_2$ norm at most $\gamma$, the discrepancy bound follows with high probability by \eqref{eq:subgaussian-tail} and a union bound.

To adapt this argument to the online setting, we first need an online version of Banaszczyk's theorem. The next lemma, due to Kulkarni, Reis, and Rothvoss~\cite{KulkarniRR24}, is sufficient for our purposes.
\begin{lemma}\label{lm:bana-online}
    There is an online algorithm such that for any sequence of vectors $v_1, \ldots, v_N$, each of $\ell_2$ norm at most $1$, and decided by an oblivious adversary, the algorithm chooses online random signs $x_1, \ldots, x_N$ such that, for all $t \in [N]$, the random vector $x_1 v_1 + \ldots + x_t v_t$ is $10$-subgaussian.
\end{lemma}


Our reduction from online discrepancy minimization to online factorization is now given by the next lemma. 

\begin{lemma}\label{lm:disc-to-fact}
    Suppose that the columns $a_1, \ldots, a_N$ of an $m\times N$ matrix $A$ arrive online in the oblivious model. If there is an online factorization algorithm for $A$ in the column arrival model achieving factorization value $\gamma$, then there is an online algorithm that computes a sequence of signs $x_1, \ldots, x_N \in \{-1,+1\}$ such that, with probability at least $1-\frac{1}{N}$, 
    \[
    \max_{t=1}^N \|x_1 a_1 + \ldots + x_t a_t\|_\infty = O(\gamma \sqrt{\log mN}).
    \]
\end{lemma}
\begin{proof}
    Suppose that the online factorization algorithm computes the factorization $L_t R_t = A_t$, where $A_t$ is the matrix with columns $a_1, \ldots, a_t$. In particular, at time $t$, the factorization algorithm computes a column vector $r_t$ so that $L_t r_t = a_t$, $\|L_t\|_{2\to\infty} \le \gamma$, and $\|r_t\|_2 \le 1$. Below, when writing $L_t r_i$ for $i \le t$, we assume $r_i$ has been padded with $0$ entries so that the matrix vector product is well defined. 
    
    We use the algorithm from Lemma~\ref{lm:bana-online} with $r_1, \ldots, r_N$, in the order in which they are added to $R_N$, to compute random signs $x_1, \ldots, x_N \in \{-1,+1\}$ so that the random vector $u_t \eqdef \sum_{i = 1}^t x_i v_i$ is $10$-subgaussian for all $t$. Notice that 
    \[
    x_1 a_1 + \ldots + x_t a_t
    = 
    L_t (x_1 r_1 + \ldots + x_t r_t) = L_t u_t.
    \]
    Then, by \eqref{eq:subgaussian-tail} and union bound, we have that
    \[
    \Pr[\|x_1 a_1 + \ldots + x_t a_t\|_\infty \ge t]
    = 
    \Pr[\|L_t u_t\|_\infty \ge t]
    \le 2m e^{-t^2/(100\gamma^2)}.
    \]
    Setting $t$ to be a large enough multiple of $\gamma \ln(2mN)$ and applying a union bound over all $t \in [N]$ gives the result.
\end{proof}

As noted in the Introduction, while the algorithm in Lemma~\ref{lm:bana-online} is not efficient, another online algorithm (the Self-Balancing Walk) of Alweiss, Liu, and Sawnhey~\cite{alweiss2020discrepancy} runs in polynomial time, and gives the weaker guarantee that, conditional on an event that happens with probability $1-\frac{1}{2N}$, $x_1 v_1 + \ldots + x_t v_t$ is $O(\sqrt{\log N})$-subgassuan for each $t$. Using this algorithm instead in the proof of Lemma~\ref{lm:disc-to-fact} would give a polynomial time algorithm achieving the discrepancy bound 
\[
    \max_{t=1}^N \|x_1 a_1 + \ldots + x_t a_t\|_\infty = O(\gamma \log mN).
\]

Theorem~\ref{thm:disc-main} (restated below) now follows from Lemma~\ref{lm:disc-to-fact} and Theorem~\ref{thm:disc-main}. 

\thmmaindisc*

The slightly weaker bound \eqref{eq:disc-main-efficient} achievable by a polynomial time algorithm follows by using the efficient variant of Lemma~\ref{lm:disc-to-fact} based on the Self-Balancing Walk.




\section{General Online Factorization Algorithms}
\label{sec:gen-fact}

In this section we prove Theorem~\ref{thm:fact-main}. 
By Lemma~\ref{lm:fact-doubling}, Theorem~\ref{thm:fact-main} is implied by the following lemma.

\begin{lemma}\label{lm:fact-main}
    There is an online algorithm that solves the $(\gamma, C)$-bounded online factorization problem for any $m\times N$ matrix $Q$ with $C = O\left(  \gamma_2(Q)\log(N)^3\right)$.
\end{lemma}

Our approach is to first give an algorithm that achieves $\|R_t\|_{F} \le \sqrt{N}$ rather than the stronger guarantee $\|R_t\|_{1\to 2}\le 1$. I.e., the algorithm only guarantees that the average squared $\ell_2$ norm of columns of $R_t$ is bounded by $1$, rather than the maximum squared $\ell_2$ norm. Using this algorithm as a base, we develop several reductions that allows us to solve the bounded online factorization problem. 

\subsection{Online Average Factorization}\label{sec:online-average}

Let us define the $(\gamma,C)$-bounded online \emph{average} factorization problem the same way we defined the $(\gamma,C)$-bounded online factorization problem, with the modification that, at every time step $t$, $R_t$ must satisfy $\|R_t\|_F \le \sqrt{N}$.

Note that $\min\{\|L\|_{2\to\infty} \|R\|_F: LR = Q\} = \gamma_F(Q^T)$, and that $\gamma_2(Q) \ge \frac{1}{\sqrt{N}}\gamma_F(Q^T)$ for any $m\times N$ matrix $Q$. We will give an algorithm for the online average factorization problem that will either extend $R_{t-1}$ to $R_t$ and $L_{t-1}$ to $L_t$ so that $\|L_t\|_{2\to\infty} \le \gamma$ and $\|R_t\|_{F} \le \sqrt{N}$, or will certify the stronger inequality that $\gamma_F(Q^T) > \gamma\sqrt{N}$ rather than just $\gamma_2(Q) > \gamma$ as the definition of problem requires. This algorithm is the first step in proving Lemma~\ref{lm:fact-main}. 

It will be convenient to reformulate the online average factorization problem as a problem of solving a semi-definite program (SDP) online. This is done by the following lemma. 
\begin{lemma}\label{lm:online-ellipsoid}
    Suppose the rows of an $m\times N$ query matrix $Q$ arrive online, and there is an online algorithm that outputs a sequence of positive semidefinite (PSD) matrices $X_0, X_1, X_2, \ldots$ that, at each time step $t\ge 1$, either reports that there is no positive semidefinite matrix $X$ such that $\tr(X) \le \gamma^2 N$ and $X\succeq q_i^T q_i$ for all $i \in [t]$, or else outputs a new PSD matrix $X_t$ in the sequence such that
    \begin{itemize}
        \item $X_{t} \succeq X_{t-1}$;
        \item $ X_t \succeq  q_t^T q_t$;
        \item $\tr(X_t)\le C^2 \gamma^2 N$.
    \end{itemize}
    Then there exists an online factorization mechanism that solves the $(\gamma,C)$-bounded online \emph{average} factorization problem for $Q$.
\end{lemma}

Before we prove Lemma~\ref{lm:online-ellipsoid}, we need another lemma.

\begin{lemma}\label{lm:fact-psd}
    For any real number $\gamma >0$, $m\times N$ matrix $Q$ with rows $q_1, \ldots, q_m$ (as row vectors), and $k\times N$ matrix $R$, there exists a factorization $Q = LR$ with $\|L\|_{2\to \infty} \le \gamma$ if and only if $\gamma^2 R^TR \succeq q_i^T q_i$ for every $i\ \in [m]$. Moreover, if the latter property is satisfied, then we can take $L=QR^+$ for the Moore-Penrose pseudoinverse $R^+$ of $R$.
\end{lemma}
\begin{proof}
    First, suppose there is a factorization $Q = LR$ with $\|L\|_{2\to \infty} \le \gamma$. Let $\ell_i$ be the $i$-th row of $L$ (taken as a row vector). We have $\ell_i R = q_i$. Moreover, $\gamma^2 I \succeq \ell_i^T \ell_i$ because $\ell_i^T \ell_i$ is a rank $1$ matrix with its only nonzero eigenvalue equal to $\|\ell_i\|_2^2 \le \gamma^2$. Therefore, $\gamma^2 R^T R \succeq R^T (\ell_i^T \ell_i) R = q_i^T q_i$. 
    
    Next, we show that if $\gamma^2 R^T R\succeq q_i q_i^T$ for every $i \in [m]$, then there exists an $L$ with $Q = LR$ and $\|L\|_{2 \to \infty} \le \gamma$, and, moreover, this holds for $L = QR^+$. Since $R R^+$ is an orthogonal projection matrix (onto the column span of $R$), we have $R R^+ \preceq I$.  Together with the assumption $\gamma^2 R^T R\succeq q_i^T q_i$, this implies
    \[
    \gamma^2 I \succeq \gamma^2 (R^+)^T R^T R R^+ \succeq (R^+)^Tq_i^T q_i R^+ = (q_i R^+)^T (q_i R^+) = \ell_i^T \ell_i.
    \]
    Therefore, $\|\ell_i\|_2^2 \le \gamma^2$. 
    
    It remains to verify that $LR = QR^+ R = Q$. Notice that $\gamma^2 R^T R \succeq q_i^T q_i$ implies that $q_i$ is in the rowspan of $R$ (otherwise there is a vector $v$ such that $Rv = 0$ and $q_i v \neq 0$, which implies $0 = v^T R^T Rv < (q_i v)^2$). Since $R^+ R$ is an orthogonal projection onto the rowspan of $R$, we then have $q_i R^+ R = q_i$. Because this  holds for every $i$, we have shown $QR^+ R = Q$.
\end{proof}

Notice that Lemma~\ref{lm:fact-psd} immediately implies that 
\begin{equation}\label{eq:fact-psd}
\gamma_F(Q^T) = \min\{\|L\|_{2\to\infty}\|R\|_F: LR = Q\}
= \min\{\sqrt{\tr(X)}: X\succeq q_i^T q_i\ \forall i \in [m], X \succeq 0\},
\end{equation}
where $X$ ranges over $N\times N$ matrices.
Indeed, by rescaling $L$ and $R$, we can always assume that the optimal factorization $LR = Q$ satisfies $\|L\|_{2\to\infty} = 1$ and $\|R\|_{F} = \gamma_F(Q^T)$. By Lemma~\ref{lm:fact-psd}, such a factorization exists if and only if there is a $k\times N$ matrix $R$ (for some $k$) such that $R^T R\succeq q_i^T q_i$ for all $i \in [m]$, and $\gamma_F(Q^T) = \|R\|_F = \sqrt{\tr(R^T R)}$. The change of variables $X = R^T R$ now gives \eqref{eq:fact-psd}. The proof of Lemma~\ref{lm:online-ellipsoid} translates this proof to the online setting.

\begin{proof}[Proof of Lemma~\ref{lm:online-ellipsoid}]
    Note first that if the algorithm asserts, at any time step, that $\min\{\tr(X): X\succeq q_i^T q_i\ \forall i \in [t], X \succeq 0\} > \gamma^2 N$, then, by \eqref{eq:fact-psd}, $\gamma_2(Q_t) \ge \frac{1}{\sqrt{N}}\gamma_F(Q_t^T)> \gamma$.

    Let us now consider the steps when the algorithm outputs a new $X_t$.
    We will construct the matrices $R_t$ online so that the invariant $R_t^T R_t = \frac{1}{C^2\gamma^2} X_t$ is maintained.
    We set $R_0$ to be such that $R_0^T R_0 =  \frac{1}{C^2\gamma^2} X_0$, which is possible because $X_0 \succeq 0$. Since $X_{t}-X_{t-1}\succeq 0$, we can write $M_{t}^T M_{t} = \frac{1}{C^2\gamma^2}(X_{t} - X_{t-1})$ for some matrix $M_{t}$. We form $R_{t}$ by adding to $R_{t-1}$ the rows of $M_{t}$. Then it is clear that $R_{t}^T R_{t} = R_{t-1}^T R_{t-1}  + M_t^T M_t = \frac{1}{C^2\gamma^2}X_{t}$, and our invariant is maintained. 

    Notice that, since $X_t \succeq X_s$ for any $s \le t$, and since we assume that $X_{s} \succeq q_{s}^Tq_{s}$ holds for every time step $s \le t$, we also have that $C^2\gamma^2R_t^T R_t = X_{t} \succeq q_{t}^Tq_{t}$. Then, by Lemma~\ref{lm:fact-psd} there is a matrix $L$ so that $LR_t = Q_t$ and $\|L\|_{2\to\infty} \le C\gamma$. Taking $\ell_{t}$ to be the last row of $L$ gives $\ell_{t} R_{t} = q_{t}$ and $\|\ell_t\|_2 \le C\gamma$. (More explicitly, we can take $\ell_{t} = q_{t} R_{t}^+$, where $R_{t}^+$ is the pseudoinverse of $R_{t}$.) Finally, $\|R_{t}\|_{F} = \sqrt{\tr(R_t^T R_t)} = \frac{\sqrt{\tr(X_t)}}{C\gamma} \le \sqrt{N}$. 
\end{proof}

Using the reformulation in Lemma~\ref{lm:online-ellipsoid}, we give an online factorization algorithm to solve the $(\gamma,C)$-bounded average factorization problem. We use the semidefinite program (SDP) \eqref{eq:fact-psd} and also derive its dual, and use both to derive a simple primal-dual algorithm.

\begin{lemma}\label{lm:online_F_transpose_factorization}
For any $\gamma > 0$, there is an online factorization algorithm that solves the $(\gamma, \sqrt{3})$-bounded online average factorization problem for any matrix $Q$.
\end{lemma}
\begin{proof}
We give an online algorithm for the reformulation in Lemma~\ref{lm:online-ellipsoid}. 
Let us define a function 
\[
X(y_1, y_2,\ldots,y_m)=\gamma^2 I+2{\gamma}\sqrt{N}\sqrt{\sum_{i=1}^my_iq_i^T q_i},
\]
where $\sqrt{\sum_{i=1}^m y_iq_i^T q_i}$ is the positive semidefinite matrix square root. While $X$ is a function of $y_1, \ldots, y_m$, we will make sure that at time step $t$, $y_{t+1} = \ldots = y_m = 0$. We will write $X(y_1, \ldots, y_t)$ for $X(y_1, \ldots, y_m)$ at time step $t$ in order to emphasize that, at step $t$, $X$ is only a function of $y_1, \ldots y_t$. Our initial matrix is $X_0 = \gamma^2 I$.  We run the following algorithm upon arrival of $q_t$:
\begin{figure}[H]\label{fig:alg-avg}
    \begin{framed}
    \begin{algorithmic}
    \INDSTATE[0]Receive $q_t$
    \INDSTATE[0]Initialize $y_t=0$
    \INDSTATE[0]\textbf{If} $X(y_1,\ldots,y_t)\succeq q_t^Tq_t$ is not true:
    \INDSTATE[1]Increase $y_t$ until $X(y_1,\ldots,y_t)\succeq q_t^Tq_t$ or $\tr(X(y_1,\ldots,y_t))=3{\gamma}^2N$.
    \INDSTATE[0]\textbf{If} $X(y_1,\ldots,y_t)\not\succeq q_t^Tq_t$:
    \INDSTATE[1]Assert that $\min\{\tr(X): X\succeq q_i^T q_i \ \forall i \in [t], X \succeq 0\} > \gamma^2 N$, and terminate.
    \INDSTATE[0]\textbf{Otherwise}: Output $X_t\leftarrow X(y_1,\ldots,y_t)$.
    \end{algorithmic}
    \end{framed}
    \vspace{-5mm}\caption{Online SDP algorithm} 
\end{figure}

From the definition of the algorithm, it is clear that when $X_t$ is output, the conditions $\tr(X_t) \le 3\gamma^2 N$ and $X_t \succeq q_t^T q_t$ are satisfied. To show that $X_t \succeq X_{t-1}$, recall the classical fact that the square root function is operator monotone on $[0,\infty)$, i.e., $A\succeq B \implies \sqrt{A}\succeq \sqrt{B}$ for positive semidefinite matrices $A$, $B$ (Proposition V.1.8~in~\cite{Bhatia}). Now $X_t \succeq X_{t-1}$ follows from the observation that the matrix $\sum_i y_i q_i^T q_i$ is monotonically increasing (in the positive semidefinite order) because the $y_i$ values only increase.

Let us now consider a dual formulation of $\gamma_F(Q^T)$:
\begin{align}
\gamma_F(Q_t^T)
&=
\min\{\sqrt{\tr(X)}: X\succeq q_i^T q_i\ \forall i \in [t], X \succeq 0\}\notag\\
&=
\max\left\{\left[\tr\left(\sqrt{\sum_{i=1}^t y_iq_i^Tq_i}\right)\right]:\sum_{i=1}^ty_i\leq 1, y_1, \ldots, y_t \ge 0\right\}, \label{eq:sdp-duality}
\end{align}
where $\sqrt{\sum_{i=1}^t y_iq_i^Tq_i}$ is again the positive semidefinite square root. In fact we only need the $\ge$ direction of the second equality, which we prove at the end, for completeness. The proof of equality can be found, e.g., in~\cite{NikolovT24}.

It remains to show that when the algorithm asserts that $\min\{\tr(X): X\succeq q_i^T q_i\ \forall i \in [m], X \succeq 0\} > \gamma^2 N$, this assertion is correct. This happens when we are done increasing $y_t$ but $X(y_1, \ldots y_t) \not\succeq q_t^T q_t$. Therefore, it must be that $\tr(X(y_1,\ldots, y_t)) = 3\gamma^2 N$. 
By the definition of $X(y_1, \ldots, y_t)$, we then have that 
\[
\tr\left(\sqrt{\sum_{i=1}^t y_iq_i^Tq_i}\right) = \gamma \sqrt{N}.
\]
If we can show that $\sum_{i=1}^t y_i < 1$, \eqref{eq:sdp-duality} would imply that the assertion that $\min\{\tr(X): X\succeq q_i^T q_i\ \forall i \in [m], X \succeq 0\} > \gamma^2 N$ is true. Consider the time step when $y_s$ is being increased for some $s \le t$. We have that
\begin{align*}
    \frac{\partial \tr(X(y_1, \ldots, y_s))}{\partial y_s}&={\gamma\sqrt{N}}\,q_s\left(\sum_{i=1}^s y_iq_i^Tq_i\right)^{-1/2}q_s^T\\
    &\ge 2{\gamma}^2 N\, q_sX(y_1, \ldots, y_s)^{-1}q_s^T,
\end{align*}
where the inequality follows since $\left(\sum_{i=1}^s y_iq_i^Tq_i\right)^{-1/2} \succeq 2\gamma \sqrt{N}X(y_1, \ldots, y_s)^{-1}$.
While $y_s$ is being increased we have $X(y_1, \ldots, y_s) \not\succeq q_s^T q_s$. We claim this is equivalent to $q_s^TX(y_1, \ldots, y_s)^{-1}q_s > 1$. To see this, note that, for $X = X(y_1, \ldots, y_s)$, 
\[
X \succeq q_s^T q_s \iff I \succeq X^{-1/2} q_s^T q_s X^{-1/2} 
\iff q_s X^{-1} q_s^T \le 1.
\]
Therefore, we know that $\frac{\partial \tr(X(y_1, \ldots, y_s))}{\partial y_s} > 2\gamma^2 N$. Integrating this inequality, we have that, for any $s \in [t]$, at the time we stop updating $y_s$, we have
\[
\tr(X(y_1,\ldots, y_s))- \tr(X(y_1,\ldots, y_{s-1})) > 2\gamma^2 N y_s,
\]
and, $\tr(X(y_1)) - \gamma^2 N > 2\gamma^2 N y_1$. Summing these inequalities together, we get that 
\[
2\gamma^2 N = \tr(X(y_1, \ldots, y_t)) - \gamma^2 N > 2\gamma^2 N \sum_{i=1}^t y_i,
\]
and $\sum_{i=1}^t y_i < 1$ as we wanted to show.

To finish the proof, we give a short derivation of the $\ge$ direction of \eqref{eq:sdp-duality}, i.e., of weak duality. Take any $y_1,y_2,\ldots, y_t \ge 0$ s.t. $\sum_{i=1}^t y_i\leq 1$, and also any positive definite $X$ such that $X \succeq q_i^T q_i$ for all $i \in [t]$. As we saw above, this implies $q_i X^{-1} q_i^T \le 1$. Let $D$ be a $t\times t$ diagonal matrix with entries $D_{i,i} = \sqrt{y_i}$. We have
\begin{align*}
    \tr\left(\sqrt{\sum_{i=1}^t y_iq_i^Tq_i}\right)^2
    &=\norm{DQ_t}_{tr}^2\\
    &=\norm{DQ_t\sqrt{X}^{-1}\sqrt{X}}_{tr}^2
    \le \|\sqrt{X}\|_F^2 \|DQ_t\sqrt{X}^{-1}\|_F^2
    = \tr(X)\sum_{i=1}^t y_iq_iX^{-1}q_i^T
    \leq \tr(X).
\end{align*}
Above, $\|\cdot\|_{tr}$ is the trace norm, i.e., the sum of the singular values, and the inequality is a matrix variant of the Cauchy-Schwarz inequality (Corollary~IV.2.6~of~\cite{Bhatia} with the trace norm and $p=q=2$). The bound for all positive semidefinite $X$ follows by continuity. 
%
%
%
%
%
\end{proof}
\begin{rem}
    One may implement the continuous increment of $y_t$ via binary search to achieve an arbitrarily close approximation of the algorithm in Figure~\ref{fig:alg-avg}. 
\end{rem}

\subsection{Reduction to Average Factorization}

In this section we build our reduction from online factorization to average online factorization. As mentioned in the introduction, the reduction is based on an idea of Muthukrishnan and Nikolov~\cite{MuthukrishnanN12}: to get a factorization certifying a bound on the $\gamma_2(Q)$ norm from an average factorization $Q = LR$ where $R$ is bounded only in Frobenius norm, we can remove from $R$ the columns whose $\ell_2$ norm is larger than $\sqrt{2} \|R\|_F/\sqrt{N}$, and factor only the corresponding columns of $Q$; the other columns of $Q$ can then be factored recursively. 
Here, we implement the same approach in an online setting. To do so, we introduce two transformations that allow the algorithm to handle insertion and one-time deletion of input matrix columns. Then we use this ability to ``kick out'' columns whose current norm in the right matrix of the factorization is too large, and we handle these deleted columns recursively.

We start with some necessary definitions.
\begin{definition}
    Given a matrix $R$ with columns indexed by a set $U$, we define, for any $x \in U$, $s_R(x)$ to be the $\ell_2$ norm of the column indexed by $x$.
\end{definition}

Our reduction will rely on the ability to insert and delete columns of a matrix whose rows arrive online.  For us, deleting a column at a time step $t$ will mean that we zero out the entries of this column in rows that arrive at step $t$ or later. Inserting a column will have the conventional meaning of adding a column to the matrix. Next we define some useful notation that allows us to refer to these operations.  
\begin{definition}
    For an $m \times N$ matrix $Q$ with columns indexed by $U$, and a sequence $\mathcal{U} = (U_0, \ldots, U_m)$ where $U_0,\ldots,U_m\subseteq U$, define $Q|_{\mathcal{U}}$ to be the matrix obtained by replacing all entries in the $i$-th row of $Q$ that do not correspond to elements of $U_i$ by $0$. Moreover, if $U_0\subseteq U_1\ldots\subseteq U_m\subseteq U$, we say this sequence is insertion-only.

    For a single set $U' \subseteq U$, let $Q|_{U'}$ be the matrix obtained by replacing all entries in columns not corresponding to elements of $U'$ by $0$.
\end{definition}
We now introduce a blackbox transformation that augments and online average factorization algorithm so that it can also support insertion of elements. To this end, we define one more variant of the online average factorization problem. In the $(\gamma,C)$-bounded online average factorization problem \emph{with insertions}, the algorithm receives, at each time step $t$, both the row $q_t$ of $Q$, and also a subset $U_t$ of the column indices of $Q$, such that $\mathcal{U} = (U_0, U_1, U_2, \ldots)$ is insertion-only. At each time step $t$, the algorithm should either assert that $\gamma_2(Q_t) > \gamma$, or compute an online factorization $L_t R_t = Q_t|_{U_t}$ such that $\norm{L_t}_{2\to\infty} \le C \gamma$ and $\norm{R_t}_{F} \le \sqrt{|U_t|}$. To simplify some of the future arguments, we also require that all rows in $R_t$ are supported on $U_t$. 

\begin{lemma}[Handling insertions] \label{lm:insertable}
    Let $Q$ be an $m\times N$ matrix with columns indexed by $U$. Suppose that for any fixed $U' \subseteq U$ there is an online factorization algorithm that solves the $(\gamma,C)$-bounded online average factorization problem for $Q|_{U'}$.
    Then there exists an online factorization algorithm
    that solves the $(\gamma,C')$-bounded online average factorization problem with insertions for $Q$, with $C' \le C\log_2(2N)$.
\end{lemma}
\begin{proof}
    We will simultaneously run several instances of the given algorithm, each with a matrix $Q|_{V}$, for some $V \subseteq U$. If an instance is running on $Q|_{V}$, we say that $V$ is the universe for this instance, and say the size of the instance is $\abs{V}$. We say an instance is \emph{active} once it is initialized, until at some point we may stop executing it permanently, at which point it becomes inactive. If any instance asserts that $\gamma_2(Q_t|_V) > \gamma$, we assert that $\gamma_2(Q_t) > \gamma$, which is valid since $\gamma_2(Q_t) \ge \gamma_2(Q_t|_V)$ for any $V \subseteq U$.
    
    Let $\mathrm{bin}(x)_k$ be the $k$-th least significant bit in the binary representation of a positive integer $x$, where $\mathrm{bin}(x)_0$ is the least significant bit. We will ensure the following key properties:
    \begin{enumerate}
        \item All instances that are active at time $t$ have disjoint universes, and the union of the universes is $U_t$. Moreover, all instances (active and inactive) have universes contained in $U_t$.
        \item The sizes of the instances are powers of 2, and all instances of the same size are disjoint.
        \item At time $t$, for every $k$, there exists an active instance with size $2^k$ if and only if $\mathrm{bin}(\abs{U_t})_k=1$.
    \end{enumerate}
      
    Elements are partitioned among active instances as follows. Let us order $U_t$ so that the elements in $U_1$ come before those in $U_2 \setminus U_1$, which come before those in $U_3\setminus U_2$, etc. For every bit $k$ such that $\mathrm{bin}(|U_t|)_k = 1$, we have an active instance of size $2^k$ with  universe containing the $\left(2^{k+1} \cdot \lfloor \,\abs{U_t}/2^{k+1}\rfloor+1\right)$-th to the $\left(2^{k+1} \cdot \lfloor \,\abs{U_t}/2^{k+1}\rfloor+2^k\right)$-th elements in the order. For example, if $U_t$ is $10=(1010)_2$, then the first 8 element will be in an active instance, and the last 2 will be in another.


    At time $t$, we modify the instances as follows. Let $k^*$ be the largest bit $k$ such that $\mathrm{bin}(|U_{t}|)_k \neq \mathrm{bin}(|U_{t-1}|)_k$. It must be that $\mathrm{bin}(|U_{t-1}|)_{k^*} = 0$ and $\mathrm{bin}(|U_{t}|)_{k^*} = 1$ as $|U_t| \ge |U_{t-1}|$. We do the following:
    \begin{itemize}
        \item For every bit $k < k^*, \mathrm{bin}(U_{t-1})_k=1$, we turn its corresponding instance inactive.
        \item For every bit $k\leq k^*, \mathrm{bin}(U_{t})_k=1$, we initiate a new active instance with size $2^k$, holding the elements as defined above.
    \end{itemize}
    Properties 1.~and 3.~clearly hold by how the above process is defined. To see property 2.~holds, observe that, for any element at any time, either the active instance it belongs to remains active, or it becomes inactive and the element is included in an active and strictly larger instance. In particular, all active instances of size $2^k$ for $k > k^*$ remain active, and we create a new active instance of size $2^{k^*}$. The latter new active instance must include all elements in all instances that were active at time $t-1$ and had sizes $2^k$ for $k < k^*$, because the total size of these instances is at most $2^{k^*}-1$.

    Below is a simple example where a single element is being inserted at each time, where blue boxes are active instances, and grey boxes are inactive instances.
       \begin{figure}[H]
    \centering
    \includegraphics[width=13cm]{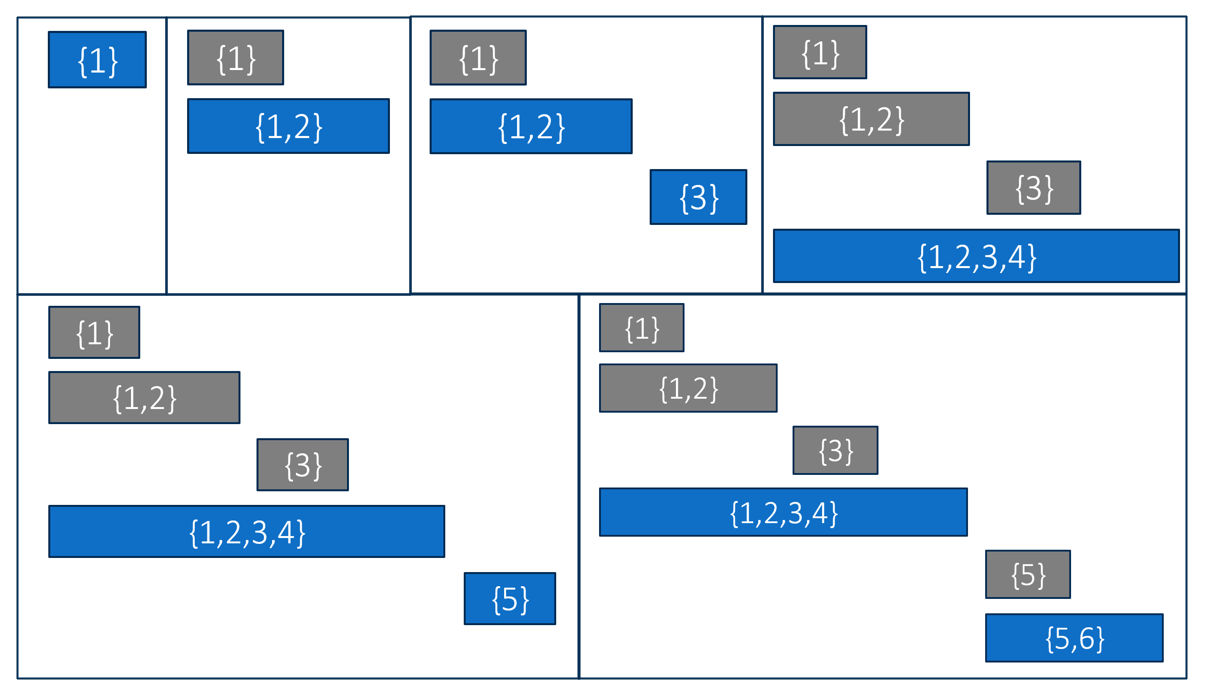}
    \caption{A simple example of insertion}
\end{figure}

    Then, for each active instance $\mathcal{I}$ with universe $V_{\mathcal{I}}$, on arrival of the row $q_t$ and ${U_t}$, we send the row $q_t|_{V_{\mathcal{I}}}$ as input to $\mathcal{I}$, and we combine the outputs as follows:
    \begin{figure}[H]
    \begin{framed}
    \begin{algorithmic}
    \INDSTATE[0]Initialize $\ell'_t$, $R'_t$ to be empty matrices.
    \INDSTATE[0]\textbf{for each} instance $\mathcal{I}$
    \INDSTATE[1] Obtain $l_{\mathcal{I},t}$ and $R_{\mathcal{I},t}$
    \INDSTATE[1] $\ell'_t\leftarrow\begin{pmatrix}
\ell'_t & \mathbf{1}\{\mathcal{I}\text{ is active}\}\ell_{\mathcal{I},t}
\end{pmatrix}$
    \INDSTATE[1]$R_t'\leftarrow\begin{pmatrix}
 {R_t'}\\
 R_{\mathcal{I},t}\\
 
\end{pmatrix}$
    
    \end{algorithmic}
    \end{framed}
    \vspace{-5mm}\caption{Combining Instances}
    \end{figure}

    Here we assume that in all instances $\mathcal{I}$, the right matrix $R_{\mathcal{I}, t}$ of the factorization has $|U|$ columns, and those that are not inside $V_{\mathcal{I}}$ are set to be all zeros. By how $\ell'_t$ is constructed, inactive instances do not contribute to $\ell'_t R'$, thus we have
    \begin{align*}
        \ell'_tR'_t
        =\sum_{\text{every active instance }\mathcal{I}}\ell_{\mathcal{I},t}R_{\mathcal{I},t}
        =\sum_{\text{every active instance }\mathcal{I}}q_{t}|_{V_\mathcal{I}}
        =q_t|_{U_t}.
    \end{align*}
    Here we used property 1 of the instances (they are disjoint and their union equals $U_t$).
    The above is true for any $t$, which implies $L'_tR'_t=Q_t|_{\mathcal{U}_t}$.
    
    Notice that $R'_t$ is constructed by concatenating all individual $R_{\mathcal{I},t}$'s of every instance $\mathcal{I}$, no matter if $\mathcal{I}$ is active or not. Since each instance outputs an online factorization, our algorithm can only add new rows to $R'_{t}$ as $t$ increases. Moreover, since the universe of each instance is contained in $U_t$, we have that all rows of $R'_t$ are supported on $U_t$. Aside from the bounds on $\|L'_t\|_{2\to \infty}$ and $\|R'_t\|_F$, this concludes the proof of correctness of our construction.
    
    We now proceed to prove the bounds on the matrix norms. It is straightforward to check that, because at any point in time we have at most $\log_2(2|U_t|) \le \log_2(2N)$ active instances, each with $\|L_{\mathcal{I},t}\|_{2\to\infty} \le C\gamma$, we have
    \(
    \|L'_t\|_{2\to \infty} \le C\gamma\sqrt{\log_2(2N)}.
    \)
     For the bound on $\|R'_t\|_F$, observe that we may bound the number of instances of size $2^k$, counting both active and inactive ones, by $2^{-k}\abs{U_t}$, since they are all disjoint by property 2. This gives
     \begin{align*}
        \norm{R'_t}^2_F&=\sum_{\text{instance }\mathcal{I}}\norm{R_{\mathcal{I},t}}^2_F
        \leq \sum_{k=0}^{\floor{\log_2{\abs{U_t}}}}2^{-k}\abs{U_t}*2^k
        \leq \abs{U_t}\log_2({2\abs{U_t}}) \le |U_t|\log_2(2N).
    \end{align*}
    To finish the proof, we let our online factorization algorithm output the matrices $R_t = \frac{1}{\sqrt{\log_2(2N)}}R'_t$ and $L_t = \sqrt{\log_2(2N)}L'_t$.
\end{proof}

We now introduce our second blackbox transformation that allows one time deletion of columns of $Q_t$, where deletion means we zero out elements in the deleted column from all new rows of $Q_t$. We use a similar approach as the previous proof. First, we give one more definition.
\begin{definition}
    For a set $U$, we say that a sequence of subsets $\mathcal{U}=(U_0,\ldots,U_m)$ where $U_0, \ldots, U_m\subseteq U$ is semi-dynamic if for any time $i<j<k$, $$U_i\cap U_k\subseteq U_j.$$
    In other words, the set of times $t$ such that $x \in U_t$ forms a contiguous interval for each $x \in U$.
\end{definition}
This captures the setup where any element can be added to the universe at anytime, but it cannot return once being deleted. 

\begin{lemma}[Handling semi-dynamic inputs] \label{lm:revocable}
Let $Q$ be an $m\times N$ matrix with columns indexed by $U$. 
Suppose there is an online factorization algorithm that solves the $(\gamma,C)$-bounded online factorization with insertions problem for $Q$ and an insertion only sequence $\mathcal{U} = (U_0, U_1, \ldots)$, and at time $t$ computes an online factorization $L_t R_t = Q_t|_{\mathcal{U}}$. Then there exists another factorization algorithm that at time $t$, in addition to $q_t$, receives as input $U'_t\subseteq U_t$, where $\mathcal{U}'_t = (U'_0,U'_1, \ldots)$ is semi-dynamic such that $U_t = \bigcup_{s \le t } U'_s$, and either
\begin{itemize}
    \item asserts that $\gamma_2(Q_t) > \gamma$, or
    \item computes a factorization $L'_t R'_t = Q_t|_{\mathcal{U}'_t}$ such that 
    \begin{align*}       \norm{L'_t}_{2\rightarrow\infty}\leq\norm{L_t}_{2\rightarrow\infty}{\log{4N}};\ \ \ \ \norm{R'_t}_{F}\leq \norm{R_t}_{F}.
    \end{align*}
\end{itemize}
Furthermore, let $t^{\text{del}}(x)$ to be the smallest value of $t$ such that $x \in U'_{t-1}$ and $x \not \in U'_{t}$, or $\infty$ if there is no such $t$. Then for any $x$, 
    $s_{R_t'}(x)\leq s_{R_{t^{\text{del}}(x)}}(x).$
\end{lemma}
\begin{proof}
    The algorithm maintains two matrices $\bar{R}$ and $R^*$ online, both starting empty. Then the matrix $R'$ is a function of $\bar{R}$ and $R^*$, and, in fact, contains the union of the rows of the two matrices, rescaled to satisfy the bound on $\|R'\|_F$ and on the individual $\ell_2$ norms of columns. We run the online algorithm for online average factorization with insertions on the insertion-only sequence $\mathcal{U}$. At time $t$, after the algorithm has processed $q_t$ and $U_t$, for each row $r$ it adds to $R_{t-1}$ to form $R_t$, we append to $\bar{R}_{t-1}$ a row $\bar{r}=r|_{U'_t}$ to form $\bar{R}_t$. If the algorithm, instead of adding rows to $R_{t-1}$, asserts that $\gamma_2(Q_t)>\gamma$, then we make the same assertion and terminate. 

    Next we describe what rows we add to $R^*_t$. 
    Fix any row $\bar r$ of $\bar{R}_t$, and let $t_0$ be the time when $\bar r$ was first added to $\bar{R}_t$. 
     Define $t_1,\ldots, t_a$ to be the times $t_0 < t'\le t$ when $U'_{t'-1} \setminus U'_{t'} \neq \emptyset$, i.e., when some elements were deleted. Let $\Delta_1,\ldots, \Delta_a$ be the sets $\Delta_i = \{x \in \uni: t^{\text{del}}(x) = t_i\}$ of elements deleted at time $t_i$. 
    Define sets $S_1,\ldots,S_a$ by 
    \[S_i:=\bigcup_{j=i+1-f(i)}^i \Delta_j,\]
    where $f(i)$ is the largest power of 2 that divides $i$ (thus the function value is at least 1). I.e., $S_i$ consists of all elements deleted at the last $f(i)$ time steps when a deletion occurred, up to time step $t_i$. We need the following two key properties of this construction.
    \begin{enumerate}
        \item We have 
        \begin{equation}\label{eq:delete-decomp}
        U'_{t_0} \setminus U'_t =  \bigcup_{i = 1}^a \Delta_a =  \bigcup_{k: \text{ bin}(a)_k=1} S_{2^k \cdot \lfloor a/2^k \rfloor},
        \end{equation}
        and, moreover, the sets $S_{2^k \cdot \lfloor a/2^k \rfloor}$, where $k$ ranges over indices s.t.~$\text{bin}(a)_k=1$, are disjoint.  The first equality follows immediately from the definition of the $\Delta_i$ because $U'_{t_0} \setminus U'_t$ consists of elements in $U'_{t_0}$ that were deleted by time $t$. The second equality is easy to see from the binary expansion of $a$.

        \item Any element $x$ of $U'_{t_0}$ is contained in at most $\log_2(2N)$ of the sets $S_i$. Suppose that $t^\text{del}(x) = t_j$, and note that $x \in S_i$ if and only if $i = 2^k \cdot \lceil j/2^k \rceil$ for some integer $k \ge 0$. There are at most $1 + \log_2 a$ such choices of $i$, so $x$ is contained in at most $\log_2(2a) \le \log_2(2N)$ of the $S_i$, where we use the observation that $a \le N$ since at each time step $t_i$ some element of $U$ gets deleted, and no element is deleted twice.  
    \end{enumerate}
    We now define $R^*_t$ to contain $-\bar{r}|_{S_i}$ for all $i \in [a]$. Note that this operation is performed for every row $\bar r$ in $\bar{R}_t$ independently, i.e., each row has its own set of times $t_i$ and corresponding $S_i$ sets. It is easy to check that $R^*_t$ thus defined can be maintained online.

    The union of the rows of $\frac{1}{\sqrt{\log_2(4N)}}\bar{R}_t$ and $\frac{1}{\sqrt{\log_2(4N)}}R^*_t$ forms the matrix $R'_t$. To see that one can compute the rows of $Q_t|_{\mathcal{U}_t}$ as linear combinations of the rows of $R'$, let ${r}$ be some row in ${R}_t$, first added at time $t_0$, and take the corresponding row $\bar{r} = r|_{U'_{t_0}}$ of $\bar{R}_t$; defining $t_1, \ldots, t_a$ and $S_1, \ldots, S_a$ as above, by \eqref{eq:delete-decomp} we have  
    \[{r}|_{U'_{t}}= \bar{r}-\sum_{\text{bin}(a)_k=1}\bar{r}|_{S_{2^k \cdot \lfloor a/2^k\rfloor}}.\]
    Here we used the fact that $U'_t \setminus U'_{t_0}$ contains elements $x$ of the universe that were inserted after time $t_0$, i.e., $x \in U_t \setminus U_{t_0}$, and, by the definition of online factorization with insertions, $r$ is supported only on $U_{t_0}$.
    Now, to construct $\ell'_t$ on receiving $q_t$, we take the row $\ell_t$ in the factorization $L_t R_t = Q|_{\mathcal{U}}$, and replace each coordinate of it that corresponds to some row $r$ of $R_t$, with a coordinate (of the same value) corresponding to the row $\bar{r}$ in $\bar{R}_t$, and coordinates (also of the same value) corresponding to  $-\bar{r}|_{S_{2^k \cdot \lfloor a/2^k\rfloor}}$ for each $k$ such that $\mathrm{bin}(a)_k = 1$. Then we scale the resulting row by $\sqrt{\log_2(4N)}$ to get $\ell'_t$.
    
    Since each coordinate of $\ell_t$ gets replaced by at most $1+\log_2(2N)$ copies of it, scaled up by $\sqrt{\log_2(4N)}$,  to form $\ell'_t$, we get the bound on $\norm{L'_t}_{2\to\infty}$. For each row $\bar{r}$ of $\bar{R}$, each coordinate in $\bar{r}$ is replicated at most $1 + \log_2(2N)$ times in $R'_t$, possibly negated, and scaled down by $\sqrt{\log_2(4N)}$. This gives the bound on $\norm{R'_t}_F$.

    By how our algorithm is defined, it is clear that 
    once $t \ge t^{\text{del}}(x)$, any rows added into $\bar{R}_t$, can only have 0 in the column indexed by $x$. Since $R^*_t$ only contains copies of rows in $\bar{R}_t$ with some entries zeroed out, this proves the final claim.
\end{proof}
The transformation in Lemma~\ref{lm:revocable} guarantees that once a column is deleted (its index $x$ stops appearing in $U'_t$), the norm of the corresponding column of $R'_t$ is bounded from above for all remaining time steps.

With the two transformations above we can handle any semi-dynamic sequence. We are now ready to implement an online version of the reduction from factorization to average factorization.


\begin{lemma}[Kicking elements out]\label{lm:simple_reduction}
Let $Q$ be an $m\times N$ matrix with columns indexed by $U$. Suppose that for any fixed $U'$ there is an online factorization algorithm that solves the $(\gamma,C)$-bounded online average factorization problem for $Q|_{U'}$. Then there exists another online factorization algorithm that solves the $(\gamma,C')$-bounded online factorization problem with $C' \le C\sqrt{2}\log(4N)^3$.
\end{lemma}
\begin{proof}
To describe our algorithm, we define several \textbf{zones}, numbered from $0$ to $K = \lfloor\log_2(N)\rfloor$. A zone with number $\zeta$ corresponds to a semi-dynamic factorization problem with sequence of subsets $\mathcal{U}_\zeta = (U_{\zeta,0}, U_{\zeta,1}, \ldots)$ of $U$, decided by our algorithm. We ensure that, at each time step $t$, the sets $U_{\zeta,t}$ over different $\zeta$ are disjoint, and that their union is $U$. At the beginning of the algorithm, we have $U_{0,0}=U$, and $U_{\zeta,0} = \emptyset$ for all $\zeta > 0$. In each zone, we apply the transformation in Lemma~\ref{lm:insertable} to the given algorithm to obtain an insertion-only factorization algorithm with a sequence $\mathcal{U}^{\text{ins}}_\zeta = (U_{\zeta,0}^{\text{ins}} = U_{\zeta,0}, U^{\text{ins}}_{\zeta,1},\ldots)$. Then we apply Lemma \ref{lm:revocable} to the insertion only algorithm to obtain a semi-dynamic algorithm. \textbf{Both} algorithms will be run for each zone, with the factorization computed by the insertion-only algorithm deciding the set $U_{\zeta,t}$ for the semi-dynamic algorithm. If the insertion-only algorithm asserts that $\gamma(Q_t)>\gamma$, then we terminate and make the same assertion. Note that in the proof of Lemma~\ref{lm:revocable} we only assert $\gamma(Q_t)>\gamma$ only if the insertion-only algorithm does that.

A key process in this algorithm is ``kicking out'' elements from one zone to the next. If  $x \in U_{\zeta,t}$ is such that $s_{R^{\text{ins}}_{\zeta,t}}(x)^2 > 2$, where $R^{\text{ins}}_{\zeta,t}$ is the matrix  computed by the insertion-only algorithm, then we remove $x$ from $U_{\zeta,t}$ and insert it into $U_{\zeta+1,t}$ (and $U^{\text{ins}}_{\zeta+1,t}$). The precise algorithm is described below:
\begin{figure}[H]
    \begin{framed}
    \begin{algorithmic}
    \INDSTATE[0]Receive input $q_t$.
    \INDSTATE[0]Initialize $S\leftarrow \emptyset$, a buffer for kicking out elements.
    \INDSTATE[0]Initialize $\ell'_{t}, R'_{t}$ to be empty.
    \INDSTATE[0]\textbf{For each} zone $\zeta=0,1, \ldots, K$  \textbf{do}:
    \INDSTATE[1]$U^{ins}_{\zeta, t} \leftarrow  U^{ins}_{\zeta, t-1}\cup S$
    \INDSTATE[1]$U_{\zeta, t} \leftarrow  U_{\zeta, t-1}\cup S$
    \INDSTATE[1]$S\leftarrow \emptyset$
    
    \INDSTATE[1]
    \INDSTATE[1]Obtain $R^{ins}_{\zeta, t}$ from the algorithm in Lemma~\ref{lm:insertable} on input $q_t$ and $U^{ins}_{\zeta, t}$.
    \INDSTATE[1]If the algorithm asserts that $\gamma_2(Q_t) > \gamma$, we make this assertion too, and terminate.
    \INDSTATE[1]\textbf{For} each $x \in U_{\zeta,t}$ s.t. $s_{R^{\text{ins}}_{\zeta,t}}(x)^2 > 2$:
    \INDSTATE[2]$U_{\zeta, t} \leftarrow U_{\zeta,t}\setminus \{x\}$ 
    \INDSTATE[2]$S\leftarrow S \cup \{x\}$.
    \INDSTATE[1]
    \INDSTATE[1]Run the algorithm in Lemma~\ref{lm:revocable} on $q_t$ and $U_{\zeta, t}$ to obtain $\ell_{\zeta,t}R_{\zeta, t}=q_t|_{U_{\zeta, t}}$
    \INDSTATE[1] $\ell'_t\leftarrow\begin{pmatrix}
\ell'_t & \ell_{\zeta,t}
\end{pmatrix}$
    \INDSTATE[1]$R_t'\leftarrow\begin{pmatrix}
R_t'\\ R_{\zeta,t}
\end{pmatrix}$
    \INDSTATE[0]\textbf{End}
    \end{algorithmic}
    \end{framed}
    \vspace{-5mm}\caption{Algorithm within each zone}
    \label{fig:kick-out}
    \end{figure}

 By Lemma \ref{lm:insertable}, for any zone $\zeta$, we have
\[
\sum_{x \in U^{\text{ins}}_{\zeta,t}}s_{R^{ins}_{\zeta, t}}(x)^2 = \norm{R^{ins}_{\zeta, t}}^2_F\leq |U^{\text{ins}}_{\zeta,t}|.
\]
Therefore, by Markov's inequality, less than half of the elements of $U^{\text{ins}}_{\zeta,t}$ will be ``kicked out'', i.e., inserted into $U^{\text{ins}}_{\zeta+1,t}$. Thus, at any time step $t$, we have $|U^{\text{ins}}_{\zeta,t}| < 2^{-\zeta} N$, and, for the final zone $\zeta = K = \lfloor\log_2(N)\rfloor$ there will be no kicked out elements.

It is clear that the sets $U_{\zeta,t}$ over different $\zeta$ are disjoint, because an element is inserted into $U_{\zeta+1,t}$ only after being deleted from $U_{\zeta,t}$. By the argument above, we also have that $\bigcup_{\zeta = 0}^KU_{\zeta,t} = U$. Without loss of generality, we can again assume that all $R^{\text{ins}}_{\zeta, t}$ and $R_{\zeta, t}$ matrices have $N$ columns indexed by $U$ by padding with all-zeroes columns. Since the factorization $L'_t R'_t$ we compute equals $\sum_{\zeta = 0}^K Q_t|_{U_{\zeta,t}}$, this implies that $L'_t R'_t = Q_t$, i.e., the factorization is correct.


Since, for any zone $\zeta$, we delete any $x$ at the first time $t$ for which $s_{R^{\text{ins}}_{\zeta,t}}(x)^2 > 2$, we have that
 \begin{align*}
     \norm{R_{\zeta,t}}^2_{1\rightarrow2}&= \max_x{s^2_{R_{\zeta,t}}(x)}
     \leq \max_x{s^2_{R^{ins}_{\zeta,t^{del}(x)}}(x)}
     \leq 2,
 \end{align*}
where the first inequality is by Lemma \ref{lm:revocable}.

Now we are ready to finish the proof. For any $t$, we have
\begin{align*}
    \norm{R'_t}_{1\rightarrow2}^2&\leq \sum_{\zeta=0}^K\norm{R'_{\zeta,t}}^2_{1\rightarrow2}\leq 2\log_2{(2N)}\\
    \norm{L'_t}_{2\rightarrow\infty}^2
    &\leq \sum_{\zeta=0}^K\norm{L'_{\zeta,t}}^2_{2\rightarrow\infty}
   \leq C^2 \gamma^2 \log_2{(4\abs{U})}^5,
\end{align*}
where we used that $1+K \le \log_2(2N)$, and, in the second line, we used Lemmas~\ref{lm:insertable}~and~\ref{lm:revocable} to get that $\norm{L'_{\zeta,t}}_{2\rightarrow\infty} \le C\gamma \log(4N)^2$. Scaling $R'_t$ by $\frac{1}{\sqrt{2\log(2N)}}$ and $L'_t$ by $\sqrt{2\log(2N)}$ finishes the proof.
\end{proof}

Lemma~\ref{lm:fact-main} is now implied by Lemmas~\ref{lm:online_F_transpose_factorization}~and~\ref{lm:simple_reduction} with $C = \sqrt{6}\log(4N)^3$. As mentioned above, Lemmas~\ref{lm:fact-doubling}~and~\ref{lm:fact-main} imply Theorem~\ref{thm:fact-main}. Note that for each new row of $Q$ that arrives, the reduction requires amortized $\polylog(N)$ time on top of invoking the online average factorization algorithm's update procedure amortized $\polylog(N)$ times. Thus the overall algorithm requires amortized $\polylog(N)$ time to process each query assuming the average factorization algorithm is implemented using binary search.

\section{Lower Bounds}\label{section:lowerbounds}

In this section we prove two related lower bounds for two of the online problems we have studied. First we show a lower bound for online factorization algorithms in the row arrival model. The lower bound implies that the factor $\log_2(\gamma_2(Q)/\|q_1\|_\infty)^{\frac12 + c}$ in Theorem~\ref{thm:fact-main} cannot be improved significantly. Whether the other logarithmic factors  can be improved is an open question. Next we use the same lower bound construction to show a similar lower bound for online query release. Note that both lower bounds are for oblivious (i.e., non-adaptive) adversaries.

The main idea behind the lower bound for online factorization is to exploit the fact that a competitive online factorization needs to achieve value competitive against $\gamma_2(Q_t)$ for all intermediate time steps $t$. In particular, if $C$ is the competitiveness ratio, then we must have that $\|\ell_t\|_2 \le C \gamma_2(Q_t)$, which can be a much stronger requirement than $\|\ell_t\|_2 \le C \gamma_2(Q)$. We exploit this to show that, for an appropriate choice of $Q$, this forces $C = \Omega(\sqrt{\log(\gamma_2(Q)/\|q_1\|_\infty)})$.

\begin{theorem}\label{thm:fact-lb}
    For all positive integers $N$, there exists an $N\times N$ matrix $Q$ with $\|q_1\|_\infty = 1$, and $\log_2(\gamma_2(Q)) = \Theta(N)$, such that, for any online factorization algorithm for $Q$ in the row arrival model, there exists a time step $t$ for which the factorization value achieved by the algorithm by step $t$ is at least $\sqrt{\frac12 \log_2(\gamma_2(Q))}\cdot\gamma_2(Q_t)$.
\end{theorem}
\begin{proof}
    We will prove the theorem for any $N$ that is a (non-negative) power of $2$, and will argue that the online factorization algorithm achieves value at least $\sqrt{\log_2(\gamma_2(Q))}\gamma_2(Q_t)$. The statement for all $N$ follows by padding $Q$ with all-zeros rows and columns until the next power of $2$.

    For $N$ a power of $2$, let $H$ be a Hadamard matrix, and let $h_t$ be the $t$-th row of $H$. Define $Q$ so that $q_t = 2^{(t-1)/2} h_t$. We have $\gamma_2(Q_t) \le \|Q_t\|_{1\to 2} = \sqrt{2^t-1} \le 2^{t/2}$. In the other direction, $\gamma_2(Q_t) \ge \|q_t\|_\infty = 2^{(t-1)/2}$. Together, these inequalities give $\gamma_2(Q_t) = \Theta(2^{t/2})$.
    
    Suppose that there is an online factorization algorithm that achieves value $C\gamma_2(Q_t) \le C 2^{t/2}$ by time step $t$ for every $t \in [N]$. Take $R = R_N$ to be the final right matrix computed by the algorithm, and define a matrix $\tilde{L}$ with rows $\tilde{\ell}_t = 2^{-(t-1)/2}\ell_t$, suitably padded with zeros so that the multiplication $\tilde{\ell}_t R$ is well-defined. By assumption, $\|\tilde{\ell}_t\|_2 \le 2^{-(t-1)/2}\cdot C2^{t/2} = \sqrt{2}C$, so $\norm{\tilde{L}}_{2\to\infty} \le \sqrt{2}C$. Notice also that $\tilde{\ell}_t R = h_t$, so the factorization $\tilde{L}R = H$ certifies that $\gamma_2(H) \le \sqrt{2}{C}$. On the other hand, it is easy to see that $\gamma_2(H) \ge \sqrt{N}$. For example, using \eqref{eq:sdp-duality} with $y_1 = \ldots = y_N = \frac{1}{N}$, we have
    \[
    \gamma_2(H) \ge \frac{1}{\sqrt{N}}\gamma_F(H^T)
    \ge \frac{1}{N} \tr\left(\sqrt{H^T H}\right)
    = \frac1N \tr(\sqrt{N} I)= \sqrt{N}.
    \]
    Therefore, $C \ge \sqrt{\frac{N}{2}} \ge \sqrt{\log_2\gamma_2(Q)}$.
\end{proof}

Essentially the same construction gives a similar lower bound for interactive private algorithms answering online queries.
\begin{theorem}\label{thm:priv-lb}
    For all positive integers $N$, there exists a sequence of statistical queries $\mathcal{Q} = (q_1, \ldots, q_N)$ such that the following holds. Suppose that there exists an $(\varepsilon,\delta)$-differentially private interactive algorithm that receives $\mathcal Q$ online, and produces answers $a_1, \ldots, a_N$, so that, for every dataset $x$ of size $n$, with probability at least $\frac23$, 
    \[
    |q_t(x) - a_t| \le \frac{C \gamma_2(Q_t)}{\varepsilon n}.
    \]
    Then, for all large enough $n$ it must hold that $C=\Omega\left({\sqrt{\log_2(\gamma_2(Q))}}\right)$. Moreover, the query matrix $Q$ of $\mathcal{Q}$ satisfies $\|q_t\|_\infty =1$, and $\log_2(\gamma_2(Q)) = \Theta({N})$.
\end{theorem}
\begin{proof}
    We let $\mathcal Q$ be the queries with the query matrix $Q$ defined in the proof of Theorem~\ref{thm:fact-lb}. As in Theorem~\ref{thm:fact-lb}, we may assume $N$ is a power of $2$. By assumption, we have that
    \[
    |h_t(x) - 2^{-(t-1)/2}a_t| \le \frac{\sqrt{2}C}{\varepsilon n},
    \]
    where $\mathcal{H} = (h_1, \ldots, h_N)$ are the statistical queries with query matrix the $N\times N$ Hadamard matrix $H$. Note that, by the postprocessing property of differential privacy, the answers $a_1, \ldots, 2^{-(N-1)/2}a_N$ to $\mathcal{H}$ are $(\varepsilon,\delta)$-differentially private. However, for all large enough $n$, no $(\varepsilon,\delta)$-differentially private algorithm can be $(\alpha,\frac13)$-accurate for $\mathcal H$ unless $\alpha = \Omega\left(\frac{\sqrt{N}}{\varepsilon n}\right)$~\cite{DworkY08}, which shows that $C = \Omega(\sqrt{N}) = \Omega(\sqrt{\log(\gamma_2(Q))}).$
\end{proof}

This lower bound implies that the factor $\log_2(\gamma_2(Q)/\|q_t\|_\infty)^{\frac12 + c}$ in Theorem~\ref{thm:main-priv} also cannot be improved significantly.
\section{Competitive Online Query Release for Small Datasets}\label{sect:jon}

In this section we introduce our online competitive algorithm for small datasets, which is described in Section~\ref{sec:intro-smalln}.  First we present the main discrepancy theoretic tools we need, then the main tools from differential privacy, and then the algorithm and its analysis. 

\subsection{Hereditary Discrepancy Tools}

First we recall the definition and basic properties of \emph{hereditary discrepancy}.  Let $\mat{Q}$ be a $m \times N$ matrix.  Given a subset of columns $S \subseteq [N]$, we use $\QS$ to denote the $m \times |S|$ submatrix of $\mat{Q}$ consisting of all columns of $\mat{Q}$ in the set $S$.

\begin{defn}[Hereditary Discrepancy] \label{def:hdisc}
Given a $m \times N$ matrix $\mat{Q}$ and a parameter $1 \leq w \leq N$, we define the \emph{restricted hereditary discrepancy}
$$
\hdisc(\mat{Q}, w) := \max_{S \subseteq [N] : |S| \leq w} \left(\min_{\vec{x} \in \{\pm 1\}^{|S|}} \| \QS \vec{x} \|_{\infty}\right).
$$ 
\end{defn}

Hereditary discrepancy is closely related to the rounding of a certain linear program.  Specifically, given any real-valued vector $\vec{y} \in [0,1]^N$, we can round $\vec{y}$ to a Boolean vector $\vec{x} \in \zo^N$ such that the coordinates of $Q(\vec{y}-\vec{x})$ is bounded by the hereditary discrepancy of $Q$.  Specifically, we have the following classical result.

\begin{thm} [\cite{LovaszSV86}] \label{thm:ldisc2hdisc}
For every $m \times N$ matrix $\mat{Q}$ and parameter $1 \leq w \leq N$,
$$
\max_{\vec{y} \in [0,1]^N : \|\vec{y}\|_0 \leq w} \left(\min_{\vec{x} \in \{0,1\}^N : \| \vec{x} \|_0 \leq w} \| \mat{Q}\vec{y} - \mat{Q}\vec{x} \|_\infty\right) \leq \hdisc(\mat{Q},w).
$$
\end{thm}

We also define the following \emph{modified hereditary discrepancy} that we denote $\hdisc^*(\mat{Q},w)$: For an $m \times N$ matrix $\mat{Q}$, we let $\mat{Q}^{*}$ denote the $(m+1) \times N$ matrix formed by appending the vector $(1,\dots,1) \in \R^N$ as the $(m+1)$-st row of $\mat{Q}$.  Observe that for any vector $\vec{x} \in [0,1]^N$ we have
\begin{equation*}
\mat{Q}^{*} \vec{x} =
\begin{bmatrix}
\mat{Q} \vec{x}~ \\
~~\| \vec{x} \|_1~
\end{bmatrix}
\end{equation*}

\begin{defn}[Modified Hereditary Discrepancy]  Given a $m \times N$ matrix $\mat{Q}$ and a parameter $1 \leq w \leq U$, we define the \emph{modified restricted hereditary discrepancy}
$$
\hdisc^*(Q,w) := \hdisc(Q^{*},w)
$$
\end{defn}

The next lemma will be useful for simplifying the statement of our bounds, and is immediate from the definition of (modified) hereditary discrepancy.
\begin{fact} \label{lem:sublineardisc} For every $Q \in [0,1]^{m \times N}$ and every $c,w \in \N$, $\hdisc^*(Q,cw) \leq c \cdot \hdisc^*(Q,w)$.
\end{fact}

\begin{rem}
Since duplicating rows does not increase the hereditary discrepancy, for any matrix $\mat{Q}$ that contains the row $(1,\dots,1)$, $\hdisc^*(Q,w) = \hdisc(Q,w)$.  Further, for every $\mat{Q} \in [0,1]^{m \times N}$ and every $w \in \N$, $\hdisc^*(Q,w) \geq 1$.
\end{rem}

The connection between hereditary discrepancy and differential privacy was first made in the work of Muthukrishnan and Nikolov~\cite{MuthukrishnanN12}.  Here we state their result (rescaled for relative error and reparameterized) that we will use to establish the competitive property of our algorithm.
\begin{thm}[Modification of \cite{MuthukrishnanN12} for Normalized Error] \label{lem:hdisc2opt}
Let $Q$ be the  $m \times N$ query matrix of a set of $m$ statistical queries over universe $U$ of size $N$, and let $n \in \N$ be a dataset size.  For every $(\frac18,\frac18)$-differentially private algorithm $M_Q \from U^n \to [0,1]^m$, there exists a dataset $x \in U^n$ such that
$$
\ex{}{\| Q(x) - M_Q(x) \|_{\infty} } \geq \alpha = \Omega\left( \frac{\hdisc^*(Q,n)}{n \log n} \right).
$$
Moreover, there is no $(\frac18,\frac18)$-differentially private algorithm $M_Q \from U^n \to [0,1]^m$ that is $(\alpha, \frac18)$-accurate for $Q$.
\end{thm}

For completeness, we will sketch a proof of this result.
\begin{proof}[Proof Sketch]
    The result of Muthukrishnan and Nikolov is stated in terms of \emph{unnormalized error} and \emph{restricted hereditary partial discrepancy}.  Specifically, they show that for every set of queries $\mathcal{Q}$ with query matrix $Q$, and every $n$, and every differentially private mechanism $M$, there exists a dataset $x \in U^t$ of size $t \leq n$, with histogram $h_x$, such that
    $$
        \ex{}{\| Qh_x - M(x)\|_\infty} \geq C \max_{S \subseteq [N] \atop |S| \leq n} \min_{z \in \{-1,0,1\}^|S| \atop \|z\|_1 \geq |S|/10} \|Qz\|_\infty.
    $$
    for some constant $C > 0$.  By a standard argument relating hereditary partial discrepancy to hereditary discrepancy, the right hand side is at least equal to the hereditary discrepancy up to an $O(\log n)$ factor.  Thus
    $$
        \ex{}{\| Qh_x - M(x)\|_\infty} \geq \frac{C \hdisc(Q,n)}{\log n}
    $$
    for some constant $C$.  Rescaling by $t$ we get that the normalized error satisfies
    $$
        \ex{}{\left\| \frac{1}{t} Q h_x - \frac{1}{t} M(x)\right\|_\infty} \geq \Omega\left(\frac{\hdisc(Q,n)}{t \log n} \right).
    $$
    Thus, for every mechanism, there is some dataset size $t$ where the normalized error is at least $\frac{C \cdot \hdisc(Q,n)}{t \log n}$.  If there is a mechanism for datasets of size $n$ with normalized error smaller than $\alpha = \frac{C \cdot \hdisc(Q,n)}{n \log n}$ then there is a mechanism for datasets of size $t \leq n$ with normalized error smaller than $n\alpha / t = \frac{C \cdot \hdisc(Q,n)}{t \log n}$.  To obtain this mechanism, we pad the dataset with $n-t$ copies of any fixed element of the domain.  Then upon receiving the answers with error smaller than $\alpha$, we can subtract off the answers of the queries of the padding elements and rescale by $(n/t)$ to obtain answers that are accurate with error smaller than $n\alpha / t$, which gives a contradiction.
\end{proof}

\subsubsection{Sparse Approximations}

The key property of hereditary discrepancy that we will use is the following: If $\mathcal{Q}$ is a set of Boolean statistical queries whose corresponding matrix has low hereditary discrepancy, and $x$ is a dataset of size $n$, then we can find a smaller dataset $y$, whose size is related to the hereditary discrepancy of $Q$ that is a good approximation of $x$ with respect to the queries $Q$.  

As a starting point, for \emph{any} set of queries $\mathcal{Q}$, we can approximate $x$ by a small dataset via subsampling the elements of the dataset.  Specifically, using a standard application of Hoeffding's inequality and a union bound, we have the following lemma.
\begin{lem} \label{lem:sampling}
For every $\alpha \in (0,1), m \in \N$, every set of $m$ statistical queries $\mathcal{Q}$ over $U$, and every dataset $x \in U^*$, there exists a dataset $y \in U^s$ for $s = \frac{\ln m}{2\alpha^2}$ such that
$
\| \mathcal{Q}(x) - \mathcal{Q}(y) \|_{\infty} \leq \alpha.
$
\end{lem}

The key lemma we will use is a refined version of the above statement where $s \approx \frac{\log m}{\alpha} \cdot \hdisc(Q,\frac{1}{\alpha^2})$ and $Q$ is the query matrix of $\mathcal{Q}$.  Note that, if we focus on the dependence on $\alpha$, then for a worst-case set of queries where $\hdisc(Q,n) \approx \sqrt{n}$, the bound we get on the size of the small dataset is $s \approx \frac{1}{\alpha^2}$, which is similar to what we get from the subsampling argument.  But for queries with lower hereditary discrepancy the bound can be much smaller.  

Our lemma will follow from the following more general statement.

\begin{lem} \label{lem:sparsification}
Let $\alpha \in (0,1)$ be a parameter.  Let $\mat{Q} \in [0,1]^{m \times N}$ be a matrix and let ${x} \in [0,1]^N$ be a vector with $\|x\|_0 \leq w$ such that $\hdisc^*(Q,w) \leq \frac{\alpha w}{4}$.  Then there exists a Boolean vector $\vec{y} \in \zo^N$ such that
$$
\|y\|_1 := s \leq \frac{5 \cdot \hdisc^*\left(Q, w\right)}{\alpha}
~~~\textrm{and}~~~
\left\| \mat{Q}\frac{\vec{x}}{w} - \mat{Q}\frac{\vec{y}}{s} \right\|_{\infty} \leq \alpha.
$$
\end{lem}
\begin{proof}[Proof of Lemma \ref{lem:sparsification}]
As a shorthand, let $H := \hdisc^*\left(Q, w\right)$.  Let $2H \leq t \leq w$ be a parameter.  We consider the vector $\frac{t}{w} \vec{x} \in \{0, \frac{t}{w}\}^U \subseteq [0,1]^U$.  Observe that $\left\| \frac{t}{w}x \right\|_1 = t$.
By Theorem \ref{thm:ldisc2hdisc}, and the definition of $\hdisc^*$, there exists a Boolean vector $y \in \zo^{N}$ such that
\begin{equation}
\left\| Q \frac{t\vec{x}}{w} - Q \vec{y} \right\|_{\infty} \leq H \text{ and }\label{eq:sparseapprox1} \left| t - \| y \|_1 \right| \leq H 
\end{equation}
As a shorthand, let $s := \|\vec{y}\|_1$ so that \eqref{eq:sparseapprox1} states $| t - s | \leq H$.  We can now calculate
\begin{align*}
\left\| \mat{Q}\frac{\vec{x}}{w} - \mat{Q}\frac{\vec{y}}{s} \right\|_{\infty} 
\leq{} &\left\| \mat{Q}\frac{\vec{x}}{w} - \mat{Q}\frac{\vec{x}}{w}\frac{t}{s} \right\|_{\infty} 
+ \left\| \mat{Q}\frac{\vec{x}}{w}\frac{t}{s} -  \mat{Q}\frac{\vec{y}}{s}\right\|_{\infty} \\
\leq{} &\left|1-\frac{t}{s}\right|
+ \left\| \mat{Q}\frac{\vec{x}}{w}\frac{t}{s} -  \mat{Q}\frac{\vec{y}}{s}\right\|_{\infty} \tag{see below} \\
\leq{} &\left|1-\frac{t}{s}\right| + \frac{H}{s} \tag{by \ref{eq:sparseapprox1}} \\
\leq{} &\left|1-\frac{t}{t - H}\right| + \frac{H}{t - H} \tag{by \ref{eq:sparseapprox1}} \\
\leq{} &\left|1-\frac{t + 2H}{t}\right| + \frac{2H}{t} \tag{$t \geq 2H$} \\
\leq{} &\frac{2 H}{t} + \frac{2H}{t} \leq \frac{4H}{t}.
\end{align*}
The highlighted line comes from the fact that for every row $q_i$ in $Q$, we have 
$$
\bigg| \Big\langle q_i, \frac{\vec{x}}{w} - \frac{\vec{x}t}{ws} \Big\rangle \bigg| \leq \bigg|1- \frac{t}{s} \bigg| \cdot \bigg|\Big\langle q_i, \frac{\vec{x}}{w} \Big\rangle \bigg| \leq \bigg|1-\frac{t}{s} \bigg|,
$$
where the last inequality holds because $q_i \in [0,1]^{N}$ and $\|\vec{x}\|_1 = w$.

Now we can complete the proof by setting $t:= \frac{4H}{\alpha}$.  Observe that this choice of $t$ satisfies the constraint $t \geq 2H$, and by our assumption that $\hdisc^*(Q,w) \leq \frac{\alpha w}{4}$ we will also satisfy our constraint that $t \leq w$.  Finally, the bound on $s = \|y\|_1$ follows from the fact that $| t - s | \leq H$.
\end{proof}

\begin{rem}
Since $\hdisc^*(Q,w) \leq \sqrt{2 w \ln m}$ for every matrix $Q$, the condition $\hdisc^*(Q,w) \leq \frac{\alpha w}{4}$ in the above lemma will hold whenever $w \geq \frac{32 \ln m}{\alpha^2}$.
\end{rem}

We are now ready to state the first theorem about sparse approximations of datasets.
\begin{thm}[Sparse Approximations] \label{thm:sparseapproxn}
For every $\alpha \in (0,1)$ and every set of $m$ statistical queries $\cQ$ over universe $U$ of size $N$, with query matrix $Q \in [0,1]^{m \times N}$, and every $n \in \N$ there exists
$$
s = \frac{5}{\alpha} \cdot \hdisc^*(Q,n) = O\left(\frac{1}{\alpha} \cdot \hdisc^*(Q,n)\right)
$$
such that for every dataset $x \in U^n$ whose elements are all distinct, there exists a dataset $y \in U^*$ of size at most $s$ such that
$$\| \cQ(x) - \cQ(y) \|_\infty \leq \alpha.$$
\end{thm}

\begin{proof}[Proof of Theorem \ref{thm:sparseapproxn}]
Fix a dataset $x \in U^n$ with unique rows, $\alpha \in (0,1)$, and a set of $m$ statistical queries $\cQ$ over $U$.  We will show how to construct the desired approximation $y \in U^s$.

Let $\vec{h} \in \zo^U$ be the histogram of the dataset $x$ and observe that $\| \vec{h} \|_1 = n$.  By Lemma~\ref{lem:sparsification}, there exists a vector $\vec{h}' \in \zo^{U}$ such that
$$
\| \vec{h'} \|_1 = s \leq \frac{5}{\alpha} \cdot \hdisc^*(Q,n),
$$
and
\begin{equation*}
\left\| \mat{Q}\frac{\vec{h}}{n} - \mat{Q}\frac{\vec{h'}}{s} \right\|_{\infty} \leq \alpha.
\end{equation*}
Since $\vec{h'} \in \zo^{U}$ and $\| \vec{y} \|_1 = s$, $\vec{h'}$ is the histogram representation of some dataset $y \in U^s$.  Moreover, since $\cQ(x) = Q\frac{\vec{h}}{n}$ and $\cQ(y) = Q\frac{\vec{h'}}{s}$, we have
$
\| \cQ(x) - \cQ(y) \|_\infty \leq \alpha.
$
as desired.
\end{proof}

\begin{rem}
The restriction that $x = (x_1,\dots,x_n) \in U^n$ have unique rows can be removed by appending the number $i$ as a tag to each row $x_i$.  Doing so increases the universe size from $N$ to $Nn$.
\end{rem}

We also prove a corollary of the above theorem in which the size of the approximation does not depend on $n$, which is sometimes easier to interpret, though not strictly necessary for our later results.

\begin{cor}[Dataset-Size-Independent Sparse Approximations] \label{thm:sparseapprox}
For every $\alpha \in (0,1)$ and every set of $m$ statistical queries $\cQ$ over universe $U$ of size $N$, with query matrix $Q \in [0,1]^{m \times N}$, there exists
$$
s = \frac{5}{\alpha} \cdot \hdisc^*\left(Q,\frac{32 \ln m}{\alpha^2}\right) = O\left(\frac{\log m}{\alpha}\cdot \hdisc^*\left(Q, \frac{1}{\alpha^2}\right)\right)
$$
such that for every dataset $x \in [U]^*$ whose rows are all distinct, there exists a dataset $y \in [U]^*$ of size at most $s$ such that
$$\| \cQ(x) - \cQ(y) \|_\infty \leq \frac98 \alpha.$$
\end{cor}


\begin{proof}[Proof of Theorem \ref{thm:sparseapprox}]
Fix a dataset $x \in U^*$ with unique rows, $\alpha \in (0,1)$, and a set of $m$ statistical queries $\cQ$ over $U$.  We will show how to construct the desired approximation $y \in U^s$.

First, by Lemma~\ref{lem:sampling}, if $w = 32 \alpha^{-2} \ln m$ then there exists a dataset $z \in U^s$ such that
\begin{equation} \label{eq:sparseapprox_1} 
\| \cQ(x) - \cQ(z) \|_\infty \leq \frac18 \alpha.
\end{equation}

Now let $\vec{h} \in \zo^N$ be the histogram of the dataset $z$ and observe that $\| \vec{h} \|_1 = w$.  By Lemma~\ref{lem:sparsification}, there exists a vector $\vec{h'} \in \zo^{N}$ such that
$$
\| \vec{h'} \|_1 = s := \frac{5}{\alpha} \cdot \hdisc^*\left( Q, \frac{32 \ln m}{\alpha^2} \right),
$$
and
\begin{equation} \label{eq:sparseapprox_2}
\left\| \mat{Q}\frac{\vec{h}}{w} - \mat{Q}\frac{\vec{h'}}{s} \right\|_{\infty} \leq \alpha.
\end{equation}
Since $\vec{h'} \in \zo^{N}$ and $\| \vec{h'} \|_1 = s$, $\vec{h'}$ is the histogram representation of some dataset $y \in U^s$.  Moreover, since $\cQ(z) = Q\frac{\vec{h}}{w}$ and $\cQ(y) = Q\frac{\vec{h'}}{s}$, by applying the triangle inequality to \eqref{eq:sparseapprox_1} and \eqref{eq:sparseapprox_2} we have
$$
\| \cQ(x) - \cQ(y) \|_\infty \leq \frac98 \alpha,
$$
as desired.  This completes the proof.
\end{proof}

\subsection{Technical Tools for Differential Privacy}

Before presenting a competitive differentially private algorithm, we will introduce some important properties of differential privacy and some useful technical tools.

\paragraph{Laplace Noise and AboveThreshold.} The two differentially private components in our algorithm are Laplacian noise and the AboveThreshold algorithm.

We use $\Lap(\sigma)$ to denote the \emph{(centered) Laplace distribution}, whose density function is $$\mu_{\Lap(\sigma)}(z) \propto e^{-|z|/\sigma}.$$  The Laplace distribution is useful for ensuring differential privacy, however this aspect will be fully encapsulated within the analysis of the AboveThreshold algorithm.  However, to ensure accuracy of our algorithms, we will need the following bound on the tails of the Laplace distribution.
\begin{lem}[Concentration of Laplace Distribution] \label{lem:laplacetails}
For every $\sigma, t \geq 0$, $\pr{}{|\Lap(\sigma)| > t\sigma} \leq e^{-t}.$
\end{lem}

Privacy of our algorithms will follow directly from the privacy of the following algorithm, which is called \emph{AboveThreshold}.  The algorithm was first formalized by Dwork and Roth~\cite{DworkR14}, although it is an abstraction of the earlier \emph{sparse vector technique}~\cite{DworkNRRV09,DworkNPR10,RothR10,HardtR10}.
\begin{figure}[h!]
\begin{framed}
\begin{algorithmic}
\INDSTATE[0]{{\bf Input:} dataset $x \in U^n$, adversary $\mathcal{B}$, threshold $\tau \in R$, privacy parameter $\eps \in (0,1)$}
\INDSTATE[0]{}
\INDSTATE[0]{Let $\sigma := 4/\eps n$, let $\hat\tau = \tau + \Lap(\sigma)$}
\INDSTATE[0]{{\bf Repeat:}}
\INDSTATE[1]{$\mathcal{B}$ chooses a $(1/n)$-sensitive function $f \from U^n \to \R$}
\INDSTATE[1]{If $f(x) + \Lap(\sigma) \leq \hat\tau$: {\bf output} $\bot$ and continue the loop}
\INDSTATE[1]{Otherwise, {\bf output} $f(x) + \Lap(\sigma)$ and halt}
\end{algorithmic}
\end{framed}
\vspace{-5mm}\caption{The AboveThreshold algorithm $\mathit{AT}_{\mathcal{B}, \tau, \eps}(x)$}
\end{figure}

\begin{thm}[Privacy of $\mathit{AT}$; see e.g.~\cite{DworkR14}] \label{thm:AboveThreshold}
For every algorithm $\mathcal{B}$ for the adversary, and every $\tau \in \R$, $\eps > 0$, the algorithm $\mathit{AT}_{\mathcal{B},\tau,\eps}(x)$ is $(\eps,0)$-differentially private.
\end{thm}

\paragraph{Composition of Differential Privacy.} Finally, we review the composition properties of differential privacy that allow us to build our algorithm modularly out of these components.

\begin{figure}[h!]
\begin{framed}
\begin{algorithmic}
\INDSTATE[0]{{\bf Input:} dataset $x \in U^n$, adversary $B$,  $T \in \N$.}
\INDSTATE[0]{{\bf For} $t = 1,\dots,T$:}
\INDSTATE[1]{$B$ chooses an $(\eps_0, \delta_0)$-differentially private algorithm $M_t \from U^n \to \cR$.}
\INDSTATE[1]{$v_t = M_t(x)$}
\INDSTATE[0]{{\bf Output} $(v_1,\dots,v_T)$}
\end{algorithmic}
\end{framed}
\vspace{-5mm}\caption{Adaptive composition game $\mathit{Comp}_{B,T,\eps_0,\delta_0}(x)$}
\end{figure}

\begin{thm}[Composition of Differential Privacy \cite{DworkMNS06, DworkRV10}] \label{thm:dpcomp}
For every adversary $B$, $T \in \N$, $0 \leq \eps_0, \delta_0 \leq 1$, the algorithm $\mathit{Comp}_{B,T,\eps_0,\delta_0}$ is:
\begin{enumerate} 
\item $(\eps, \delta)$-differentially private for $\eps = T \eps_0$ and $\delta = T \delta_0$, and 
\item for every $\delta_1 > 0$, $(\eps, \delta)$-differentially private for $\eps = \sqrt{6 T \ln(\frac{1}{\delta_1})}\,\eps_0$ and $\delta = T\delta_0 + \delta_1$.
\end{enumerate}
\end{thm}

\subsection{An Online Competitive Algorithm and its Analysis}
The algorithm consists of an inner algorithm based on the median mechanism of Roth and Roughgarden~\cite{RothR10} as well as an outer wrapper based on doubling search to find the right value of $s$ for the sparse approximation.  The outer algorithm is presented in Figure \ref{alg:M} and the inner algorithm is presented in \ref{alg:Minner}.
\begin{figure}[h!]
\begin{framed}
\begin{algorithmic}
\INDSTATE[0]{{\bf Input:} A dataset $x \in U^n$ over a universe of size $|U| = N$, the privacy parameters $\eps, \delta \in (0,1)$, the maximum number of rounds $m \in \N$, and a failure parameter $\beta > 0$.}
\INDSTATE[0]{}
\INDSTATE[0]{{\bf Set Parameters:} $\eps_0 := \dfrac{\eps}{\log_2 n}~~~\delta_0 := \dfrac{\delta}{\log_2 n}~~~\beta_0 := \dfrac{\beta}{\log_2 n}$}
\INDSTATE[0]{}
\INDSTATE[0]{Let $s = 1$}
\INDSTATE[0]{{\bf Repeat} until $s > n$ or $m$ queries have been exhausted:}
\INDSTATE[1]{Run $M^{\mathrm{inner}}_{s, \eps_0, \delta_0,m,\beta_0}(x)$ and {\bf output} its transcript}
\INDSTATE[1]{Let $s = 2s$}
\end{algorithmic}
\end{framed}
\vspace{-5mm}\caption{A competitive algorithm $M_{\eps,\delta,m,\beta}(x)$} \label{alg:M}
\end{figure}

\begin{figure}[h!]
\begin{framed}
\begin{algorithmic}
\INDSTATE[0]{{\bf Input:} A dataset $x \in U^n$ over a universe of size $|U| = N$, a size parameter $s \in \N$, privacy parameters $\eps, \delta \in (0,1)$, the maximum number of rounds $m \in \N$, and a failure parameter $\beta > 0$.}
\INDSTATE[0]{}
\INDSTATE[0]{{\bf Set Parameters:} $$c := \log_2\left( \textstyle\sum_{t=1}^{s} N^t \right)~~~~\eps_0 := \frac{\eps}{\sqrt{6 c \ln(1/\delta)}}
~~~~\sigma := \frac{4}{\eps_0 n}~~~~\tau := 4 \sigma \ln(\tfrac{3m}{\beta})$$}
\INDSTATE[0]{Let $\cY = U^{\leq s}$ be the set of all datasets of size at most $s$.}
\INDSTATE[0]{{\bf Repeat} until $\cY = \emptyset$ or $m$ queries have been exhausted: (outer loop)}
\INDSTATE[1]{Let $\hat{\tau} = \tau + \Lap(\sigma)$}
\INDSTATE[1]{{\bf Repeat:} (inner loop)}
\INDSTATE[2]{Get the next query $q$}
\INDSTATE[2]{Let $\nu$ be the median of $\{q(y)\}_{y \in \cY}$}
\INDSTATE[2]{If $|q(x) - \nu| + \Lap(\sigma) \leq \hat\tau$: {\bf output} $a = \nu$ and continue the inner loop}
\INDSTATE[2]{Otherwise: {\bf output} $a = q(x) + \Lap(\sigma)$ and break the inner loop}
\INDSTATE[1]{{\bf Update:}}
\INDSTATE[2]{If $a > \nu$, let $\cY = \set{y \in \cY : q(y) > \nu}$}
\INDSTATE[2]{Otherwise, let $\cY = \set{y \in \cY : q(y) < \nu}$}
\end{algorithmic}
\end{framed}
\vspace{-5mm}\caption{The subroutine  $M^{\mathrm{inner}}_{s,\eps,\delta,m,\beta}(x)$} \label{alg:Minner}
\end{figure}

\subsubsection{Analysis of the Inner Algorithm}

\begin{lem}[Privacy of $M^{\mathrm{inner}}$] \label{lem:Minnerdp}
For every $s, m \in \N$, $0 \leq \eps, \delta \leq 1$, and $\beta > 0$, the algorithm $M^{\mathrm{inner}}_{s, \eps, \delta, m,\beta}$ is $(\eps, \delta)$-differentially private.
\end{lem}
\begin{proof}[Proof of Lemma \ref{lem:Minnerdp}]
The proof proceeds in two steps.  First, every iteration of the inner loop is an instance of the AboveThreshold algorithm (Theorem \ref{thm:AboveThreshold}) with threshold $\tau$ and parameter $\eps_0$, and therefore is $(\eps_0, 0)$-differentially private.  To verify that $M^{\mathrm{inner}}$ can be expressed as an instance of AboveThreshold, observe that in each round the function being tested against the threshold is $f(x) := |q(x) - \nu|$ where $q$ is a statistical query, which is $1/n$-sensitive.  Also observe that when $f(x) + \Lap(\sigma) \leq \hat\tau$, the algorithm outputs $a = \nu$, whereas AboveThreshold would output $\bot$.  However, conditioned on the outcome of previous iterations of the outer loop, $\nu$ is already determined and does not depend on the dataset $x$, so releasing $\nu$ instead of $\bot$ does not compromise privacy.  That is, it would not change the distribution of the above algorithm if we modified it to answer $a = \bot$ in place of $a = \nu$, and relied on the party choosing $q$ to determine $\nu$ himself.

Second, because $|\cY|$ is initially $\sum_{t=1}^{s} N^t$, and every iteration of the outer loop decreases the size of $\cY$ by a factor of at least $2$, there can be at most $c$ iterations of the outer loop.  Thus, the entire algorithm is an adaptive composition of at most $c$ instances of the $(\eps_0,0)$-differentially private AboveThreshold algorithm.  By Theorem \ref{thm:dpcomp}, and our choice of $\eps_0 = \eps / \sqrt{6 c \ln(1/\delta)}$, the algorithm satisfies $(\eps, \delta)$-differential privacy.  This completes the proof.
\end{proof}

\begin{lem}[Accuracy of $M^{\mathrm{inner}}$] \label{lem:Minneracc}
Let $\cQ = \set{q_1,q_2,\dots}$ be the sequence of queries received (which contains at most $m$ queries) in an execution of $M^{\mathrm{inner}}_{s,\eps,\delta,m,\beta}(x)$ and let $a = (a_1,a_2,\dots)$ be the corresponding answers, with probability at least $1-\beta$,
$$
\| \cQ(x) - a \|_{\infty} \leq \frac{3\tau}{2} = O\left( \frac{\sqrt{s \log N \log(\frac{1}{\delta})}\log(\frac{m}{\beta})}{\eps n} \right),
$$
and, if there exists $y \in U^s$ such that $\|\cQ(x) - \cQ(y) \|_{\infty} \leq \tau/2$, then the algorithm does not halt before exhausting the queries.
\end{lem}
\begin{proof}[Proof of Lemma \ref{lem:Minneracc}]
We will show that the conclusion of the lemma holds conditioned on the event that every sample from $\Lap(\sigma)$ drawn in the execution of the algorithm is bounded by $\tau/4$ in absolute value.  Observe that there are at most $m$ queries issued in the execution of the algorithm and thus there are at most $3m$ samples drawn from $\Lap(\sigma)$.  By our choice of $\sigma, \tau$, tail bounds for the Laplace distribution (Lemma \ref{lem:laplacetails}), and a union bound, the probability that any of these draws exceeds $\tau/4$ in absolute value is at most $\beta$.

Now, suppose that the algorithm does not halt before answering $m$ queries.  Then every query $q(x)$ is answered with either $a = \nu$ or with $a = q(x) + \Lap(\sigma)$.  In the latter case, we have that $| q(x) - a | \leq \tau/4$.  In the former case, we have that $| q(x) - \nu | + \Lap(\sigma) \leq \hat\tau$, where $\hat\tau = \tau + \Lap(\sigma)$.  This condition implies that $|q(x) - \nu| \leq 3\tau/2$.

Thus, the only way that the algorithm can fail to answer each query to within $3\tau/2$ is if it halts early because $\cY = \emptyset$.  However we claim that if there is an element $y \in U^s$ such that $\|\cQ(x) - \cQ(y) \|_\infty \leq \tau/2$ then it can never be eliminated from $\cY$ during any update.  To see that this is the case, consider some update step and assume without loss of generality that $a > \nu$.  In order to perform such an update, we must have $q(x) > \nu + \tau/2$, which implies that $q(y) > \nu$.  Therefore $y$ will survive the update and remain in $\cY$.  An analogous argument applies to update steps in which $a \leq \nu$.

By the previous two arguments, if none of the samples from $\Lap(\sigma)$ is too large in magnitude, then if there is a suitable $y \in U^s$ then the algorithm will not stop early, and if the algorithm does not stop early it will answer every query accurately.  This completes the proof.
\end{proof}

\subsubsection{Competitive Analysis}
\begin{lem}[Privacy of $M$] \label{lem:Mdp}
For every $0 \leq \eps, \delta \leq 1$, $m \in \N$, and $\beta > 0$, the algorithm $M_{\eps,\delta,m,\beta}$ is $(\eps,\delta)$-differentially private.
\end{lem}
\begin{proof} [Proof of Lemma \ref{lem:Mdp}]
By Lemma \ref{lem:Minnerdp}, each iteration satisfies $(\eps_0, \delta_0)$-differential privacy.  The main loop of $M$ runs for at most $\log_2(n)$ iterations.  No other step of $M$ depends on the dataset $x$.  Thus the lemma follows by the composition property of differential privacy (Theorem \ref{thm:dpcomp}).
\end{proof}

\begin{lem}[Accuracy of $M$] \label{lem:Macc}
Let $\cQ = \set{q_1,\dots,q_m}$ be the sequence of queries received in an execution of $M_{\eps,\delta,m,\beta}(x)$ and let $a = (a_1,\dots,a_m)$ be the corresponding answers, then with probability at least $1-\beta$,
$$
\| \cQ(x) - a \|_{\infty} \leq \tilde{O}\left(\frac{\sqrt{ \hdisc^*(\cQ,n) \log N \log(\frac{1}{\delta})} \log(\frac{m}{\beta})}{\eps n}\right)^{2/3}
$$
\end{lem}
\begin{proof} [Proof of Lemma \ref{lem:Macc}]
By Lemma \ref{lem:Minneracc}, and a union bound over the $\log_2 n$ iterations of the main loop of $M$, we have that with probability at least $1-\beta$, we will have the following two conditions:
\begin{enumerate}
\item[(1)] In each invocation of $M^{\mathrm{inner}}_{s,\eps_0,\delta_0,m,\beta_0}(x)$ with sparsity parameter $s$, if the queries issued are $\cQ^{s} = \set{q^s_1,q^s_2,\dots}$ and the answers given are $a^s = (a^s_1,a^s_2,\dots)$, then
\begin{align*}
\| \cQ^s(x) - a^s \|_{\infty} 
\leq \alpha^{(s)} ={} &O\left( \frac{\sqrt{s \log N \log(\frac{1}{\delta_0})} \log(\frac{m}{\beta_0})}{\eps_0 n} \right) \\
={} &O\left( \frac{\sqrt{s \log N \log(\frac{\log n}{\delta})} \log(\frac{ m \log n}{\beta}) \log n}{\eps n} \right)
\end{align*}
\item[(2)] If $\cQ = \cQ^1 \cup \cQ^2 \cup \dots$ is the entire sequence of queries issued in the execution of $M_{\eps,\delta,m,\beta}(x)$, then if $s^*$ is such that there exists a dataset $y \in U^{s^*}$ such that
$$
\| \cQ(y) - \cQ(x) \|_{\infty} \leq \alpha^{(s^*)}/3 \Longrightarrow \|\cQ^{s^*}(y) - \cQ^{s^*}(x) \|_{\infty} \leq \alpha^{(s^*)}/3,
$$ 
the queries will be exhausted in the execution of $M^{\mathrm{inner}}$ with sparsity parameter $s^*$.  By part (1) above, this will imply
$$
\| \cQ(x) - a \|_{\infty} \leq \alpha^{(s^*)} = O\left( \frac{\sqrt{s^* \log N \log(\frac{\log n}{\delta})} \log(\frac{m \log n}{\beta}) \log n}{\eps n} \right).
$$
\end{enumerate}
Assume that these conditions both hold.  To complete the proof we determine a suitable parameter of $s^*$ in which the algorithm stops.  By Theorem \ref{thm:sparseapproxn}, if we let $H := \hdisc^*(Q,n)$, then for every $\alpha$, and some $s = O(\frac{H}{\alpha})$, there exists a dataset $y \in U^s$ such that $\| \cQ(y) - \cQ(x) \|_{\infty} \leq \alpha$.  Rearranging, for every $s$, for some $\gamma^{(s)} = O(\frac{H}{s})$, there exists $y \in U^s$ such that $\| \cQ(y) - \cQ(x) \|_{\infty} \leq \alpha^{(s)}$.  Therefore, by part (2) above, the algorithm will stop in some round $s^*$ such that $\gamma^{(s^*)} \leq \alpha^{(s^*)}/3$.  Solving this equation gives an upper bound on $s^*$ of
$$
s^* = \tilde{O}\left( \frac{H^2 \eps^2 n^2}{\log N \log(\frac{1}{\delta}) \log^2(\frac{m \log n }{\beta}) \log n} \right)^{1/3}
$$
which gives an upper bound on $\alpha^{s^*}$ of
$$
\alpha^{(s^*)} = \tilde{O}\left(\frac{\sqrt{ H \log N \log(\frac{1}{\delta})} \log(\frac{m \log n}{\beta}) \log n}{\eps n}\right)^{2/3}
$$
Observe that, since $H = \hdisc^*(\cQ,n) \geq 1$, whenever $s^*$ would be less than $1$, the bound on $\alpha^{(s^*)}$ will be at least $1$, which is trivial to achieve for statistical queries.  Thus, to minimize complexity of the argument, we will ignore these pathological parameter regimes.  This completes the proof.
\end{proof}

Combining Lemmas \ref{lem:Mdp} and \ref{lem:Macc} gives the following theorem.

\begin{thm}
There exists an $(\eps,\delta)$-differentially private interactive algorithm that takes a dataset $x \in U^n$ from a universe of size $N$ such that if a data analyst asks a set of adaptively chosen Boolean statistical queries $\mathcal{Q} = (q_1,\dots,q_m)$, and the algorithm answers each query with $a = (a_1,\dots,a_m)$, then with probability at least $1 - \beta$,
$$
\| Q(x) - a \|_{\infty} \leq \alpha = \tilde{O}\left(\frac{\sqrt{ \hdisc^*(Q,n) \log N \log(\frac{1}{\delta})} \log(\frac{m \log n}{\beta}) \log n}{\eps n}\right)^{2/3}.
$$
\end{thm}

To interpret the above theorem, note that for every differentially private algorithm $M_{Q}$ there exists a dataset $x \in U^n$ such that 
$$
\ex{}{\| Q(x) - M_{Q}(x) \|_{\infty}} = \tilde{\Omega}\left(\frac{\hdisc^*(Q,n)}{n} \right) = \alpha \cdot \tilde{\Omega}\left(\frac{\hdisc^*(Q,n)^2}{n \log N \log^2 m}\right)^{1/3}.
$$


\addcontentsline{toc}{section}{References}
\bibliography{biblio}
\bibliographystyle{alpha}

\appendix
\section{Boolean Matrices with Bounded Shatter Function Exponent}
\label{sec:vcdim}


While Theorem~\ref{thm:fact-main} is competitive against the $\gamma_2$ norm, we can also show that if we have extra knowledge about the matrix, we may achieve even stronger guarantees. In this section we study the scenario where the matrix is Boolean (i.e., has entries in $\{0,1\}$), and the complexity of its rows is bounded. More precisely, we give an online factorization mechanism whose value is bounded in terms of the shatter function exponent of the matrix, and its dimensions. Recall that the shatter function exponent of a Boolean matrix $Q$ is the smallest $d$ so that, for any $s$, any submatrix of $Q$ with $s$ columns has at most $O(s^d)$ distinct rows. Our main result is given in Theorem~\ref{thm:fact-vcdim}, restated below.

\factvcdim*


Rather than use Lemma~\ref{lm:online_F_transpose_factorization}, our starting point for proving Theorem~\ref{thm:fact-vcdim} is an explicit online average factorization of matrices with bounded shatter function exponent. The factorization is an online version of an offline factorization by Muthukrishnan and Nikolov~\cite{MuthukrishnanN12}, which is, in turn, based on a discrepancy upper bound by Matou\v{s}ek~\cite{Mat-halfspaces}. Before we describe the online average factorization algorithm, we state the classical packing bound of Haussler, which is key to the analysis of the value of the factorization.

\begin{lemma}[\cite{Haussler}]\label{lm:packing}
If an $N$ column Boolean matrix $M$ has shatter function exponent $d$, and any pair of rows of $M$ have $\ell_2$ distance at least $\sqrt{\Delta}$, then the number of rows of $M$ is at most $O\left(\left(\abs{U}/\Delta\right)^d\right)$.
\end{lemma}

The next lemma is our basic construction of an online average factorization for matrices with bounded shatter function exponent. 

\begin{lemma}\label{lm:net_queries}
Let $Q$ be an $m\times N$ binary matrix with columns indexed by $U$ and with shatter function exponent $d$. There is an online average factorization algorithm that at time $t$ outputs the factorization $L_tR_t=Q_t$ such that
$$\norm{L_t}^2_{2\rightarrow\infty}=O(m^{1-1/d}) \indent \norm{R_t}^2_{F}\leq N.$$
\end{lemma}
\begin{proof}
    We construct a ``net'' data structure that stores row vectors from $Q$, as they arrive over time. We then use this data structure to construct our matrices $L_t$ and $R_t$.
    
    The data structure is partitioned into ``layers'', $\mathcal{L}_0,\ldots,\mathcal{L}_h$, each holding a subset of the vectors, where any two vectors in $\mathcal{L}_i$ are separated by at least $\sqrt{2^{-i}N}$ in $\ell_2$ distance. This means the top layer $\mathcal{L}_0={\varnothing}$ only holds one vector. Without loss of generality we assume it is the all-zeros vector (which we can always insert into $Q$ as the first vector). For the rest of the proof, we set $h=\ceil{\log{N}}$, which means the bottom layer $\mathcal{L}_h$ can hold all $2^{N}$ possible vectors without violating the distance assumption. For every vector $q\in \mathcal{L}_i$, $i > 0$, there is a parent vector $\rho_i(q) \in \mathcal{L}_{i-1}$ such that $\|q- \rho_i(q)\|_2\leq \sqrt{2^{-i+1}N}$. Thus, the data structure can be thought of as a tree, with the $0$ vector as its root, and the vectors $q_t$ as the leaves.

    Initially, the data structure only contains the all-zeros vector in $\mathcal{L}_0$. When a new row $q_t$ of $Q$ arrives, we do the following to maintain the data structure.
    \begin{figure}[H]
    \begin{framed}
    \begin{algorithmic}
    \INDSTATE[0]Receive input $q_t$.
    \INDSTATE[0]\textbf{For} $i$ from $h$ to 0 \textbf{do}:
    \INDSTATE[1]\textbf{If} there exist $q^*\in \mathcal{L}_i$ s.t. $\|q^*- q_t\|_2\leq \sqrt{2^{-i}N}$:
    \INDSTATE[2]$\rho_{i+1}({q_t})\leftarrow q^*$ //Choose the closest if there are multiple
    \INDSTATE[2]\textbf{Break}
    \INDSTATE[1]\textbf{Otherwise}:
    \INDSTATE[2]Insert a copy of $q_t$ into $\mathcal{L}_i$
    \INDSTATE[2]$\rho_{i+1}({q_t})\leftarrow q_t$
    \end{algorithmic}
    \end{framed}
    \vspace{-5mm}\caption{Inserting into data structure}
    \label{fig:net-pseudocode}
    \end{figure}
Clearly, the distance requirement we defined above will be satisfied. We use the notation $\rho(q_t) = \rho_{h+1}(q_t)$ for the vector in $\mathcal{L}_h$ that is closest to $q_t$. For this lemma, $\rho(q_t)$ will simply be a copy of $q_t$, since we set $h=\ceil{\log{N}}$ which means in the very first iteration the condition inside the ``if'' cannot be satisfied. We will consider a more general case later in this section.



Notice that during the insertion process, identical copies of the new vector can be inserted into several layers until, at a certain point, a close enough vector is found that triggers the ``if'' statement in the insertion algorithm. In the example above, children and parents that are the same vector are connect by blue edges, and otherwise they are connected by red edges. We name these red edges {\color{red} difference edges}. I.e., a difference edge is a pair $(q,\rho_i(q))$, where $q \in \mathcal{L}_i$ and $\rho_i(q) \neq q$.

After the data structure is updated, we update $R_t$ by processing the newly added difference edges as follows.
    \begin{figure}[H]
    \begin{framed}
    \begin{algorithmic}
    \INDSTATE[0]\textbf{For} $i$ from $h$ to 1 \textbf{do}:
    \INDSTATE[1]\textbf{For each} newly added {\color{red} difference edge} $(q,\rho_i(q))$ \textbf{do}:
    \INDSTATE[2]Add to $R_t$ a new row $\eps_i^{-1}(q-\rho_i(q))$
    \end{algorithmic}
    \end{framed}
    \vspace{-5mm}\caption{Updating $R_t$} 
    \end{figure}
That is, every difference edge will correspond to a row in $R_t$. Here $\eps_i$ is a parameter for every layer that we will set later. Note that at the beginning of the algorithm $R_0$ is an empty matrix.

The last step is to output a row $\ell_t$ so that $\ell_tR=q_t$.
\begin{figure}[H]
    \begin{framed}
    \begin{algorithmic}
    \INDSTATE[0]$\ell_t\leftarrow {0}$
    \INDSTATE[0]$currentNode\leftarrow \rho(q_t)$
    \INDSTATE[0]\textbf{For} $i$ from $h$ to 1 \textbf{do}:
    \INDSTATE[1]\textbf{If} $currentNode\neq \rho_i(currentNode)$:
    \INDSTATE[2]// This is a difference edge, there must be a row corresponding to it in $R_t$
    \INDSTATE[2]Set $j$ s.t.~the $j$-th row of $R_t$ is $\eps_i^{-1}(currentNode-\rho_i(currentNode))$
    \INDSTATE[2]$\ell_j\leftarrow \eps_i$
    \INDSTATE[2]$currentNode\leftarrow \rho_i(q_t)$
    \end{algorithmic}
    \end{framed}
    \vspace{-5mm}\caption{Computing $\ell_t$} 
    \end{figure}

To see the correctness of the factorization, observe that the sum of $q- \rho_i(q)$ over difference edges $(q, \rho_i(q))$ on the path from $q_t$ to the zero vector in $\mathcal{L}_0$ is just $q_t$ itself.

Now we prove the bounds on $\|L_t\|_{2\to\infty}$ and $\|R_t\|_F$, using Haussler's packing lemma (Lemma~\ref{lm:packing}). A direct consequence of the packing lemma is that $\abs{\mathcal{L}_i}\leq 2^{di}$. 
Let $h^*=\ceil{\frac{\log_2{m}}{d}}$, and $\eps_i=m^{\frac{d-1}{2d}}(1+\abs{h^*-i})^{-1.5}$.
By how our algorithm is defined, every new vector creates exactly one difference edge, so there at $m$ difference edges in total. The contribution of any difference edge $(q, \rho_i(q))$ to $\|R_t\|_F^2$ is bounded by  $\eps_i^{-2}\|q- \rho_i(q)\|_2^2 \le f(i)$ where $f(i)$ is defined as $\eps_i^{-2}2^{1-i}N$.  Observe that there exist some absolute constant $C$ s.t. for any $i>j$, $f(i)\leq Cf(j)$. Thus, up to constants, the total contribution of the $m$ difference edges to $\|R_t\|_F^2$ is maximized when they are all at levels as close as possible to level $0$. 
We have\footnote{Here and in the rest of the section we use the notation $A\lesssim B$ when there exists an absolute constant $C$ so that $A \le CB$.}
\begin{align*}
    \norm{R_t}^2_{F}\lesssim \sum_{i=1}^{h^*}{2^{di}\cdot 2^{1-i}N\eps_i^{-2}}
    &\lesssim Nm^{\frac{1-d}{d}}\sum_{i=1}^{h^*}{2^{(d-1)i}(1+h^*-i)^3}
    \lesssim Nm^{\frac{1-d}{d}}2^{(d-1)h^*}
    \lesssim N.
\end{align*}
Finally, to bound $\|\ell_t\|_2$, note that $\ell_t$ contains at most entry, equal to $\eps_i$, for each layer $i$, so we have
\begin{align*}
    \norm{\ell_t}_2^2\leq \sum_{i=1}^h \eps_i^2
    = m^{\frac{d-1}{d}}\sum_{i=1}^h (1+\abs{h^*-i})^{-3}
    \lesssim m^{\frac{d-1}{d}}.
\end{align*}
Note that up to here we only showed we can achieve $\norm{R_t}_{F}^2=O(N)$ rather than $\|R_t\|_F^2 \le {N}$,  so to finish the proof we merely divide $R_t$ by a big enough constant $C$, and multiply $L_t$ by the same $C$. 
\end{proof}
Using exactly the same reduction in Lemma \ref{lm:simple_reduction} gives us the following:
\begin{corollary}\label{col:upper_bound_queries}
    Let $Q$ be an $m\times N$ Boolean matrix with shatter function exponent $d$. There is an online factorization algorithm for $Q$ with value $O(m^{1-1/d}\log^6{N})$.
\end{corollary}

\paragraph{Optimized Reduction.} In the case of Boolean matrices with bounded shatter function exponent (or VC dimension), the analysis in Lemma \ref{lm:simple_reduction} can be improved when we aim for a bound on the factorization value that is polynomial in $N$ rather than $m$. On a high level, the improvement comes from adding the identity matrix to the right matrix in the online average factorization, and making sure the rows of the identity matrix are not repeated in the different zones in the reduction in Lemma~\ref{lm:simple_reduction}. 
Towards this end, we first state a slightly different version of Lemma~\ref{lm:net_queries}.
\begin{lemma}\label{lm:net_universe}
Let $Q$ be an $m\times N$ binary matrix with columns indexed by $U$ and with shatter function exponent $d$. There is an online factorization algorithm that at time $t$ outputs the right matrix
$R_t = \begin{pmatrix}
    I\\R'_t
\end{pmatrix}$
and the last row in the left matrix 
$L_t = \begin{pmatrix}
L^{(I)}_t & L'_t
\end{pmatrix}$, where 
\(L_t R_t
=
Q_t,
\)
$I$ is the $N\times N$ identity matrix, and 
\begin{align*}
    \norm{L'_t}^2_{2\rightarrow\infty}&=O(N^{1-1/d}\log^{-1+2/d}{N}), &\norm{R'_t}^2_{F}=O\left(\frac{N}{\log{4N}}\right),\\
    \norm{L^{(I)}_t}^2_{2\rightarrow\infty}&=O(N^{1-1/d}\log^{2/d}{N}).
\end{align*}
\end{lemma}
\begin{proof}
    We will slightly modify the data structure we defined in Lemma \ref{lm:net_queries}. The data structure $\mathcal{L}_0, \ldots, \mathcal{L}_h$ is constructed as before, using the algorithm in Figure~\ref{fig:net-pseudocode}, but with a different value of $h$ (see the next paragraph). We initialize $R_0 = I$, and the rest of $R_t$, which we denote $R'_t$ is updated as before using the data structure. The first $N$ columns of $L_t$, which we denote $L_t^{(I)}$, correspond to the rows of the identity matrix in $R_t$.
    
    In Lemma \ref{lm:net_queries}, we set the $h$ so that the bottom layer $\mathcal{L}_h$ can contain every possible vector. Now, we instead set $h=\floor{\frac{\log_2{N}-2\log_2{\log_2{N}}}{d}}$. Recall that $\rho(q_t)$ is the closest vector in $\mathcal{L}_h$ to $q_t$. Because of our choice of $h$, it is no longer true that $\rho(q_t) = q_t$, so the sum of $q - \rho_i(q)$ over difference edges $(q,\rho_i(q))$ on the path from $\rho(q_t)$ to the zero vector in $\mathcal{L}_0$ is $\rho(q_t)$ rather than $q_t$. 
    To address this, we modify the vector $\ell_t$. 
    \begin{figure}[H]
    \begin{framed}
    \begin{algorithmic}
    \INDSTATE[0]$\ell^{(I)}_t\leftarrow q_t-\rho(q_t)$ // The only modification
    \INDSTATE[0]$\ell'_t\leftarrow {0}$
    \INDSTATE[0]$currentNode\leftarrow \rho(q_t)$
    \INDSTATE[0]\textbf{For} $i$ from $h$ to 1 \textbf{do}:
    \INDSTATE[1]\textbf{If} $currentNode\neq \rho_i(currentNode)$:
    \INDSTATE[2]Set $j$ s.t.~the $j$-th row of $R'_t$ is $\eps_i^{-1}(currentNode-\rho_i(currentNode))$
    \INDSTATE[2]$\ell'_j\leftarrow \eps_i$
    \INDSTATE[2]$currentNode\leftarrow \rho_i(q_t)$
    \INDSTATE[0] \textbf{Output} $\ell_t \leftarrow \begin{pmatrix}
        \ell^{(I)}_t & \ell'_t
    \end{pmatrix}.$
    \end{algorithmic}
    \end{framed}
    \vspace{-5mm}\caption{Computing $\ell_t$}
    \end{figure}

    In other words, we choose $R'_t$ the same way we chose $R_t$ in Lemma~\ref{lm:net_queries}, and $L'_t$ the same way we chose $L_t$. The same reasoning as in Lemma~\ref{lm:net_queries} gives us that $\ell'_t R'_t = \rho(q_t)$, so we set $\ell_t^{(I)} = q_t - \rho(q_t)$, which guarantees that 
    \[
    \ell_t R_t = \begin{pmatrix} \ell_t^{(I)} & \ell'_t \end{pmatrix}\begin{pmatrix} I \\ R'_t \end{pmatrix}=\ell_t^{(I)} + \ell'_t R'_t = q_t - \rho(q_t) + \rho(q_t) = q_t.
    \]
    
    We now proceed to bounding $\|L^{(I)}_t\|_{2\to\infty}$ $\|L'_t\|_{2\to\infty}$, and $\|R'_t\|_F$. Recall that each difference edge $(q, \rho_i(q))$ contributes $\eps_i^{-2} \|q-\rho_i(q)\|_2^2 \le \eps_i^{-2} 2^{1-i} N$ to $\|R'_t\|_F^2$. We set $\eps_i=\frac{N^{\frac{d-1}{2d}}\log_2^{0.5}{(4N)}}{(1+h-i)^{1.5}\log_2^{1-1/d}{N}}$.
    Since, by Lemma~\ref{lm:packing}, $|\mathcal{L}_i|\le 2^{di}$, so there are at most $2^{di}$ difference edges $(q, \rho_i(q))$ with $q \in \mathcal{L}_i$, we have
    \begin{align*}
    \norm{R'_t}^2_{F}&\le \sum_{i=1}^{h}{2^{di}\cdot 2^{1-i}N\eps_i^{-2}}\\
    &\lesssim  \frac{N}{\log{4N}}\left({N}^{\frac{1-d}{d}}\log^{2-2/d}{({N})}\sum_{i=1}^{h}{2^{(d-1)i}(1+h-i)^3}\right)\\
    &\lesssim \frac{N}{\log{4N}}\left({N}^{\frac{1-d}{d}}\log^{2-2/d}{({N})}\cdot 2^{(d-1)h}\right)
    \lesssim  \frac{N}{\log{4N}}
\end{align*}

We also have
\begin{align*}
    {\norm{\ell'_t}_2^2}&\leq \sum_{i=1}^h \eps_i^2\\
    &= \log_2{(4N)}\cdot N^{\frac{d-1}{d}}\log_2^{-2+2/d}{({N})}\cdot\sum_{i=1}^h (1+h-i)^{-3}\\
    &\lesssim {N}^{\frac{d-1}{d}}\log^{-1+2/d}{{N}}
\end{align*}

It remains to bound $\norm{\ell^{(I)}_t}_2$. By the construction of the data structure, we know that either $q_t = \rho(q_t) = \rho_{h+1}(q_t)$, or  $\|q_t-\rho{(q_t)}\|_2^2 \le 2^{-h} N$. Thus we have
\begin{align*}
    \norm{\ell^{(I)}_t}_2^2\leq \|q_t - \rho{(q_t)}\|_2^2\leq 2^{-h}{N}\leq {N}^{1-1/d}\log_2^{2/d}{{N}}
\end{align*}
This finishes the proof.
\end{proof}

We are now ready to prove our main result for Boolean matrices with bounded shatter function exponent.
\begin{proof}[Proof of Theorem~\ref{thm:fact-vcdim}]
    By Corollary \ref{col:upper_bound_queries} it suffices to only prove the second bound. To do so, we follow the same outline as Lemma~\ref{lm:simple_reduction}. We start with the online average factorization algorithm from Lemma~\ref{lm:net_universe}, and then apply the two reductions that allow us to support insertions (Lemma~\ref{lm:insertable}) and one-time deletions (Lemma~\ref{lm:revocable}). We make several modifications to the transformations in order to optimize the value of the factorization. We can then apply the strategy from Lemma~\ref{lm:simple_reduction} to get an online factorization algorithm from an online average factorization algorithm that supports semi-dynamic outputs. Again we slightly optimize the construction from Lemma~\ref{lm:simple_reduction} and analyze it carefully to get the promised factorization value. 
    
    To begin, let us slightly modify the algorithm for combining instances in Lemma~\ref{lm:insertable} to the following:
    \begin{figure}[H]
    \begin{framed}
    \begin{algorithmic}
    \INDSTATE[0]Initialize $\ell'_t$, $R_t'$ to be empty matrices.
    \INDSTATE[0]\textbf{for each} instance $\mathcal{I}$:
    \INDSTATE[1] Let $k$ be such that the size of $\mathcal{I}$ is $2^k$, and $\beta_k=(1+\floor{\log_2{N}}-k)^{1.5}$.
    \INDSTATE[1] Obtain $l_{\mathcal{I},t}$ and $R_{\mathcal{I},t}$
    \INDSTATE[1] $\ell'_t\leftarrow \begin{pmatrix}
\ell'_t & \mathbf{1}(\mathcal{I}\text{ is active})\beta_{k}l_{\mathcal{I},t}
\end{pmatrix}$
    \INDSTATE[1]$R_t'\leftarrow\begin{pmatrix}
 {R_t'}\\
 R_{\mathcal{I},t}/\beta_{k}
\end{pmatrix}$
    
    \end{algorithmic}
    \end{framed}
    \vspace{-5mm}\caption{Insertion with scaling}
    \end{figure}
 The only difference with the algorithm in Lemma~\ref{lm:insertable} is that we scale every instance by a parameter $\beta_k$ that depends on its size. Let us apply this transformation to the online average factorization algorithm in Lemma \ref{lm:net_universe}. Since the online average factorization algorithm outputs a factorization where the rows of $R_t$ and the columns of $L_t$ are partitioned into two parts, and the transformation in Lemma~\ref{lm:insertable} only concatenates and scales rows and columns of the factorizations output by its instances, we can also think of the factorization that is computed by the transformation as being similarly partitioned. So, we can write the output of the resulting online average factorization algorithm with insertions as
 $$\begin{pmatrix}
        \hat{L}^{ins}_t & L^{ins}_t
    \end{pmatrix}\begin{pmatrix}
         \hat{R}^{ins}_t \\R^{ins}_t
    \end{pmatrix},$$
where every row of $\hat{R}^{ins}_t$ is a scaled copy of a row of the identity matrix, so has exactly one non-zero entry. We can now use this online factorization algorithm with the transformation in Lemma~\ref{lm:revocable}. We  slightly modify this transformation as well, as follows. If at time $t$, when some element $x$ is deleted from the universe, there is a row of $\hat{R}^{ins}_t$ whose only nonzero entry is in the column indexed by $x$, then we simply replace the corresponding entry of $\ell'_{t'}$ with $0$ for all future time steps $t' \ge t$. In this way, we only need to apply the transformation in Lemma~\ref{lm:revocable} to $L^{ins}_t$ and $R^{ins}_t$. Moreover we will not normalize the right matrix of the online factorization to have Frobenius norm at most $\sqrt{N}$. With these modifications,  we may assume that at step $t$ the online factorization algorithm outputs the left matrix
\(\begin{pmatrix}
       \hat{L}'_t & L'_t 
    \end{pmatrix}\) and the right matrix 
\(\begin{pmatrix}
\hat{R}_t \\R'_t
\end{pmatrix}\) where $\hat{L}'_t$ can be obtained from $\hat{L}^{ins}_t$ by replacing some entries with zeros, $\hat{R}_t$ is a submatrix of $\hat{R}^{ins}_t$ and has one nonzero entry in each row, and we have the bounds
    \begin{align*}
        \norm{L'_t}^2_{2\rightarrow\infty}\leq\norm{L^{ins}_t}^2_{2\rightarrow\infty}\log_2{4N};    &&    \norm{R'_t}^2_{F}\leq \norm{R^{ins}_t}^2_{F}\log_2{4N}.
    \end{align*}
    Finally, for any element $x$, we have
    $s^2_{R_t'}(x)\leq s^2_{R^{ins}_{t^{del}(x)}}(x)\log_2{4N}.$
This gives us our online average factorization algorithm for semi-dynamic inputs.

 
 For the rest of the proof we follow the approach in Lemma~\ref{lm:simple_reduction}. We use the same basic algorithm and reuse the notation of Lemma~\ref{lm:simple_reduction}, with some small modifications. Inside any zone $\zeta$, let $N_{\zeta, t}=\abs{U_{\zeta, t}^{ins}}$. In each zone $\zeta$, the factorizations output, respectively, by the insertion-only and the semi-dynamic online factorization algorithms are
 \[
 \begin{pmatrix}\hat{L}^{ins}_{\zeta,t}\end{pmatrix}
 \begin{pmatrix}
     \hat{R}^{ins}_{\zeta,t}\\
     R^{ins}_{\zeta,t}
 \end{pmatrix}
 = Q_t|_{U^{ins}_{\zeta,t}},
 \quad \quad \quad
 \begin{pmatrix}\hat{L}_{\zeta,t}\end{pmatrix}
 \begin{pmatrix}
     \hat{R}_{\zeta,t}\\
     R_{\zeta,t}
 \end{pmatrix}
 = Q_t|_{U_{\zeta,t}}.
 \]
 
  Applying the same analysis as in Lemma \ref{lm:insertable}, we have
    \begin{align*}
       \norm{R^{ins}_{\zeta,t}}^2_F&\lesssim N_{\zeta, t}\sum_{k=0}^{\floor{\log{N_{\zeta, t}}}}\frac{1}{(1+\floor{\log_2{N}})-k)^3\log{4N}}\\
        \norm{\hat{R}^{ins}_{\zeta,t}}^2_{1\rightarrow 2}&\lesssim \sum_{k=0}^{\floor{\log{N_{\zeta, t}}}}\frac{1}{(1+\floor{\log_2{N}})-k)^3}\\
         \norm{L^{ins}_{\zeta,t}}^2_{2\rightarrow \infty}
         &\lesssim \sum_{k=0}^{\floor{\log{N_{\zeta, t}}}}(1+\floor{\log_2{N}}-k)^32^{k(1-1/d)}\log^{-1+2/d}{N}\\
          \norm{\hat{L}^{ins}_{\zeta,t}}^2_{2\rightarrow \infty}&\lesssim \sum_{k=0}^{\floor{\log{N_{\zeta, t}}}}(1+\floor{\log_2{N}}-k)^32^{k(1-1/d)}\log^{2/d}{N}
    \end{align*}
    Above, the bound on $\hat{R}^{ins}_{\zeta,t}$  is in terms of the $1\to 2$ norm rather than the Frobenius norm because each row of this matrix has at most one nonzero entry.

    We modify the algorithm in Figure~\ref{fig:kick-out} so that $x$ is deleted from $U_{\zeta,t}$ if 
    \[s^2_{R'_{\zeta,t}}(x) > \frac{C_1}{\log_2{4N}}\sum_{k=0}^{\floor{\log_2{N}}}\frac{1}{(1+\floor{\log_2{N}})-k)^3},\]
    where $C_1$ is a big enough constant.
    Note that this condition for deleting elements of $U_{\zeta,t}$ does not depend on $\hat{R}^{ins}$. We can now bound
      \begin{align*}
         \norm{R'_{\zeta,t}}^2_{1\rightarrow2}&\le \max_x{s^2_{R^{ins}_{\zeta,t^{del}(x)}}(x)}\log_2{4N_{\zeta,t}}\\
        &\leq \frac{C_1\log_2{4N_{\zeta,t}}}{\log_2{4N}}\sum_{k=0}^{\floor{\log_2{N_{\zeta,t}}}}{(1+\floor{\log_2{N}}-k)^{-3}}
        \lesssim \sum_{k=0}^{\floor{\log_2{N_{\zeta,t}}}}{(1+\floor{\log_2{N}}-k)}^{-3}.
     \end{align*}
    Setting $C_1$ large enough and using the bound on $\|R^{ins}_t\|_F^2$ and Markov's inequality, we have that at most half the elements of $U^{ins}_{\zeta,t}$ will be ``kicked out'' and inserted into $U^{ins}_{\zeta+1,t}$. Thus, $N_{\zeta,t}\le 2^{-\zeta} N$, and we have
    \begin{align*}
        \norm{R'_t}_{1\rightarrow2}^2\leq \sum_{\zeta=0}^{\floor{\log_2{N}}}\norm{R'_{\zeta,t}}^2_{1\rightarrow2}&\lesssim
         \sum_{\zeta=0}^{\floor{\log_2{N}}}\sum_{k=0}^{\floor{\log_2 N}-\zeta }{(1+\floor{\log_2{N}}-k)}^{-3}\\
         &\le 
         \sum_{\zeta=0}^{\infty}\sum_{k'=\zeta}^{\infty }{(1+k')}^{-3}
         = \sum_{k' = 0}^\infty {(1+k')}^{-2} < 2
    \end{align*}
    where we used the substitution $k' = \floor{\log_2 N} - k$, and exchanged the order of summation in the last equality. We can bound $\norm{\hat{R}_t}_{1\to 2}$ by an absolute constant using an analogous argument.  

    We now proceed to bounding the value of the factorization. We have 
    \begin{align*}
    \norm{\hat{L}_t}_{1\to 2}^2 + \norm{L'_t}_{2\rightarrow\infty}^2
    &\leq \sum_{\zeta=0}^{\floor{\log_2 N}}\left[ \log_2{(4N)} \norm{L^{ins}_{\zeta,t}}^2_{2\rightarrow \infty}+\norm{\hat{L}^{ins}_{\zeta,t}}^2_{2\rightarrow \infty} \right]\\
    & \lesssim \sum_{\zeta=0}^{\floor{\log_2 N}}\sum_{k=0}^{\floor{\log_2{N}}-\zeta}(1+\floor{\log_2{N}}-k)^32^{k(1-1/d)}\log_2^{2/d}{(4N)}\\
    & \lesssim N^{1-1/d}\log_2^{2/d}{(4N)}\sum_{\zeta=0}^{\floor{\log_2 N}}\sum_{k'=\zeta}^{\floor{\log_2{N}}}(1+k')^3 2^{-k'(1-1/d)}\\
    &\le N^{1-1/d}\log^{2/d}{(4N)} \sum_{k' = 0}^\infty (1+k')^4 2^{-k'(1-1/d)}\\
    &\lesssim N^{1-1/d}\log^{2/d}{(4N)}.
    \end{align*}
    Above, the third line again uses the substitution $k' = \floor{\log_2 N} - k$, the fourth line follows by changing the order of summation, and the last line uses the assumption $d \ge 2$, under which $\sum_{k' = 0}^\infty (1+k')^4 2^{-k'(1-1/d)}< \infty$. 
    
    Setting $R_t = \frac{1}{C}\begin{pmatrix} \hat{R}_t \\ R'_t\end{pmatrix}$ and $L_t= C\begin{pmatrix}
        \hat{L}_t & L'_t
    \end{pmatrix}$
    for an appropriate absolute constant $C$ finishes the proof. 
\end{proof}

\end{document}